\documentclass[]{article}
 \usepackage{graphics}
\usepackage{amsmath}
\usepackage{amsfonts}
\usepackage{hyperref}
\usepackage{amsthm}
 \usepackage{amscd}
\usepackage{mathrsfs}

\newtheorem{theorem}{Theorem}[section]
\newtheorem{proposition}[theorem]{Proposition}

\newtheorem{corollary}[theorem]{Corollary}
\newtheorem{definition}[theorem]{Definition}
\hypersetup{
	colorlinks = true,
	linkcolor  = red,
	citecolor  = green,
	filecolor  = cyan,
	urlcolor   = magenta,}

\newcommand{\ds}{\displaystyle}

\numberwithin{equation}{section}

\title{Direct and inverse scattering problems for the first-order discrete system associated with the derivative NLS system}
\author{T. Aktosun and R. Ercan\\
Department of Mathematics\\
University of Texas at Arlington\\
Arlington, TX 76019-0408, USA}

\date{}

\begin{document}

\maketitle

\begin{abstract}
The direct and inverse scattering problems are analyzed
for a first-order discrete system
associated with the semi-discrete version of the derivative NLS system.
The Jost solutions, the scattering coefficients, the bound-state dependency and norming
constants are investigated and related to the corresponding quantities for two
particular discrete linear systems associated with the semi-discrete
version of the NLS system.
The bound-state data set with any multiplicities is described
in an elegant manner in terms of a pair of constant matrix triplets.
Several methods are presented to the solve the inverse problem.
One of these methods involves a discrete Marchenko system
using as input the scattering data set consisting
of the scattering coefficients and the bound-state information,
and this method is presented in a way generalizable
to other first-order systems both in the discrete and continuous cases
for which a Marchenko system is not yet available.
Finally, using the time-evolved scattering data set,
the inverse scattering transform is applied on the corresponding semi-discrete derivative NLS system, and in the reflectionless case
certain explicit solution formulas are presented
in closed form expressed in terms of the two matrix triplets.
\end{abstract}

\section{Introduction}
\label{sec:section1}

In this paper we are interested in analyzing the direct and inverse scattering problems for the  first-order discrete system
\begin{equation}\label{1.1}
\begin{bmatrix}
\alpha_n\\
\noalign{\medskip}
\beta_n
\end{bmatrix}=
\begin{bmatrix}
z & \left(z-\ds\frac{1}{z}\right)q_n\\
\noalign{\medskip}
z\, r_n &\ds\frac{1}{z}+ \left(z-\ds\frac{1}{z}\right)q_n\,r_n
\end{bmatrix}
\begin{bmatrix}
\alpha_{n+1}\\
\noalign{\medskip}
\beta_{n+1}
\end{bmatrix},\qquad n\in\mathbb{Z},
\end{equation}
where $z$ is the spectral parameter taking values on the unit circle $\mathbb{T}$ in the complex $z$-plane $\mathbb{C},$ $n$ is the discrete independent variable taking values in the set of integers $\mathbb{Z}$,  the complex-valued scalar quantities
$q_n$ and $r_n$ correspond the respective values evaluated at
$n$ for the potential pair $(q,r),$ and
$\begin{bmatrix}
\alpha_n\\
\beta_n
\end{bmatrix}$ corresponds to the value of the wavefunction at the spacial location
$n.$
We assume that $q_n$ and $r_n$ are rapidly decaying in the sense that they vanish faster than any negative powers of $|n|$ as $n\to\pm\infty$.  We also assume that
\begin{equation}\label{1.1a}
1-q_nr_n\ne 0,\quad 1+q_nr_{n+1}\ne 0, \qquad n\in\mathbb{Z}.
\end{equation}
The complex-valued quantities $\alpha_n$ and $\beta_n$ depend on
the spectral parameter $z,$ but in our notation we usually suppress
that $z$-dependence.

The system in \eqref{1.1} is used as a model for an infinite lattice where
the particle with an internal structure at the lattice point $n$ experiences local forces
from the potential values $q_n$ and $r_n.$ Since we assume
that $q_n$ and $r_n$ vanish sufficiently fast as $n\to\pm\infty$, a scattering scenario can be established for \eqref{1.1}.\par

The direct scattering problem for \eqref{1.1} is described as
the determination of the scattering data set consisting of the scattering coefficients and bound-state information when the potential pair $(q,r)$ is known. The inverse scattering problem for \eqref{1.1} consists of the recovery of the potential pair $(q,r)$ when the scattering data set is given. Since $q_n$ and $r_n$  vanish sufficiently fast as $n\to\pm\infty$, it follows from \eqref{1.1} that any solution to \eqref{1.1} has the asymptotic behavior
\begin{equation}
\label{1.2}
\begin{bmatrix}
\alpha_n\\
\noalign{\medskip}
\beta_n
\end{bmatrix}=
\begin{bmatrix}
a_{\pm}z^{-n}\left[1+o(1)\right]\\
\noalign{\medskip}
b_{\pm}z^{n}\left[1+o(1)\right]
\end{bmatrix}, \qquad n\to\pm\infty,
\end{equation}
for some constants $a_{\pm}$ and $b_{\pm}$
that may depend on $z$ but not on $n$. By choosing two of the four coefficients $a_{+}$, $a_{-}$, $b_{+}$, $b_{-}$ appearing in
\eqref{1.2} in a specific way, we obtain a particular solution to \eqref{1.1}.
Note that \eqref{1.1} has two linearly independent solutions, and its general solution
can be expressed as a linear combination of
any two linearly independent solutions.

The discrete system \eqref{1.1} is related to the integrable semi-discrete system
\begin{equation}\label{1.2a}
\begin{cases}
i\dot{q}_n+\ds\frac{q_{n+1}}{1-q_{n+1}r_{n+1}}-
\ds\frac{q_{n}}{1-q_{n}r_{n}}-\ds\frac{q_{n}}{1+q_{n}r_{n+1}}+
\ds\frac{q_{n-1}}{1+q_{n-1}r_{n}}=0,\\
\noalign{\medskip}
i\dot{r}_n-\ds\frac{r_{n+1}}{1+q_{n}r_{n+1}}+\ds\frac{r_{n}}{1+q_{n-1}r_{n}}
+\ds\frac{r_{n}}{1-q_{n}r_{n}}-\ds\frac{r_{n-1}}{1-q_{n-1}r_{n-1}}=0,
\end{cases}
\end{equation}
which is known as  the semi-discrete derivative NLS (nonlinear Schr\"odinger) system or the semi-discrete Kaup-Newell system \cite{tsuchida2002integrable,tsuchida,tsuchida2010new}. From the denominators in \eqref{1.2a} we see why we need the restriction \eqref{1.1a}. Note that an overdot in \eqref{1.2a} denotes the derivative with respect to
the independent variable $t,$ which is interpreted as the time variable
and is suppressed in \eqref{1.2a}.
In our analysis of \eqref{1.1}, without loss of generality
we can either assume that $q_n$ and
$r_n$ are independent of $t$ or they
contain $t$ as a parameter.

We analyze the direct and inverse scattering problems for \eqref{1.1}
 by using the connection to the two first-order discrete systems
\begin{equation}\label{1.2aa}
\begin{bmatrix}
\xi_n\\
\noalign{\medskip}
\eta_n
\end{bmatrix}=
\begin{bmatrix}
z &z\, u_n\\
\noalign{\medskip}
\ds\frac{1}{z}\, v_n &\ds\frac{1}{z}
\end{bmatrix}
\begin{bmatrix}
\xi_{n+1}\\
\noalign{\medskip}
\eta_{n+1}
\end{bmatrix},\qquad n\in\mathbb{Z},
\end{equation}
	\begin{equation}\label{1.2ab}
\begin{bmatrix}
\gamma_n\\
\noalign{\medskip}
\epsilon_n
\end{bmatrix}=
\begin{bmatrix}
z &z\,p_n\\
\noalign{\medskip}
\ds\frac{1}{z}\, s_n &\ds\frac{1}{z}
\end{bmatrix}
\begin{bmatrix}
\gamma_{n+1}\\
\noalign{\medskip}
\epsilon_{n+1}
\end{bmatrix},\qquad n\in\mathbb{Z},
\end{equation}
where $u_n$ and $v_n$ are the values for the potential pair $(u,v)$ and $p_n$ and $s_n$ are the values for $(p,s)$. By choosing $(u,v)$ and $(p,s)$  as in \eqref{x3.1}--\eqref{x3.4}, we relate
 the relevant quantities for \eqref{1.1}, \eqref{1.2aa}, \eqref{1.2ab} to each other. Such relevant quantities include the Jost solutions, the scattering coefficients, and the bound-state data sets for each of \eqref{1.1},  \eqref{1.2aa}, \eqref{1.2ab}.

 We remark that in the literature it is always assumed that the bound states for  \eqref{1.1}, \eqref{1.2aa}, \eqref{1.2ab} are simple. In our paper
 we do not make such an artificial
assumption because we easily and in an elegant way handle the bound states of any multiplicities, and this is done by using a pair of constant matrix triplets
describing the bound-state values of the spectral parameter $z$ and the corresponding norming constants.

The systems \eqref{1.2aa} and \eqref{1.2ab} are of importance also in their own, and they are known as the Ablowitz-Ladik systems or as the discrete AKNS systems. It is
 possible \cite{tsuchida2010new}
 to transform \eqref{1.2aa} and \eqref{1.2ab} into
	\begin{equation}\label{1.2an}
\begin{bmatrix}
\tilde{\xi}_{n+1}\\
\noalign{\medskip}
\tilde{\eta}_{n+1}
\end{bmatrix}=
\begin{bmatrix}
z &u_n\\
\noalign{\medskip}
 v_n &\ds\frac{1}{z}
\end{bmatrix}
\begin{bmatrix}
\tilde{\xi}_{n}\\
\noalign{\medskip}
\tilde{\eta}_{n}
\end{bmatrix},\qquad n\in\mathbb{Z},
\end{equation}
\begin{equation}\label{1.2am}
\begin{bmatrix}
\tilde{\gamma}_{n+1}\\
\noalign{\medskip}
\tilde{\epsilon}_{n+1}
\end{bmatrix}=
\begin{bmatrix}
z &p_n\\
\noalign{\medskip}
s_n &\ds\frac{1}{z}
\end{bmatrix}
\begin{bmatrix}
\tilde{\gamma}_{n}\\
\noalign{\medskip}
\tilde{\epsilon}_{n}
\end{bmatrix},\qquad n\in\mathbb{Z}.
\end{equation}
Note that \eqref{1.2aa} and \eqref{1.2an} also differ from each other by the fact that the appearances of the wavefunction values evaluated at $n$ and $n+1$ are switched. The same remark also applies to \eqref{1.2ab} and \eqref{1.2am}.

As already pointed out by Tsuchida \cite{tsuchida2010new}, the analysis of the direct and inverse scattering problems for an Ablowitz-Ladik system written in the form of \eqref{1.2an} and \eqref{1.2am} is unnecessarily complicated. For example, the analysis provided in \cite{ablowitzPrinari2003} for \eqref{1.2an} involves separating the scattering data into two parts containing even and odd integer powers of $z,$ respectively. This unnecessarily makes the analysis cumbersome. Furthermore, if we use \eqref{1.2an} with the roles of $n$ and $n+1$ switched compared to \eqref{1.2aa} and use the scattering coefficients from the right instead of the scattering coefficients from the left as input, then the analysis of the inverse scattering problem for \eqref{1.2an} by the Marchenko method becomes unnecessarily complicated.

The researchers who are mainly interested in nonlinear evolution equations use only the scattering coefficient from the right without referring to the scattering coefficients from the left. In this paper, we are careful in making a distinction between the right and left scattering data sets. The right and left transmission coefficients in a first-order discrete linear system are unequal unless the coefficient matrix in that system has determinant equal to $1.$ One can verify that the coefficient matrix in \eqref{1.1} has its determinant equal to $1,$ whereas the corresponding determinants for \eqref{1.2aa} and \eqref{1.2ab} are given by $1-u_nv_n$ and $1-p_ns_n$, respectively. Thus, the left and right transmission coefficients for each of  \eqref{1.2aa} and \eqref{1.2ab} are unequal.

 The scattering and inverse scattering problems for \eqref{1.1} have partially been analyzed by Tsuchida in \cite{tsuchida2010new}. Our own analysis is complementary to Tsuchida's work in the following sense. Tsuchida's main interest in \eqref{1.1} is confined to its relation to \eqref{1.2a}, and he only deals with the right scattering coefficients. Tsuchida exploits certain gauge transformations to relate \eqref{1.1} to two discrete Ablowitz-Ladik systems, and he
 assumes that the bound states are all simple.
 Tsuchida's expressions for the scattering coefficients not only involve the Jost solutions to the relevant linear system but also the Jost solutions to the corresponding adjoint system, whereas in our case the scattering coefficients are expressed in terms of the Jost solutions to the relevant linear system only. In our opinion the latter description of the scattering coefficients provides physical insight and intuition into the analysis of direct and inverse problems. Tsuchida
  formulates a Marchenko system
 given in (4.12c) and (4.12d) of \cite{tsuchida2010new},
 somehow similar to our own alternate
 Marchenko system \eqref{6.22d} and \eqref{6.23}, but it lacks the
 appropriate symmetries existing in our alternate
 Marchenko system.
  In formulating his Marchenko system Tsuchida uses a Fourier transformation with respect to $z^2$ and not with respect to $z$. Furthermore, in Tsuchida's formulation it is not quite clear how the scattering data sets for \eqref{1.1}, \eqref{1.2aa}, \eqref{1.2ab} are related to each other.

%  In our formulation we carefully relate all the relevant quantities
%for \eqref{1.1}, \eqref{1.2aa}, \eqref{1.2ab} to each other.
% Such quantities include the Jost solutions, the
% scattering coefficients, and the bound-state information.
%We make a distinction between the left and right scattering coefficients. We do not assume the simplicity of bound states and we handle any number of bound states of any multiplicities in an elegant manner. Our formulation is systematic and our method can be applied on other similar problems.
One of the important accomplishments of our paper is the introduction of a standard Marchenko formalism for
\eqref{1.1} using as input the scattering data from
\eqref{1.1} only. The formulation
of our standard Marchenko system \eqref{Z.0} is a significant
generalization step to solve inverse problems
for various other discrete and continuous systems for which
a standard Marchenko theory has not yet been formulated.
As mentioned already, we also introduce
an alternate Marchenko formalism for \eqref{1.1}
using as input the scattering data sets
 from \eqref{1.2aa} and \eqref{1.2ab}.
Both our standard and alternate Marchenko systems we introduce have the appropriate symmetry properties and resemble the standard Marchenko systems arising in other continuous and
discrete systems.
The alternate Marchenko method in our paper
corresponds to the discrete analog of the systematic approach \cite{AE19} we presented to solve the inverse scattering problem for the energy-dependent AKNS system given in (1.1) of \cite{AE19}. Besides \cite{AE19} the most relevant reference  for our current work is the important paper by Tsuchida \cite{tsuchida2010new}.

Our paper is organized as follows. In Section~\ref{sec:section2} we introduce the Jost solutions and the scattering coefficients for each of \eqref{1.1}, \eqref{1.2aa}, \eqref{1.2ab} and we present some relevant properties of
those Jost solutions and scattering coefficients. In that section we also prove that the linear dependence of the appropriate pairs of Jost solutions occurs
at the poles of the corresponding transmission coefficients for each of
\eqref{1.1}, \eqref{1.2aa}, \eqref{1.2ab}.
In Section~\ref{sec:section3} when the corresponding potential pairs
are related to each other as in
\eqref{x3.1}--\eqref{x3.4}, we relate the
Jost solutions and scattering coefficients
for \eqref{1.1} to those for \eqref{1.2aa} and \eqref{1.2ab}.
In that section we also present certain relevant properties of
the Jost solutions to \eqref{1.1} and
express the potentials $q_n$ and $r_n$ in terms of the values at $z=1$ of the Jost solutions to \eqref{1.2aa} and \eqref{1.2ab}.
In Section~\ref{sec:section4} we describe the bound-state data sets
for each of \eqref{1.1}, \eqref{1.2aa}, \eqref{1.2ab} in terms of two
matrix triplets, which allows us to handle bound states of any multiplicities
in a systematic manner that can also be used for other systems both in the continuous and discrete cases. In the formulation of the Marchenko method
we show how the Marchenko kernels contain the matrix triplets
in a simple and elegant manner. Also in that section, when
the potential pairs for \eqref{1.1}, \eqref{1.2aa}, \eqref{1.2ab} are related as in
\eqref{x3.1}--\eqref{x3.4}, we show how the corresponding bound-state
data sets are related to each other.
In Section~\ref{sec:section5} we outline the steps to solve the direct
problem for \eqref{1.1}.
In Section~\ref{sec:section6} we introduce the Marchenko system
\eqref{Z.0} using as input the scattering data directly related to
\eqref{1.1} and we describe how the potentials
$q_n$ and $r_n$ are recovered from the solution
\eqref{Z.0}.
In Section~\ref{sec:section7} we present
our alternate Marchenko system given in \eqref{6.22d} and \eqref{6.23}
using as input the scattering data sets from
\eqref{1.2aa} and \eqref{1.2ab}, as we also show how
$q_n$ and $r_n$ are recovered from the solution to the
alternate Marchenko system.
In Section~\ref{sec:section8} we describe various
methods to solve the inverse problem for
\eqref{1.1} by using as input the scattering data for
\eqref{1.1} and outline how the
potentials $q_n$ and $r_n$ are recovered.
Finally, in Section~\ref{sec:section9}
we provide the solution to the integrable
nonlinear system \eqref{1.2a} via the inverse scattering transform.
This is done by providing the time evolution of the scattering data
for \eqref{1.1} and by determining the corresponding time-evolved
potentials $q_n$ and $r_n.$
In that section we also present some explicit solution formulas for
\eqref{1.2a} corresponding to time-evolved reflectionless
scattering data for
\eqref{1.1}, and such solutions are explicitly expressed in terms of
the two matrix triplets describing the time-evolved bound-state
data for \eqref{1.1}.

\section{The Jost solutions and scattering coefficients}
\label{sec:section2}

In this section we introduce the Jost solutions and the scattering coefficients for each of the linear systems given in \eqref{1.1}, \eqref{1.2aa}, \eqref{1.2ab}, and we present some of their relevant properties. For clarification, we use the superscript $(q,r)$ to denote the quantities relevant to \eqref{1.1}, use $(u,v)$ for those relevant to \eqref{1.2aa}, and use $(p,s)$ for those relevant to \eqref{1.2ab}. When these three potential pairs decay rapidly in their respective equations as $n\to\pm\infty$, the corresponding coefficient matrices all reduce to the same unperturbed coefficient matrix. In other words, each of \eqref{1.1}, \eqref{1.2aa}, \eqref{1.2ab} corresponds to the same unperturbed system
	\begin{equation*}
\label{x2.1}
\mathring{\Psi}_n=
\begin{bmatrix}
z &0\\
\noalign{\medskip}
0&\ds\frac{1}{z}
\end{bmatrix}
\mathring{\Psi}_{n+1},\qquad n\in\mathbb{Z},
\end{equation*}
where the general solution is a linear combination of the two linearly independent solutions  $\begin{bmatrix}
z^{-n}\\0
\end{bmatrix}$ and $\begin{bmatrix}
0\\z^{n}
\end{bmatrix}$, i.e. we have
\begin{equation}\label{x2.2}
\mathring{\Psi}_n=a\begin{bmatrix}
z^{-n}\\
\noalign{\medskip}
0
\end{bmatrix}+b\begin{bmatrix}
0\\
\noalign{\medskip}
z^{n}
\end{bmatrix},\qquad n\in\mathbb{Z},
\end{equation}
with $a$ and $b$ being two complex-valued scalars that are independent of $n$ but may depend on $z.$

There are four Jost solutions for each of \eqref{1.1}, \eqref{1.2aa}, \eqref{1.2ab}, and they are obtained by assigning specific values to $a$ and $b$ as $n\to+\infty$ or $n\to-\infty$. We uniquely define the four Jost solutions $\psi_n,$ $\phi_n,$ $\bar{ \psi}_n,$ $\bar{ \phi}_n$ to each of \eqref{1.1}, \eqref{1.2aa}, \eqref{1.2ab} so that they satisfy the  respective asymptotics
\begin{equation}\label{x2.3}
\psi_n=\begin{bmatrix}
o(1)\\
\noalign{\medskip}
z^n[1+o(1)]
\end{bmatrix} ,\qquad  n\to+\infty,
\end{equation}
\begin{equation}\label{x2.4}
\phi_n=\begin{bmatrix}
z^{-n}[1+o(1)]\\
\noalign{\medskip}
o(1)
\end{bmatrix} ,\qquad   n\to-\infty,
\end{equation}
\begin{equation}\label{x2.5}
\bar{\psi}_n=\begin{bmatrix}
z^{-n}[1+o(1)]\\
\noalign{\medskip}
o(1)
\end{bmatrix} ,\qquad  n\to+\infty,
\end{equation}
\begin{equation}\label{x2.6}
\bar{\phi}_n=\begin{bmatrix}
o(1)\\
\noalign{\medskip}
z^{n}[1+o(1)]
\end{bmatrix} ,\qquad  n\to-\infty.
\end{equation}
We remark that an overbar does not denote complex conjugation. We will use the notation $\psi_n^{(q,r)},$ $\phi_n^{(q,r)},$ $\bar{ \psi}_n^{(q,r)},$ $\bar{ \phi}_n^{(q,r)}$ to refer to the respective Jost solutions for \eqref{1.1}; use $\psi_n^{(u,v)},$ $\phi_n^{(u,v)},$ $\bar{ \psi}_n^{(u,v)},$ $\bar{ \phi}_n^{(u,v)}$ for the respective Jost solutions for \eqref{1.2aa}; and use $\psi_n^{(p,s)},$ $\phi_n^{(p,s)},$ $\bar{ \psi}_n^{(p,s)},$ $\bar{ \phi}_n^{(p,s)}$ for the respective Jost solutions for \eqref{1.2ab}.

The asymptotics of the Jost solutions complementary to \eqref{x2.3}--\eqref{x2.6} are used to define  the corresponding scattering coefficients compatible with \eqref{x2.2}. We have
\begin{equation}\label{x2.7}
\psi_n=\begin{bmatrix} \displaystyle\frac{L}{T_{\rm l}}\,z^{-n}\left[1+o(1)\right]\\
\noalign{\medskip}
\displaystyle\frac{1}{T_{\rm l}}\,z^{n}\left[1+o(1)\right]
\end{bmatrix}, \qquad   n\to-\infty,
\end{equation}
\begin{equation}\label{x2.8}
\phi_n=\begin{bmatrix}
\displaystyle\frac{1}{T_{\rm r}}\,z^{-n}\left[1+o(1)\right]\\
\noalign{\medskip}
\displaystyle\frac{R}{T_{\rm r}}\,z^{n}\left[1+o(1)\right]
\end{bmatrix} , \qquad   n\to+\infty,
\end{equation}
\begin{equation}\label{x2.10}
\bar{\psi}_n=\begin{bmatrix}
\displaystyle\frac{1}{\bar{T}_{\rm l}}\,z^{-n}\left[1+o(1)\right]\\
\noalign{\medskip}
\displaystyle\frac{\bar{L}}{\bar{T}_{\rm l}}\,z^{n}\left[1+o(1)\right]
\end{bmatrix}, \qquad  n\to-\infty,
\end{equation}
\begin{equation}\label{x2.9}
\bar{\phi}_n=\begin{bmatrix}
\displaystyle\frac{\bar{R}}   {\bar{T}_{\rm r}}\,z^{-n}\left[1+o(1)\right]\\\noalign{\medskip}
\displaystyle\frac{1}{\bar{T}_{\rm r}}\,z^{n}\left[1+o(1)\right]
\end{bmatrix}, \qquad   n\to+\infty,
\end{equation}
where $T_{\rm l}$ and $\bar{T}_{\rm l}$ are the transmission coefficients from the left, $T_{\rm r}$ and $\bar{T}_{\rm r}$ are the transmission coefficients from the right, $R$ and $\bar{R}$ are the reflection coefficients from the right, and $L$ and $\bar{L}$ are the reflection coefficients from the left. We will also say left scattering coefficients instead of scattering coefficients from the left, and similarly we will use right scattering coefficients and scattering coefficients from the right interchangeably.

Note that we will use $T_{\rm r}^{(q,r)},$ $T_{\rm l}^{(q,r)},$ $R^{(q,r)},$ $L^{(q,r)},$ $\bar{T}_{\rm r}^{(q,r)},$ $\bar{T}_{\rm l}^{(q,r)},$  $\bar{R}^{(q,r)},$ $\bar{L}^{(q,r)}$ to refer to the scattering coefficients for \eqref{1.1}; use $T_{\rm r}^{(u,v)},$ $T_{\rm l}^{(u,v)},$ $R^{(u,v)},$ $L^{(u,v)},$ $\bar{T}_{\rm r}^{(u,v)},$ $\bar{T}_{\rm l}^{(u,v)},$  $\bar{R}^{(u,v)},$ $\bar{L}^{(u,v)}$ for the scattering coefficients for \eqref{1.2aa}; and use $T_{\rm r}^{(p,s)},$ $T_{\rm l}^{(p,s)},$ $R^{(p,s)},$ $L^{(p,s)},$ $\bar{T}_{\rm r}^{(p,s)},$ $\bar{T}_{\rm l}^{(p,s)},$  $\bar{R}^{(p,s)},$ $\bar{L}^{(p,s)}$ for the scattering coefficients for \eqref{1.2ab}.

 Related to the linear system \eqref{1.2aa}, let us introduce the quantities $D_n^{(u,v)}$ and $D_\infty^{(u,v)}$ as
\begin{equation}\label{x2.D_n}
D_n^{(u,v)}:=\ds\prod_{j=-\infty}^{n}(1-u_j\,v_j),\quad D_\infty^{(u,v)}:=\ds\prod_{j=-\infty}^{\infty}(1-u_j\,v_j).
\end{equation}
 From the fact that $u_n$ and $v_n$ are rapidly decaying and that $1-u_nv_n\ne 0$ for $n\in\mathbb{Z},$ it follows that $D_n^{(u,v)}$ and $D_\infty^{(u,v)}$ are each well defined and nonzero.
Similarly, related to the linear system \eqref{1.2ab}, we let
\begin{equation}\label{x2.D_na}
D_n^{(p,s)}:=\ds\prod_{j=-\infty}^{n}(1-p_j\,s_j),\quad D_\infty^{(p,s)}:=\ds\prod_{j=-\infty}^{\infty}(1-p_j\,s_j).
\end{equation}
 From the fact that $p_n$ and $s_n$ are  decaying rapidly and that $1-p_ns_n\ne 0$ for $n\in\mathbb{Z},$ we see that $D_n^{(p,s)}$ and $D_\infty^{(p,s)}$ are each well defined and nonzero.

 In the next theorem we list some relevant analyticity properties of the Jost solutions to \eqref{1.2aa}.

 \begin{theorem}
	\label{thm:theorem x2.1}
	Assume that the potentials $u_n$ and $v_n$ appearing in \eqref{1.2aa} are rapidly decaying and $1-u_nv_n\ne 0$ for $n\in\mathbb{Z}$. Then, the corresponding Jost solutions to \eqref{1.2aa} satisfy following:
	
	\begin{enumerate}
		\item[\text{\rm(a)}] For each $n\in\mathbb{Z}$ the quantities $z^{-n}\,\psi_n^{(u,v)},$ $z^{n}\,\phi_n^{(u,v)},$ $z^{n}\,\bar{\psi}_n^{(u,v)},$  $z^{-n}\,\bar{\phi}_n^{(u,v)}$ are even in $z$ in their respective domains.
		
		\item[\text{\rm(b)}] The quantity $z^{-n}\,\psi_n^{(u,v)}$  is analytic in $|z|<1$ and
		continuous in $|z|\le 1$.
		
		\item[\text{\rm(c)}] The quantity $z^{n}\,\phi_n^{(u,v)}$ is analytic in $|z|<1$ and
		continuous in $|z|\le 1$.
		
		\item[\text{\rm(d)}] The quantity $z^{n}\,\bar{\psi}_n^{(u,v)}$ is analytic in $|z|>1$ and
		continuous in $|z|\ge 1$.
		
		\item[\text{\rm(e)}] The quantity $z^{-n}\,\bar{\phi}_n^{(u,v)}$ is analytic in $|z|>1$ and
		continuous in $|z|\ge 1$.
		
		\item[\text{\rm(f)}] The Jost solution $\psi_{n}^{(u,v)}$ has the expansion
		\begin{equation}\label{x2.11}
	\psi_{n}^{(u,v)}=\sum_{l=n}^{\infty}K_{nl}^{(u,v)}z^l, \qquad |z|\le 1,
		\end{equation}
		with the double-indexed quantities $K_{nl}^{(u,v)}$ for which we have
		\begin{equation}\label{x2.12}
			K_{nn}^{(u,v)}= \begin{bmatrix}
		0\\
		\noalign{\medskip}
		1
		\end{bmatrix},\quad
		K_{n(n+2)}^{(u,v)}= \begin{bmatrix}
		u_n\\
		\noalign{\medskip}
	\ds	\sum_{k=n}^{\infty} u_{k+1}\,v_k
		\end{bmatrix},
		\end{equation}
and that $K_{nl}^{(u,v)}=0 $ when $n+l$ is odd or  $l<n$.

		\item[\text{\rm(g)}] The Jost solution $\bar{\psi}_{n}^{(u,v)}$ has the expansion	\begin{equation}\label{x2.13}
	\bar{\psi}_{n}^{(u,v)}=\sum_{l=n}^{\infty}\bar{K}_{nl}^{(u,v)}\ds\frac{1}{z^l}, \qquad |z|\ge 1,
		\end{equation}
		with the double-indexed quantities $\bar{K}_{nl}^{(u,v)}$ for which we have
		\begin{equation}\label{x2.14}
	\bar{K}_{nn}^{(u,v)}= \begin{bmatrix}
		1\\
		\noalign{\medskip}
		0
		\end{bmatrix},\quad
			\bar{K}_{n(n+2)}^{(u,v)}= \begin{bmatrix}
		\ds\sum_{k=n}^{\infty} u_k\,v_{k+1}\\
		\noalign{\medskip}
		v_n
		\end{bmatrix},
		\end{equation}
and that $\bar{K}_{nl}^{(u,v)}=0 $ when $n+l$ is odd or  $l<n$.

\item[\text{\rm(h)}] For the Jost solution $\phi_{n}^{(u,v)}$ we have the expansion
\begin{equation}\label{x2.16}
z^{n}\,\phi_{n}^{(u,v)}=\sum_{l=0}^{\infty}P_{nl}^{(u,v)}\ds z^{l}, \qquad |z|\le 1,
\end{equation}
 with the double-indexed quantities $P_{nl}^{(u,v)}$ for which we have
\begin{equation}\label{x2.17}
P_{n0}^{(u,v)}=\ds\frac{1}{D_{n-1}^{(u,v)}}
\begin{bmatrix}
1\\
\noalign{\medskip}
-v_{n-1}
\end{bmatrix},
\end{equation}
\begin{equation*}
%\label{x2.18}
P_{n2}^{(u,v)}=\ds\frac{1}{D_{n-1}^{(u,v)}}
\begin{bmatrix}
\ds\sum_{k=-\infty}^{n-2} u_{k+1}\,v_{k}\\
\noalign{\medskip}
-v_{n-2}-v_{n-1}\ds\sum_{k=-\infty}^{n-3} u_{k+1}\,v_{k}
\end{bmatrix},
\end{equation*}
with $D_{n-1}^{(u,v)}$ being the quantity defined in \eqref{x2.D_n} and that $P_{nl}^{(u,v)}=0 $ when $l$ is odd or  $l<0$.
\item[\text{\rm(i)}] For the Jost solution $\bar{\phi}_{n}^{(u,v)}$ we have the expansion
\begin{equation*}
%\label{x2.19}
z^{-n}\,\bar{\phi}_{n}^{(u,v)}=\sum_{l=0}^{\infty}\bar{P}_{nl}^{(u,v)}\ds\frac{1}{z^{l}}, \qquad |z|\ge 1,
\end{equation*}
with the double-indexed quantities $\bar{P}_{nl}^{(u,v)}$ for which we have
\begin{equation*}
%\label{x2.20}
\bar{P}_{n0}^{(u,v)}=\ds\frac{1}{D_{n-1}^{(u,v)}}
\begin{bmatrix}
-u_{n-1}\\
\noalign{\medskip}
1
\end{bmatrix},
\end{equation*}
\begin{equation*}
%\label{x2.21}
\bar{P}_{n2}^{(u,v)}=\ds\frac{1}{D_{n-1}^{(u,v)}}
\begin{bmatrix}
-u_{n-2}-u_{n-1}\ds\sum_{k=-\infty}^{n-3} u_{k}\,v_{k+1}\\
\noalign{\medskip}
\ds\sum_{k=-\infty}^{n-2} u_{k}\,v_{k+1}
\end{bmatrix},
\end{equation*}
and that $\bar{P}_{nl}^{(u,v)}=0 $ when $l$ is odd or  $l<0$.

\item[\text{\rm(j)}] The scattering coefficients for \eqref{1.2aa} are even in $z$ in their respective domains. The domain for the reflection coefficients
    is the unit circle $\mathbb{T}$ and the domains for the transmission coefficients consist of the union of $\mathbb{T}$ and their regions of extensions.

\item[\text{\rm(k)}] The quantities $1/T_{\rm l}^{(u,v)}$ and $1/T_{\rm r}^{(u,v)}$ have analytic extensions in $z$ from $z\in\mathbb{T}$ to $|z|<1$ and those extensions are continuous for $|z|\le 1.$ Similarly, the quantities $1/\bar{T}_{\rm l}^{(u,v)}$ and $1/\bar{T}_{\rm r}^{(u,v)}$ have extensions from $z\in\mathbb{T}$ so that they are analytic in $|z|>1$ and continuous in $|z|\ge 1.$
	\end{enumerate}
\end{theorem}

\begin{proof}
	We can write \eqref{1.2aa} for $\psi_n^{(u,v)}$ in the equivalent form
	\begin{equation}\label{x.401}
	z^{-n}\,\psi_n^{(u,v)}=\begin{bmatrix}
	z^{2} &z^{2}\, u_n\\
	\noalign{\medskip}
	v_n&1
	\end{bmatrix}z^{-n-1}\,\psi_{n+1}^{(u,v)}, \qquad n\in \mathbb{Z}.
	\end{equation}
 From \eqref{x2.3} and the iteration of \eqref{x.401} in $n,$ it follows that $z^{-n}\,\psi_n^{(u,v)}$ is an even function of $z.$ By proceeding in a similar manner for the remaining Jost solutions, we complete the proof of (a). The expansion of $z^{-n}\,\psi_n^{(u,v)}$ obtained in (a) contains only nonnegative integer powers of $z^{2}$ and is uniformly convergent in $z$ for $|z|\le1,$ from which we conclude (b) and (f). The proofs for (c), (d), (e), (g), (h), (i) are obtained in a similar manner. Using (a)--(e) in \eqref{x2.7}--\eqref{x2.9} we establish (j). Finally, using (b) and the second component of the column-vector asymptotic in \eqref{x2.7}, we establish (k) for $1/T_{\rm l}^{(u,v)}.$ The remaining proofs for (k) are obtained in a similar manner.
\end{proof}

We remark that the results in Theorem~\ref{thm:theorem x2.1} holds also for \eqref{1.2ab}. For the convenience of citing those results,
we present the
following corollary.

\begin{corollary}
	\label{thm:theorem x2.1a}
	Assume that the potentials $p_n$ and $s_n$ appearing in \eqref{1.2ab} are rapidly decaying and $1-p_n s_n\ne 0$ for $n\in\mathbb{Z}$. Then, the corresponding Jost solutions to \eqref{1.2ab} satisfy all the properties stated in
Theorem~\ref{thm:theorem x2.1}. In particular we have the following:
	
	\begin{enumerate}

\item[\text{\rm(a)}] The Jost solution $\psi_{n}^{(p,s)}$ has the expansion
		\begin{equation}\label{x2.11aa}
	\psi_{n}^{(p,s)}=\sum_{l=n}^{\infty}K_{nl}^{(p,s)}z^l, \qquad |z|\le 1,
		\end{equation}
		with the double-indexed quantities $K_{nl}^{(p,s)}$ for which we have
		\begin{equation}\label{x2.12aa}
			K_{nn}^{(p,s)}= \begin{bmatrix}
		0\\
		\noalign{\medskip}
		1
		\end{bmatrix},\quad
		K_{n(n+2)}^{(p,s)}= \begin{bmatrix}
		p_n\\
		\noalign{\medskip}
	\ds	\sum_{k=n}^{\infty} p_{k+1}\,s_k
		\end{bmatrix},
		\end{equation}
and that $K_{nl}^{(p,s)}=0 $ when $n+l$ is odd or  $l<n$.

		\item[\text{\rm(b)}] The Jost solution $\bar{\psi}_{n}^{(p,s)}$ has the expansion	\begin{equation}\label{x2.13aa}
	\bar{\psi}_{n}^{(p,s)}=\sum_{l=n}^{\infty}\bar{K}_{nl}^{(p,s)}\ds\frac{1}{z^l}, \qquad |z|\ge 1,
		\end{equation}
		with the double-indexed quantities $\bar{K}_{nl}^{(p,s)}$ for which we have
		\begin{equation}\label{x2.14aa}
	\bar{K}_{nn}^{(p,s)}= \begin{bmatrix}
		1\\
		\noalign{\medskip}
		0
		\end{bmatrix},\quad
			\bar{K}_{n(n+2)}^{(p,s)}= \begin{bmatrix}
		\ds\sum_{k=n}^{\infty} p_k\,s_{k+1}\\
		\noalign{\medskip}
		s_n
		\end{bmatrix},
		\end{equation}
and that $\bar{K}_{nl}^{(p,s)}=0 $ when $n+l$ is odd or  $l<n$.

\end{enumerate}

\end{corollary}

In the next theorem we summarize the relevant properties of the scattering coefficients for \eqref{1.2aa}.

\begin{theorem}
\label{thm:theorem x2.2}
	Assume that the potentials $u_n$ and $v_n$ appearing in \eqref{1.2aa} are rapidly decaying and that $1-u_nv_n\ne 0$ for $n\in\mathbb{Z}$. The corresponding scattering coefficients in their respective domains satisfy
\begin{equation}\label{x2.26}
T_{\rm r}^{(u,v)}=D_\infty^{(u,v)}\,T_{\rm l}^{(u,v)}, \quad \bar{T}_{\rm r}^{(u,v)}=D_\infty^{(u,v)}\,\bar{T}_{\rm l}^{(u,v)},
\end{equation}
\begin{equation}\label{x2.27}
T_{\rm r}^{(u,v)}\,\bar{T}_{\rm r}^{(u,v)}=D_\infty^{(u,v)}\big[1-R^{(u,v)}\,\bar{R}^{(u,v)}\big],
\end{equation}
\begin{equation}\label{x2.28}
T_{\rm l}^{(u,v)}\,\bar{T}_{\rm l}^{(u,v)}=D_\infty^{(u,v)}\big[1-L^{(u,v)}\,\bar{L}^{(u,v)}\big],
\end{equation}
\begin{equation}\label{x2.29}
\frac{L^{(u,v)}}{T_{\rm l}^{(u,v)}}=-D_\infty^{(u,v)}\,\frac{\bar{R}^{(u,v)}}{\bar{T}_{\rm r}^{(u,v)}},
\end{equation}
\begin{equation}\label{x2.30}
\frac{\bar{L}^{(u,v)}}{\bar{T}_{\rm l}^{(u,v)}}=-D_\infty^{(u,v)}\,\frac{R^{(u,v)}}{T_{\rm r}^{(u,v)}},
\end{equation}
where  $D_\infty^{(u,v)}$ is the quantity defined in \eqref{x2.D_n}.
\end{theorem}

\begin{proof}
From \eqref{1.2aa} we get the matrix equations
\begin{equation}\label{x.4}
\begin{bmatrix}
\phi_n^{(u,v)}&\psi_n^{(u,v)}
\end{bmatrix}=
\begin{bmatrix}
z &z\, u_n\\
\noalign{\medskip}
\ds\frac{1}{z}\, v_n &\ds\frac{1}{z}
\end{bmatrix}
\begin{bmatrix}
\phi_{n+1}^{(u,v)}&\psi_{n+1}^{(u,v)}
\end{bmatrix},\qquad n\in\mathbb{Z},
\end{equation}
\begin{equation}\label{x.5}
\begin{bmatrix}
\bar{\psi}_n^{(u,v)}&\bar{\phi}_n^{(u,v)}
\end{bmatrix}=
\begin{bmatrix}
z &z\,u_n\\
\noalign{\medskip}
\ds\frac{1}{z}\, v_n &\ds\frac{1}{z}
\end{bmatrix}
\begin{bmatrix}
\bar{\psi}_{n+1}^{(u,v)}&\bar{\phi}_{n+1}^{(u,v)}
\end{bmatrix},\qquad n\in\mathbb{Z}.
\end{equation}
Using iteration on the determinants of both sides of \eqref{x.4} and \eqref{x.5}, respectively, for any pair of integers $n$ and $m$ with $m>n$  we get
\begin{equation}\label{x.6}
\det \begin{bmatrix}
\phi_n^{(u,v)}&\psi_n^{(u,v)}
\end{bmatrix}=(1-u_nv_n)\cdots (1-u_mv_m)\,\det\begin{bmatrix}
\phi_{m+1}^{(u,v)}&\psi_{m+1}^{(u,v)}
\end{bmatrix},
\end{equation}
\begin{equation}\label{x.7}
\det \begin{bmatrix}
\bar{\psi}_n^{(u,v)}&\bar{\phi}_n^{(u,v)}
\end{bmatrix}=(1-u_nv_n)\cdots (1-u_mv_m)\,\det\begin{bmatrix}
\bar{\psi}_{m+1}^{(u,v)}&\bar{\phi}_{m+1}^{(u,v)}
\end{bmatrix}.
\end{equation}
Letting $n\to -\infty$ and $m\to+\infty$ in \eqref{x.6},
with the help of \eqref{x2.3},
\eqref{x2.4},
\eqref{x2.7},
\eqref{x2.8}, and
\eqref{x2.D_n} we obtain
\begin{equation}
\label{ta2001}
\begin{vmatrix}
 1&\ds\frac{L^{(u,v)}}{T_{\rm l}^{(u,v)}}\\
\noalign{\medskip}
 0&\ds\frac{1}{T_{\rm l}^{(u,v)}}
\end{vmatrix}=D_\infty^{(u,v)}\begin{vmatrix}
\ds\frac{1}{T_{\rm r}^{(u,v)}}&0\\
\noalign{\medskip}
\ds\frac{R^{(u,v)}}{T_{\rm r}^{(u,v)}}&1
\end{vmatrix}.
\end{equation}
Similarly, letting $n\to -\infty$ and $m\to+\infty$ in \eqref{x.7}
and using \eqref{x2.5},
\eqref{x2.6},
\eqref{x2.10},
\eqref{x2.9}, and
\eqref{x2.D_n} we get
\begin{equation}
\label{ta2002}
\begin{vmatrix}
\ds\frac{1}{\bar{T}_{\rm l}^{(u,v)}}&0\\
\noalign{\medskip}
\ds\frac{\bar{L}^{(u,v)}}{\bar{T}_{\rm l}^{(u,v)}}&1
\end{vmatrix}=D_\infty^{(u,v)}\begin{vmatrix}
1&\ds\frac{\bar{R}^{(u,v)}}{\bar{T}_{\rm r}^{(u,v)}}\\
\noalign{\medskip}
0&\ds\frac{1}{\bar{T}_{\rm r}^{(u,v)}}
\end{vmatrix}.
\end{equation}
 From \eqref{ta2001} and \eqref{ta2002} we get \eqref{x2.26}.
 On the other hand, with the help of
Theorem~\ref{thm:theorem x2.1} we conclude that any two of the four
Jost solutions $\psi_n^{(u,v)},$ $\phi_n^{(u,v)},$ $\bar{ \psi}_n^{(u,v)},$
$\bar{ \phi}_n^{(u,v)}$ form a linearly independent set of solutions
to \eqref{1.2aa} when $z$ is on the unit circle $\mathbb{T}$.
We can express $\phi_n^{(u,v)}$ and $\bar{\phi}_n^{(u,v)}$ as
linear combinations of $\psi_n^{(u,v)}$ and $\bar{\psi}_n^{(u,v)}$ in a matrix form as
\begin{equation}\label{x.10}
\begin{bmatrix}
\phi_{n}^{(u,v)}\\
\noalign{\medskip}
\bar{ \phi}_n^{(u,v)}
\end{bmatrix}=\begin{bmatrix}
\ds\frac{1}{T_{\rm r}^{(u,v)}}&\ds\frac{R^{(u,v)}}{T_{\rm r}^{(u,v)}}\\
\noalign{\medskip}
\ds\frac{\bar{R}^{(u,v)}}{\bar{T}_{\rm r}^{(u,v)}}&\ds\frac{1}{\bar{T}_{\rm r}^{(u,v)}}
\end{bmatrix}\begin{bmatrix}
\bar{\psi}_{n}^{(u,v)}\\
\noalign{\medskip}
\psi_n^{(u,v)}
\end{bmatrix},\qquad z\in\mathbb{T},
\end{equation}
where the entries in the coefficient matrix are obtained with the help of
\eqref{x2.3}, \eqref{x2.5}, \eqref{x2.8}, \eqref{x2.9}
for the Jost solutions to \eqref{1.2aa}. In a
similar way, with the help of \eqref{x2.4}, \eqref{x2.6},
\eqref{x2.7}, \eqref{x2.10} for the Jost solutions to \eqref{1.2aa} we get
\begin{equation}\label{x.11}
\begin{bmatrix}
\bar{ \psi}_{n}^{(u,v)}\\
\noalign{\medskip}
\psi_n^{(u,v)}
\end{bmatrix}=\begin{bmatrix}
\ds\frac{1}{\bar{T}_{\rm l}^{(u,v)}}&\ds\frac{\bar{L}^{(u,v)}}{\bar{T}_{\rm l}^{(u,v)}}\\
\noalign{\medskip}
\ds\frac{L^{(u,v)}}{T_{\rm l}^{(u,v)}}&\ds\frac{1}{T_{\rm l}^{(u,v)}}
\end{bmatrix}\begin{bmatrix}
\phi_{n}^{(u,v)}\\
\noalign{\medskip}
\bar{\phi}_n^{(u,v)}
\end{bmatrix},\qquad z\in\mathbb{T}.
\end{equation}
For the compatibility of \eqref{x.10} and \eqref{x.11} we must have
\begin{equation}\label{x.12}
\begin{bmatrix}
\ds\frac{1}{T_{\rm r}^{(u,v)}}&\ds\frac{R^{(u,v)}}{T_{\rm r}^{(u,v)}}\\
\noalign{\medskip}
\ds\frac{\bar{R}^{(u,v)}}{\bar{T}_{\rm r}^{(u,v)}}&\ds\frac{1}{\bar{T}_{\rm r}^{(u,v)}}
\end{bmatrix}\begin{bmatrix}
\ds\frac{1}{\bar{T}_{\rm l}^{(u,v)}}&\ds\frac{\bar{L}^{(u,v)}}{\bar{T}_{\rm l}^{(u,v)}}\\
\noalign{\medskip}
\ds\frac{L^{(u,v)}}{T_{\rm l}^{(u,v)}}&\ds\frac{1}{T_{\rm l}^{(u,v)}}
\end{bmatrix}=\begin{bmatrix}
1&0\\
\noalign{\medskip}
0&1
\end{bmatrix},\qquad z\in\mathbb{T}.
\end{equation}
Then, using \eqref{x2.26} and \eqref{x.12} we obtain \eqref{x2.27}--\eqref{x2.30}.
\end{proof}
The above theorem indicates that the set of left
scattering coefficients $T_{\rm l}^{(u,v)},$
$\bar{T}_{\rm l}^{(u,v)},$ $L^{(u,v)},$
$\bar{L}^{(u,v)}$ can be expressed  in
terms of the set of right scattering
coefficients $T_{\rm r}^{(u,v)},
\bar{T}_{\rm r}^{(u,v)}, R^{(u,v)}, \bar{R}^{(u,v)},$ and vice versa.

In the next proposition we provide the asymptotics of the transmission coefficients for \eqref{1.2aa}.

\begin{proposition}
	\label{thm:theorem x2.3}
	Assume that the potentials $u_n$ and $v_n$ appearing in \eqref{1.2aa} are rapidly decaying and that $1-u_nv_n\ne 0$ for $n\in\mathbb{Z}$. Then, the transmission coefficients for \eqref{1.2aa} have their asymptotics given by
	\begin{equation}\label{x2.30a}
	T_{\rm l}^{(u,v)}=1-z^2\ds\sum_{k=-\infty}^{\infty} u_{k+1}\,v_{k}+O\left(z^4\right),\qquad z\to 0,
	\end{equation}
	\begin{equation}\label{x2.30b}
	T_{\rm r}^{(u,v)}=D_\infty^{(u,v)}\bigg[1-z^2\ds\sum_{k=-\infty}^{\infty} u_{k+1}\,v_{k}+O\left(z^4\right)\bigg],\qquad z\to 0,
	\end{equation}
		\begin{equation}\label{x2.30d}
	\bar{T}_{\rm l}^{(u,v)}=1-\frac{1}{z^2}\ds\sum_{k=-\infty}^{\infty} u_{k}\,v_{k+1}+O\left(\frac{1}{z^4}\right),\qquad z\to\infty,
	\end{equation}
	\begin{equation}\label{x2.30c}
	\bar{T}_{\rm r}^{(u,v)}=D_\infty^{(u,v)}\bigg[1-\frac{1}{z^2}\ds\sum_{k=-\infty}^{\infty} u_{k}\,v_{k+1}+O\left(\frac{1}{z^4}\right)\bigg],\qquad z\to\infty,
	\end{equation}
where  $D_\infty^{(u,v)}$ is the quantity defined in \eqref{x2.D_n}.
\end{proposition}

\begin{proof} From Theorem~\ref{thm:theorem x2.1}(k) we know that $1/T_{\rm l}^{(u,v)}$ and $1/T_{\rm r}^{(u,v)}$ are analytic in $|z|<1$ and that $1/\bar{T}_{\rm l}^{(u,v)}$ and $1/\bar{T}_{\rm r}^{(u,v)}$ are analytic in $|z|>1.$
Premultiplying both sides of \eqref{x2.11} by $z^{-n}[
0\quad 1
],$ then letting $n\to -\infty$	 in the resulting equation, and using \eqref{x2.7} with $\psi_{n}^{(u,v)}$ and \eqref{x2.12}  we obtain \eqref{x2.30a}. Similarly, premultiplying both sides of \eqref{x2.13} by $z^{n}[
1\quad 0],$ then letting $n\to -\infty$	 in the resulting equation, and using \eqref{x2.10} with $\bar{\psi}_{n}^{(u,v)}$  and \eqref{x2.14} we obtain \eqref{x2.30d}. Finally, with the help of \eqref{x2.26}, \eqref{x2.30a}, \eqref{x2.30d} we get \eqref{x2.30b} and \eqref{x2.30c}.
\end{proof}

In the next theorem we provide various other relevant properties of the
transmission coefficients for \eqref{1.2aa}.

\begin{theorem}
	\label{thm:theorem x2.3a}
	Assume that the potentials $u_n$ and $v_n$ appearing in \eqref{1.2aa} are rapidly decaying and that $1-u_nv_n\ne 0$ for $n\in\mathbb{Z}$. Then, for the transmission coefficients of \eqref{1.2aa} we have the following:

	\begin{enumerate}
		\item[\text{\rm(a)}] None of $T_{\rm l}^{(u,v)},$ $T_{\rm r}^{(u,v)},$ $\bar{T}_{\rm l}^{(u,v)},$ $\bar{T}_{\rm r}^{(u,v)}$ can vanish when $z\in\mathbb{T}.$
		
		\item[\text{\rm(b)}] We have
		\begin{equation}\label{x.501}
		\ds\frac{1}{T_{\rm l}^{(u,v)}(0)}=1,\quad \ds\frac{1}{T_{\rm r}^{(u,v)}(0)}=\ds\frac{1}{D_\infty^{(u,v)}}\ne 0,
		\end{equation}
		\begin{equation*}
%\label{x.502}
		\ds\frac{1}{\bar{T}_{\rm l}^{(u,v)}(\infty)}=1,\quad \ds\frac{1}{\bar{T}_{\rm r}^{(u,v)}(\infty)}=\ds\frac{1}{D_\infty^{(u,v)}}\ne 0.
		\end{equation*}
	
		\item[\text{\rm(c)}] The quantity $1/T_{\rm l}^{(u,v)}$ has at most a finite number of zeros in $0<|z|<1$ and the multiplicity of each such zero is finite. The zeros of $1/T_{\rm r}^{(u,v)}$ and the multiplicities of those zeros are the same as those for $1/T_{\rm l}^{(u,v)}.$
		
		\item[\text{\rm(d)}] The quantity $1/\bar{T}_{\rm l}^{(u,v)}$ has at most a finite number of zeros in $|z|>1$ and the multiplicity of each such zero is finite. The zeros of $1/\bar{T}_{\rm r}^{(u,v)}$ and the multiplicities of those zeros are the same as those for $1/\bar{T}_{\rm l}^{(u,v)}.$
		
		\item[\text{\rm(e)}] The quantities $T_{\rm l}^{(u,v)}$ and $T_{\rm r}^{(u,v)}$ are meromorphic in $|z|<1.$ The number of their poles and the multiplicities of those poles are both finite. Similarly, the quantities $\bar{T}_{\rm l}^{(u,v)}$ and $\bar{T}_{\rm r}^{(u,v)}$ are meromorphic in $|z|>1,$ and the number of their poles and the multiplicities of those poles are both finite.
		
		\item[\text{\rm(f)}] If $z_j$ is a pole of $T_{\rm l}^{(u,v)}$ and $T_{\rm r}^{(u,v)}$ in $0<|z|<1,$ then $-z_j$ is also a pole of those two transmission coefficients. Similarly, if $\bar{z}_j$ is a pole of $\bar{T}_{\rm l}^{(u,v)}$ and $\bar{T}_{\rm r}^{(u,v)}$ in $|z|>1$ then $-\bar{z}_j$ is also a pole  of $\bar{T}_{\rm l}^{(u,v)}$ and $\bar{T}_{\rm r}^{(u,v)}.$
		\end{enumerate}
\end{theorem}

	\begin{proof}
We can write \eqref{x2.28} as
	\begin{equation*}
%\label{x.503}
	\ds\frac{1}{D_\infty^{(u,v)}}=\ds\frac{1}{T_{\rm l}^{(u,v)}\,\bar{T}_{\rm l}^{(u,v)}}-\ds\frac{L^{(u,v)}\,\bar{L}^{(u,v)}}{T_{\rm l}^{(u,v)}\,\bar{T}_{\rm l}^{(u,v)}},\qquad z\in\mathbb{T}.
	\end{equation*}
Because of the continuity of $L^{(u,v)}/T_{\rm l}^{(u,v)},$ we would conclude that	if $T_{\rm l}^{(u,v)}$ vanished at some point on $\mathbb{T}$ then $L^{(u,v)}$ would have to vanish at that same point on $\mathbb{T}.$ However, this cannot happen because it would contradict \eqref{x2.28} as we have $D_\infty^{(u,v)}\ne 0.$ The remaining proofs in (a) are established in a similar manner. Note that (b) is obtained directly from \eqref{x2.30a}--\eqref{x2.30c}. The proof of (c) for $T_{\rm l}^{(u,v)}$ can be given as follows. From Theorem~\ref{thm:theorem x2.1}(k) we know that $1/T_{\rm l}^{(u,v)}$ is analytic in $|z|<1$ and continuous in $|z|\le 1$, and from (a) we know that $1/T_{\rm l}^{(u,v)}$ cannot vanish on $\mathbb{T}.$ Hence, any zeros of $1/T_{\rm l}^{(u,v)}$ must occur in the bounded region $|z|<1.$ Thus, the zeros of $1/T_{\rm l}^{(u,v)}$ in $|z|<1$ must be finite in number and each such zero must have a finite multiplicity. The remaining proof of (c) is obtained by using the first equality in \eqref{x2.26}. The proofs of (d) and (e) are obtained in a manner similar to the proof of (c). Finally, we note that (f) follows from (c), (d), (e), and Theorem~\ref{thm:theorem x2.1}(j).
	\end{proof}

We remark that the analogs of the results presented for the potential pair $(u,v)$ in Proposition~\ref{thm:theorem x2.3} and Theorems~\ref{thm:theorem x2.1}, \ref{thm:theorem x2.2}, and \ref{thm:theorem x2.3a} are also valid for the potential pair $(p,s)$.

Next, let us consider the properties of the scattering coefficients for \eqref{1.1}. Because the coefficient matrix in \eqref{1.1} has determinant equal to $1,$ in this case we can express the scattering coefficients for \eqref{1.1} in terms of the Wronskians each defined as a determinant of a 2$\times$2 matrix where the two columns are  the appropriate Jost solutions to \eqref{1.1}. We define the Wronskian of two column-vector solutions $\begin{bmatrix}
\alpha_n\\
\beta_n
\end{bmatrix} $ and $\begin{bmatrix}
\hat{\alpha}_n\\
\hat{\beta}_n
\end{bmatrix} $ to \eqref{1.1} as
\begin{equation}\label{x2.31}
\begin{bmatrix}
\begin{bmatrix}
\alpha_n\\
\noalign{\medskip}
\beta_n
\end{bmatrix};\begin{bmatrix}
\hat{\alpha}_n\\
\noalign{\medskip}
\hat{\beta}_n
\end{bmatrix}
\end{bmatrix}: =\det
\begin{bmatrix}
\alpha_n & \hat{\alpha}_n\\
\noalign{\medskip}
\beta_n &\hat{\beta}_n
\end{bmatrix}=\begin{vmatrix}
\alpha_n & \hat{\alpha}_n\\
\noalign{\medskip}
\beta_n &\hat{\beta}_n
\end{vmatrix}.
\end{equation}

We recall that the scattering coefficients for \eqref{1.2aa} cannot be obtained from the Wronskians of any two solutions to \eqref{1.2aa}  because the coefficient matrix in \eqref{1.2aa} does not have the determinant equal to $1.$ In that case, in order to obtain the scattering coefficients one can use the Wronskians of a solution to \eqref{1.2aa} and a solution to the adjoint equation corresponding to \eqref{1.2aa}. However, we prefer to express the scattering coefficients via the asymptotics of the Jost solutions as in \eqref{x2.7}--\eqref{x2.9} and this allows us to investigate the scattering coefficients in a unified way for any of the three systems \eqref{1.1}, \eqref{1.2aa}, \eqref{1.2ab}.

With the help of  \eqref{1.1} and \eqref{x2.31} one can directly verify that the determinant used in \eqref{x2.31} is independent of $n$. In terms of the Jost solutions $\psi_n^{(q,r)},$ $\phi_n^{(q,r)},$ $\bar{\psi}_n^{(q,r)},$  $\bar{\phi}_n^{(q,r)}$ satisfying \eqref{1.1} and the respective asymptotics given in \eqref{x2.3}--\eqref{x2.6}, with the help of \eqref{x2.7}--\eqref{x2.9} we express the scattering coefficients $T_{\rm l}^{(q,r)},$ $\bar{T}_{\rm l}^{(q,r)},$ $T_{\rm r}^{(q,r)},$ $\bar{T}_{\rm r}^{(q,r)},$ $R^{(q,r)},$ $\bar{R}^{(q,r)},$ $L^{(q,r)},$ $\bar{L}^{(q,r)}$ as
\begin{equation}\label{x2.32}
\frac{1}{T_{\rm l}^{(q,r)}}=\begin{vmatrix}
\phi_n^{(q,r)}&\psi_n^{(q,r)}
\end{vmatrix},\quad \frac{1}{\bar{T}_{\rm l}^{(q,r)}}=\begin{vmatrix}
\bar{\psi}_n^{(q,r)}&\bar{\phi}_n^{(q,r)}
\end{vmatrix},
\end{equation}
\begin{equation}\label{x2.33}
\frac{1}{T_{\rm r}^{(q,r)}}=\begin{vmatrix}
\phi_n^{(q,r)}&\psi_n^{(q,r)}
\end{vmatrix},\quad \frac{1}{\bar{T}_{\rm r}^{(q,r)}}=\begin{vmatrix}
\bar{\psi}_n^{(q,r)}&\bar{\phi}_n^{(q,r)}
\end{vmatrix},
\end{equation}
\begin{equation}\label{x2.34}
\frac{L^{(q,r)}}{T_{\rm l}^{(q,r)}}=\begin{vmatrix}
\psi_n^{(q,r)}&\bar{\phi}_n^{(q,r)}
\end{vmatrix},\quad \frac{\bar{L}^{(q,r)}}{\bar{T}_{\rm l}^{(q,r)}}=\begin{vmatrix}
\phi_n^{(q,r)}&\bar{\psi}_{n}^{(q,r)}
\end{vmatrix},
\end{equation}
\begin{equation}\label{x2.35}
\frac{R^{(q,r)}}{T_{\rm r}^{(q,r)}}=\begin{vmatrix}
\bar{\psi}_n^{(q,r)}&\phi_n^{(q,r)}
\end{vmatrix},\quad \frac{\bar{R}^{(q,r)}}{\bar{T}_{\rm r}^{(q,r)}}=\begin{vmatrix}
\bar{\phi}_n^{(q,r)}&\psi_{n}^{(q,r)}
\end{vmatrix}.
\end{equation}

In the next theorem we list some relevant properties of the scattering coefficients for \eqref{1.1}.

\begin{theorem}
		\label{thm:theorem x2.4}
	Assume that the potentials $q_n$ and $r_n$ appearing in \eqref{1.1} are rapidly decaying and satisfy \eqref{1.1a}. Then, we have the following:

	\begin{enumerate}
	
		\item[\text{\rm(a)}] The left and right transmission coefficients for \eqref{1.1} are equal to each other, i.e. we have
		\begin{equation}\label{x2.37}
		T_{\rm l}^{(q,r)}=T_{\rm r}^{(q,r)},\quad \bar{T}_{\rm l}^{(q,r)}=\bar{T}_{\rm r}^{(q,r)}.
		\end{equation}
		We will use $T^{(q,r)}$ to denote the common value of $T_{\rm l}^{(q,r)}$ and $T_{\rm r}^{(q,r)},$ and we will use $\bar{T}^{(q,r)}$ to denote the common value of $\bar{T}_{\rm l}^{(q,r)}$ and $\bar{T}_{\rm r}^{(q,r)}.$ The transmission coefficient $T^{(q,r)}$ has a meromorphic extension from $z\in\mathbb{T}$ to $|z|<1$ and the transmission coefficient $\bar{T}^{(q,r)}$ has a meromorphic extension from $z\in\mathbb{T}$ to $|z|>1$.
	
		\item[\text{\rm(b)}] For $z\in\mathbb{T},$ the left and right reflection coefficients for \eqref{1.1} satisfy
		\begin{equation}\label{x2.38}
		\frac{L^{(q,r)}}{T^{(q,r)}}=-\frac{\bar{R}^{(q,r)}}{\bar{T}^{(q,r)}}, \quad \frac{\bar{L}^{(q,r)}}{\bar{T}^{(q,r)}}=-\frac{R^{(q,r)}}{T^{(q,r)}},
		\end{equation}
		\begin{equation}\label{x2.39}	T^{(q,r)}\,\bar{T}^{(q,r)}=1-L^{(q,r)}\,\bar{L}^{(q,r)}=1-R^{(q,r)}\,\bar{R}^{(q,r)}.
		\end{equation}
	\end{enumerate}
\end{theorem}

\begin{proof}
The proof can be obtained as in the proof of Theorem~\ref{thm:theorem x2.2}. As an alternate proof, we see that \eqref{x2.37} follows from \eqref{x2.32} and \eqref{x2.33}; \eqref{x2.38} follows from \eqref{x2.34} and \eqref{x2.35}; and that \eqref{x2.39} is established by using the fact that the Wronskian of $\bar{ \psi}_n^{(q,r)}$ and $\psi_{n}^{(q,r)}$ as in \eqref{x2.31} yields the same value as $n \to -\infty$ and $n \to +\infty$. Finally, the aforementioned meromorphic extensions for the transmission coefficients follow from the fact that the Jost solutions $\psi_{n}^{(q,r)}$ and $\phi_{n}^{(q,r)}$ have analytic extensions in $z$ to $|z|<1$ and that the Jost solutions $\bar{\psi}_{n}^{(q,r)}$ and $\bar{\phi}_{n}^{(q,r)}$ have analytic extensions in $z$ to $|z|>1,$ where the analytic extensions can be established by iterating \eqref{1.1} and by using \eqref{x2.3}--\eqref{x2.6}.
\end{proof}

We see from \eqref{x2.37} that the left and right transmission coefficients for \eqref{1.1} are equal whereas the same does not hold for \eqref{1.2aa}. Similar to \eqref{x2.D_n} we define the related quantities $D_n^{(q,r)}$ and $D_\infty^{(q,r)}$ for \eqref{1.1} as
\begin{equation}\label{x2.42}
D_n^{(q,r)}:=\prod_{j=-\infty}^{n}\left(1-q_j\,r_j\right),\quad D_\infty^{(q,r)}:=\prod_{j=-\infty}^{\infty}\left(1-q_j\,r_j\right).
\end{equation}
For \eqref{1.1} we also define the quantities $E_n^{(q,r)}$ and $E_\infty^{(q,r)}$ as
\begin{equation}\label{x2.43}
E_n^{(q,r)}:=\prod_{j=-\infty}^{n}\left(1+q_j\,r_{j+1}\right),\quad 	E_\infty^{(q,r)}:=\prod_{j=-\infty}^{\infty}\left(1+q_j\,r_{j+1}\right).
\end{equation}
From \eqref{1.1a} and the fact that $q_n$ and $r_n$ decay rapidly, it follows that the quantities $D_n^{(q,r)},$  $D_\infty^{(q,r)},$ $E_n^{(q,r)},$  $D_\infty^{(q,r)}$ are each well defined and nonzero.

Let us introduce the scalar quantities $S_n^{(q,r)}$ and $Q_n^{(q,r)}$ as
\begin{equation}\label{Sn}
S_n^{(q,r)}:=\sum_{k=-\infty}^{n}\ds\frac{r_{k}\left(q_{k}-q_{k+1}-q_{k}
\,q_{k+1}\,r_{k+1}\right)}{(1-q_{k}\,r_{k})(1-q_{k+1}\,r_{k+1})},
\end{equation}	
	\begin{equation}\label{Qn}
Q_n^{(q,r)}:=\sum_{k=-\infty}^{n}\ds\frac{r_{k+2}\left(q_{k}-q_{k+1}-q_{k}
\,q_{k+1}\,r_{k+1}\right)}{(1+q_{k}\,r_{k+1})(1+q_{k+1}\,r_{k+2})}.
\end{equation}
Letting $n\to +\infty$ in \eqref{Sn} and \eqref{Qn} we get
	\begin{equation}\label{Sn1}
S_{\infty}^{(q,r)}:=\sum_{k=-\infty}^{\infty}\ds\frac{r_{k}\left(q_{k}-q_{k+1}-q_{k}
\,q_{k+1}\,r_{k+1}\right)}{(1-q_{k}\,r_{k})(1-q_{k+1}\,r_{k+1})},
\end{equation}	
\begin{equation}\label{Qn1}
Q_{\infty}^{(q,r)}:=\sum_{k=-\infty}^{\infty}\ds\frac{r_{k+2}\left(q_{k}-q_{k+1}-q_{k}
\,q_{k+1}\,r_{k+1}\right)}{(1+q_{k}\,r_{k+1})(1+q_{k+1}\,r_{k+2})}.
\end{equation}

 As stated in Theorem~\ref{thm:theorem x2.4}(a), the transmission coefficient $T^{(q,r)}$ for \eqref{1.1} has a meromorphic extension to $|z|<1$ and the transmission coefficient $\bar{T}^{(q,r)}$  for \eqref{1.1} has a meromorphic extension to $|z|>1.$ The asymptotics of those two transmission coefficients are presented next.

\begin{proposition}
	\label{thm:theorem x2.5}
		Assume that the potentials $q_n$ and $r_n$ appearing in \eqref{1.1} are rapidly decaying and satisfy \eqref{1.1a}. Then, the small-$z$ asymptotics of transmission coefficient $T^{(q,r)}$ for \eqref{1.1} is given by
\begin{equation}\label{x2.40} T^{(q,r)}=\ds\frac{1}{D_\infty^{(q,r)}}\left[1-z^2\,S_{\infty}^{(q,r)}
+O\left(z^4\right)\right],\qquad z \to 0,
\end{equation}
where  $D_\infty^{(q,r)}$ and $S_{\infty}^{(q,r)}$ are the scalar quantities defined in \eqref{x2.42} and \eqref{Sn1}, respectively.
Similarly,	the large-$z$ asymptotics of transmission coefficient $\bar{T}^{(q,r)}$ for \eqref{1.1} is given by
\begin{equation}\label{x2.41}	\bar{T}^{(q,r)}=\ds\frac{1}{E_\infty^{(q,r)}}\left[1-\frac{1}{z^2}
\,Q_{\infty}^{(q,r)}+O\left(\frac{1}{z^4}\right)\right],
\qquad z \to \infty,	\end{equation}
	where $E_\infty^{(q,r)}$ and $Q_{\infty}^{(q,r)}$ are the scalar quantities defined in \eqref{x2.43} and \eqref{Qn1}, respectively.
\end{proposition}

\begin{proof}
The proof is lengthy but straightforward. To obtain \eqref{x2.40} we use \eqref{1.1} with the Jost solution $\psi_{n}^{(q,r)},$  premultiply both sides of \eqref{1.1} with $z^{-n}[
0\quad 1],$ iterate the resulting equation, and for $n<m$ we get
\begin{equation}\label{x.201}
\begin{bmatrix}
0&1
\end{bmatrix}z^{-n}\,\psi_{n}^{(q,r)}=\begin{bmatrix}
0&1
\end{bmatrix}\Xi_n\,\Xi_{n+1}\cdots \Xi_m\,z^{-m-1}\,\psi_{m+1}^{(q,r)},
\end{equation}
where we have defined
\begin{equation*}
\Xi_n:=\begin{bmatrix}
0&-q_n\\
\noalign{\medskip}
0&1-q_nr_n
\end{bmatrix}+z^{2}\begin{bmatrix}
1&q_n\\
\noalign{\medskip}
r_n&q_nr_n
\end{bmatrix}.
\end{equation*}
We note that in the limit $n\to-\infty$ the left-hand side of \eqref{x.201} yields $1/T^{(q,r)}$. Letting $n\to-\infty$ and $m\to+\infty$  in \eqref{x.201}, using \eqref{x2.3} and \eqref{x2.7} for $\psi_{n}^{(q,r)}$ and also using $D_\infty^{(q,r)}$ defined in \eqref{x2.42} and $S_{\infty}^{(q,r)}$ defined in \eqref{Sn1}, after some straightforward algebra we get \eqref{x2.40}. The proof of \eqref{x2.41} is similarly obtained by using \eqref{1.1} with the Jost solution $\bar{\psi}_{n}^{(q,r)}$, premultiplying both sides of \eqref{1.1} with $z^{n}[
1\quad 0],$ iterating the resulting equation, and using \eqref{x2.5} and \eqref{x2.10} for $\bar{\psi}_{n}^{(q,r)}$ and also using $E_\infty^{(q,r)}$ defined in \eqref{x2.43} and $Q_\infty^{(q,r)}$ defined in \eqref{Qn1}.
\end{proof}

In the next theorem we provide some further relevant properties of the transmission coefficients for \eqref{1.1}.

\begin{theorem}
	\label{thm:theorem x2.6}
	Assume that the potentials $q_n$ and $r_n$ appearing in \eqref{1.1} are rapidly decaying and satisfy \eqref{1.1a}. Then, for the transmission coefficients $T^{(q,r)}$ and $\bar{T}^{(q,r)}$ of \eqref{1.1} we have the following:
	
\begin{enumerate}
		\item[\text{\rm(a)}] Neither $T^{(q,r)}$ nor $\bar{T}^{(q,r)}$ can vanish when $z\in\mathbb{T}.$
		
		\item[\text{\rm(b)}] We have
		\begin{equation}\label{x.1}
		\ds\frac{1}{T^{(q,r)}(0)}=D_\infty^{(q,r)},\quad \ds\frac{1}{\bar{T}^{(q,r)}(\infty)}=E_\infty^{(q,r)}.
		\end{equation}
	
		\item[\text{\rm(c)}] The quantity $1/T^{(q,r)}$ has at most a finite number of zeros in $0<|z|<1$ and the multiplicity of each such zero is finite.
		
 	   \item[\text{\rm(d)}]	The quantity $1/\bar{T}^{(q,r)}$ has at most a finite number of zeros in $|z|>1$ and the multiplicity of each such zero is finite.
		
		\item[\text{\rm(e)}] The transmission coefficient $T^{(q,r)}$  is meromorphic in $|z|<1,$ and the number of its poles and the multiplicities of those poles are both finite. Similarly, $\bar{T}^{(q,r)}$  is meromorphic in $|z|>1,$ and the number of its poles and the multiplicities of those poles are both finite.
		
		\item[\text{\rm(f)}] If $z_j$ is a pole of $T^{(q,r)}$ in $0<|z|<1,$ then $-z_j$ is also a pole of $T^{(q,r)}$. Similarly, if $\bar{z}_j$ is a pole of $\bar{T}^{(q,r)}$ in $|z|>1,$ then $-\bar{z}_j$ is also a pole  of $\bar{T}^{(q,r)}.$
\end{enumerate}
\end{theorem}

\begin{proof}
We note that \eqref{x.1} follows from \eqref{x2.40} and \eqref{x2.41}. The rest of the proof can be given in a way similar to the proof of Theorem~\ref{thm:theorem x2.3a}.
\end{proof}

Finally, in this section we clarify the relationship between the poles of the transmission coefficients and the linear dependence of the relevant Jost solutions for each of the linear systems \eqref{1.1}, \eqref{1.2aa}, and \eqref{1.2ab}. This clarification has many important consequences. It allows us to introduce the dependency constants at the bound states. It also allows us to deal with bound states of any multiplicities in an elegant manner. The treatment given here for the linear systems \eqref{1.1}, \eqref{1.2aa}, and \eqref{1.2ab} can be readily generalized to other linear systems both in the discrete and continuous cases.

In terms of the Jost solutions $\psi_n,$ $\phi_n,$ $\bar{\psi}_n,$ $\bar{\phi}_n$ appearing in \eqref{x2.3}--\eqref{x2.6}
for each of the linear systems
\eqref{1.1}, \eqref{1.2aa}, and \eqref{1.2ab}, we define
\begin{equation}\label{ta.5001}
a_n^{(q,r)}:=\begin{vmatrix}
\phi_n^{(q,r)}&\psi_n^{(q,r)}
\end{vmatrix},\quad \bar{a}_n^{(q,r)}:=\begin{vmatrix}
\bar{\phi}_n^{(q,r)}&\bar{\psi}_n^{(q,r)}
\end{vmatrix},
\end{equation}
\begin{equation}\label{ta.5002}
a_n^{(u,v)}:=\begin{vmatrix}
\phi_n^{(u,v)}&\psi_n^{(u,v)}
\end{vmatrix},\quad \bar{a}_n^{(u,v)}:=\begin{vmatrix}
\bar{\phi}_n^{(u,v)}&\bar{\psi}_n^{(u,v)}
\end{vmatrix},
\end{equation}
\begin{equation}\label{ta.5003}
a_n^{(p,s)}:=\begin{vmatrix}
\phi_n^{(p,s)}&\psi_n^{(p,s)}
\end{vmatrix},\quad \bar{a}_n^{(p,s)}:=\begin{vmatrix}
\bar{\phi}_n^{(p,s)}&\bar{\psi}_n^{(p,s)}
\end{vmatrix},
\end{equation}
where on the right-hand sides we have the Wronskian determinants.

The relationships among $a_n^{(q,r)},$ $\bar{a}_n^{(q,r)},$ and the transmission coefficients $T^{(q,r)}$ and
$\bar{T}^{(q,r)}$ for \eqref{1.1} are clarified in the following theorem.

\begin{theorem}
\label{thm:theorem ta.5004}
Assume that the potentials $q_n$ and $r_n$ appearing in \eqref{1.1} are rapidly decaying and satisfy \eqref{1.1a}. Then, we have the following:

\begin{enumerate}

\item[\text{\rm(a)}] The scalar quantities $a_n^{(q,r)}$ and $\bar{a}_n^{(q,r)}$ defined in \eqref{ta.5001} are independent of $n,$ depend only on $z,$ and are related to the transmission coefficients $T^{(q,r)}$ and
$\bar{T}^{(q,r)}$ appearing in \eqref{x2.32}, \eqref{x2.33}, \eqref{x2.37} as
\begin{equation}\label{ta.5005}
a_n^{(q,r)}=\ds\frac{1}{T^{(q,r)}}, \quad \bar{a}_n^{(q,r)}=-\ds\frac{1}{\bar{T}^{(q,r)}}.
\end{equation}

\item[\text{\rm(b)}] Consequently, the linear dependence of the Jost solutions $\phi_n^{(q,r)}$ and $\psi_n^{(q,r)}$
occurs at the poles of $T^{(q,r)}$ in $0<|z|<1.$ Similarly, the linear dependence of the Jost solutions $\bar{\phi}_n^{(q,r)}$ and $\bar{\psi}_n^{(q,r)}$ occurs at the poles of $\bar{T}^{(q,r)}$ in $|z|>1.$

\item[\text{\rm(c)}] In particular, if $T^{(q,r)}$ has a pole at $z=\pm z_j$ with multiplicity $m_j,$ then we have
\begin{equation}\label{ta.5006}
\ds\frac{d^k a_n^{(q,r)}}{dz^k}\bigg |_{z=\pm z_j}=0,\qquad 0\le k\le m_{j-1},
\quad n\in \mathbb{Z}.
\end{equation}
Similarly, if $\bar{T}^{(q,r)}$ has a pole at $z=\pm \bar{z}_j$ with multiplicity $\bar{m}_j,$ we then have
\begin{equation}\label{ta.5007}
\ds\frac{d^k \bar{a}_n^{(q,r)}}{dz^k}\bigg |_{z=\pm \bar{z}_j}=0,\qquad 0\le k\le \bar{m}_{j-1},\quad n\in \mathbb{Z}.
\end{equation}
\end{enumerate}
\end{theorem}
\begin{proof}
Note that \eqref{ta.5005} is obtained directly from \eqref{x2.32}, \eqref{x2.33}, \eqref{x2.37}, and \eqref{ta.5001}.
Since each of $\psi_n^{(q,r)},$ $\phi_n^{(q,r)},$ $\bar{\psi}_n^{(q,r)},$  $\bar{\phi}_n^{(q,r)}$
satisfies the same linear system \eqref{1.1}, the linear dependence
and the vanishing of the Wronskian determinant are equivalent.
We also note that \eqref{ta.5006} and \eqref{ta.5007} directly follow \eqref{ta.5005}.
\end{proof}

In the next theorem we clarify the relationships among $a_n^{(u,v)},$ $\bar{a}_n^{(u,v)},$ and the transmission coefficients for \eqref{1.2aa}.
\begin{theorem}
	\label{thm:theorem ta.5008}
Assume that the potentials $u_n$ and $v_n$ appearing in \eqref{1.2aa}
are rapidly decaying and that $1-u_nv_n\ne 0$ for $n\in\mathbb{Z}$.
Then, we have the following:
\begin{enumerate}
	
	\item[\text{\rm(a)}] The scalar quantities $a_n^{(u,v)}$ and $\bar{a}_n^{(u,v)}$
	defined in \eqref{ta.5002} depend both on $n$ and $z,$ and
	they are related to the transmission coefficients $T_{\rm r}^{(u,v)}$ and
	$\bar{T}_{\rm r}^{(u,v)}$ appearing in \eqref{x2.26} and \eqref{x2.27} as

	\begin{equation}\label{ta.5009}
	a_n^{(u,v)}=\ds\frac{D_\infty^{(u,v)}}{D_{n-1}^{(u,v)}}\,\ds\frac{1}{T_{\rm r}^{(u,v)}}, \quad \bar{a}_n^{(u,v)}=-\ds\frac{D_\infty^{(u,v)}}{D_{n-1}^{(u,v)}}\,\ds\frac{1}{\bar{T}_{\rm r}^{(u,v)}},
	\end{equation}
	where $D_n^{(u,v)}$ and $D_\infty^{(u,v)}$ are the scalar quantities defined in
	\eqref{x2.D_n}.
	
	\item[\text{\rm(b)}] Since $D_\infty^{(u,v)}\ne 0$ and $D_n^{(u,v)}\ne 0$
	 for $n\in\mathbb{Z}$ and these quantities
 do not contain $z,$ we conclude from \eqref{ta.5002} and \eqref{ta.5009} that
	 the linear dependence of the Jost solutions $\phi_n^{(u,v)}$ and $\psi_n^{(u,v)}$
	 occurs at the poles of $T_{\rm r}^{(u,v)}$ and that the linear dependence of the Jost solutions $\bar{\phi}_n^{(u,v)}$ and $\bar{\psi}_n^{(u,v)}$ occurs at the poles of $\bar{T}_{\rm r}^{(u,v)}.$
	
	\item[\text{\rm(c)}] In particular, if $T_{\rm r}^{(u,v)}$ has a pole at $z=\pm z_j$
	of multiplicity $m_j,$ then we have
	\begin{equation}\label{ta.5010}
	\ds\frac{d^k a_n^{(u,v)}}{dz^k}\bigg|_{z=\pm z_j}=0,\qquad 0\le k\le m_{j-1},
\quad n\in \mathbb{Z}.
	\end{equation}
	Similarly, if $\bar{T}_{\rm r}^{(u,v)}$ has a pole at $z=\pm \bar{z}_j$ of multiplicity $\bar{m}_j,$ then we have
	\begin{equation}\label{ta.5011}
	\ds\frac{d^k \bar{a}_n^{(u,v)}}{dz^k}\bigg|_{z=\pm \bar{z}_j}=0,\qquad 0\le k\le \bar{m}_{j-1},
\quad n\in \mathbb{Z}.
	\end{equation}
\end{enumerate}
\end{theorem}

\begin{proof}
Let us use $\mathcal{X}_n$ and $|\mathcal{X}_n|$ to denote the coefficient matrix in \eqref{1.2aa} and its determinant, respectively, i.e.
\begin{equation}\label{ta.5012}
\mathcal{X}_n:=\begin{bmatrix}
z &z\, u_n\\
\noalign{\medskip}
\ds\frac{1}{z}\, v_n &\ds\frac{1}{z}
\end{bmatrix}, \quad |\mathcal{X}_n|:=1-u_nv_n.
\end{equation}
From \eqref{1.2aa} and \eqref{ta.5002} we get
\begin{equation*}
a_n^{(u,v)}=\begin{vmatrix}
\mathcal{X}_n\,\phi_{n+1}^{(u,v)}&\mathcal{X}_n\,\psi_{n+1}^{(u,v)}
\end{vmatrix},
\end{equation*}
or equivalently
\begin{equation*}
a_n^{(u,v)}=\begin{vmatrix}
\mathcal{X}_n
\end{vmatrix}\begin{vmatrix}
\phi_{n+1}^{(u,v)}&\psi_{n+1}^{(u,v)}
\end{vmatrix}.
\end{equation*}
Iterating in this manner, from \eqref{1.2aa} and \eqref{ta.5002}, for $m>n$ we obtain
\begin{equation}\label{ta.5013}
a_n^{(u,v)}=\begin{vmatrix}
\mathcal{X}_n
\end{vmatrix}\begin{vmatrix}
\mathcal{X}_{n+1}
\end{vmatrix}
\cdots
\begin{vmatrix}
\mathcal{X}_m
\end{vmatrix}
\begin{vmatrix}
\phi_{m+1}^{(u,v)}&\psi_{m+1}^{(u,v)}
\end{vmatrix}.
\end{equation}
Letting $m\to + \infty$ in \eqref{ta.5013} and using \eqref{x2.3} and \eqref{x2.8}
 for the potential pair $(u,v),$ with the help of \eqref{ta.5012} we get
\begin{equation*}
a_n^{(u,v)}=\left(\ds\prod_{k=n}^{\infty}(1-u_k\,v_k)\right)\,\ds\frac{1}{T_{\rm r}^{(u,v)}},
\end{equation*}
which is equivalent to
\begin{equation}\label{ta.5014}
a_n^{(u,v)}=\ds\frac{\ds\prod_{k=-\infty}^{\infty}(1-u_k\,v_k)}
{\ds\prod_{k=-\infty}^{n-1}(1-u_k\,v_k)}
\ds\frac{1}{T_{\rm r}^{(u,v)}}.
\end{equation}
Using \eqref{x2.D_n} in \eqref{ta.5014} we get the first equality in \eqref{ta.5009}. The second equality in \eqref{ta.5009} is obtained by iterating \eqref{1.2aa} as
\begin{equation*}
\begin{vmatrix}
\bar{\phi}_n^{(u,v)}&\bar{\psi}_n^{(u,v)}
\end{vmatrix}=\begin{vmatrix}
\mathcal{X}_n
\end{vmatrix}\begin{vmatrix}
\mathcal{X}_{n+1}
\end{vmatrix}
\cdots
\begin{vmatrix}
\mathcal{X}_m
\end{vmatrix}
\begin{vmatrix}
\bar{\phi}_{m+1}^{(u,v)}&\bar{\psi}_{m+1}^{(u,v)}
\end{vmatrix},
\end{equation*}
and letting $m\to + \infty$ and using \eqref{x2.5} and \eqref{x2.9} for the potential pair $(u,v)$. We finally remark that (b) and (c) are direct consequences of (a).
\end{proof}

 The result of Theorem~\ref{thm:theorem ta.5008} is remarkable in the sense that, even though the scattering coefficients for \eqref{1.2aa} cannot be defined as
 some Wronskians of Jost solutions to \eqref{1.2aa} as in \eqref{x2.32}--\eqref{x2.35}, we have the relations given by
\begin{equation}\label{ta.5015}
\frac{1}{T_{\rm r}^{(u,v)}}=\ds\frac{D_{n-1}^{(u,v)}}{D_\infty^{(u,v)}}\begin{vmatrix}
\phi_n^{(u,v)}&\psi_n^{(u,v)}
\end{vmatrix},\quad \frac{1}{\bar{T}_{\rm r}^{(u,v)}}=\ds\frac{D_{n-1}^{(u,v)}}{D_\infty^{(u,v)}}\begin{vmatrix}
\bar{\psi}_n^{(u,v)}&\bar{\phi}_n^{(u,v)}
\end{vmatrix},
\end{equation}
\begin{equation}\label{ta.5015a}
\frac{1}{T_{\rm l}^{(u,v)}}=\ds\frac{1}{D_{n-1}^{(u,v)}}\begin{vmatrix}
\phi_n^{(u,v)}&\psi_n^{(u,v)}
\end{vmatrix},\quad \frac{1}{\bar{T}_{\rm l}^{(u,v)}}=\ds\frac{1}{D_{n-1}^{(u,v)}}\begin{vmatrix}
\bar{\psi}_n^{(u,v)}&\bar{\phi}_n^{(u,v)}
\end{vmatrix},
\end{equation}
\begin{equation}\label{ta.5015c}
\frac{R^{(u,v)}}{T_{\rm r}^{(u,v)}}=\ds\frac{D_{n-1}^{(u,v)}}{D_\infty^{(u,v)}}\begin{vmatrix}
\bar{\psi}_n^{(u,v)}&\phi_n^{(u,v)}
\end{vmatrix},\quad \frac{\bar{R}^{(u,v)}}{\bar{T}_{\rm r}^{(u,v)}}=\ds\frac{D_{n-1}^{(u,v)}}{D_\infty^{(u,v)}}\begin{vmatrix}
\bar{\phi}_n^{(u,v)}&\psi_{n}^{(u,v)}
\end{vmatrix},
\end{equation}
\begin{equation}\label{ta.5015b}
\frac{L^{(u,v)}}{T_{\rm l}^{(u,v)}}=\ds\frac{1}{D_{n-1}^{(u,v)}}\begin{vmatrix}
\psi_n^{(u,v)}&\bar{\phi}_n^{(u,v)}
\end{vmatrix},\quad \frac{\bar{L}^{(u,v)}}{\bar{T}_{\rm l}^{(u,v)}}=\ds\frac{1}{D_{n-1}^{(u,v)}}\begin{vmatrix}
\phi_n^{(u,v)}&\bar{\psi}_{n}^{(u,v)}
\end{vmatrix}.
\end{equation}
Note that \eqref{ta.5015} and \eqref{ta.5015c} can be obtained
by using a forward iteration as in \eqref{ta.5013} and that
\eqref{ta.5015a} and \eqref{ta.5015b} can be obtained
via a backward iteration on \eqref{1.2aa}. For example, the first equality in \eqref{ta.5015a}
 is obtained by letting
 $m\to -\infty$ in
 \begin{equation*}
\begin{vmatrix}
\phi_n^{(u,v)}&\psi_n^{(u,v)}
\end{vmatrix}=\begin{vmatrix}
\mathcal{X}_{n-1}^{-1}
\end{vmatrix}\begin{vmatrix}
\mathcal{X}_{n-2}^{-1}
\end{vmatrix}
\cdots
\begin{vmatrix}
\mathcal{X}_m^{-1}
\end{vmatrix}
\begin{vmatrix}
\phi_m^{(u,v)}&\psi_m^{(u,v)}
\end{vmatrix}.
\end{equation*}
 Note that the scalar quantity $D_\infty^{(u,v)}$ in \eqref{ta.5015}
is independent of $z$ and nonzero. Furthermore, each $D_n^{(u,v)}$
 for $n\in\mathbb{Z}$ is independent of  $z$ and nonzero.
 Thus, \eqref{ta.5015} allows us to directly relate the poles of $T_{\rm r}^{(u,v)}$ to the zeros of the Wronskian determinant
 $|\phi_n^{(u,v)}\quad \psi_n^{(u,v)}|,$ and similarly we can directly relate the poles of $\bar{T}_{\rm r}^{(u,v)}$ to the zeros of the Wronskian determinant
 $|\bar{\phi}_n^{(u,v)}\quad \bar{\psi}_n^{(u,v)}|.$

 The results stated in Theorem~\ref{thm:theorem ta.5008} hold for \eqref{1.2ab} as well.
 In the following corollary we state those results without a proof
 since that proof is similar to the proof of Theorem~\ref{thm:theorem ta.5008}.

 \begin{corollary}
 		\label{thm:theorem ta.5016}
 	Assume that the potentials $p_n$ and $s_n$ appearing in \eqref{1.2ab}
 	are rapidly decaying and that $1-p_ns_n\ne 0$ for $n\in\mathbb{Z}$.
 	Then, we have the following:
 	\begin{enumerate}
 		
 		\item[\text{\rm(a)}] The scalar quantities $a_n^{(p,s)}$ and $\bar{a}_n^{(p,s)}$
 		defined in \eqref{ta.5003} depend both on $n$ and $z,$ and
 		they are related to the transmission coefficients $T_{\rm r}^{(p,s)}$ and
 		$\bar{T}_{\rm r}^{(p,s)}$ appearing in \eqref{x2.8} and \eqref{x2.9} for the potential pair $(p,s)$ as 		
 		\begin{equation}\label{ta.5017}
 		a_n^{(p,s)}=\ds\frac{D_\infty^{(p,s)}}{D_{n-1}^{(p,s)}}\,\ds\frac{1}{T_{\rm r}^{(p,s)}}, \quad \bar{a}_n^{(p,s)}=-\ds\frac{D_\infty^{(p,s)}}{D_{n-1}^{(p,s)}}\,\ds\frac{1}{\bar{T}_{\rm r}^{(p,s)}},
 		\end{equation}
 		where $D_n^{(p,s)}$ and $D_\infty^{(p,s)}$ are the scalar quantities defined in
 		\eqref{x2.D_na}.
 		
 		\item[\text{\rm(b)}] Since $D_\infty^{(p,s)}\ne 0$ and $D_n^{(p,s)}\ne 0$
 		for $n\in\mathbb{Z}$ and these quantities do not contain $z,$
  we conclude from \eqref{ta.5003} and \eqref{ta.5017} that
 		the linear dependence of the Jost solutions $\phi_n^{(p,s)}$ and $\psi_n^{(p,s)}$
 		occurs at the poles of $T_{\rm r}^{(p,s)}$ and that the linear dependence of the Jost solutions $\bar{\phi}_n^{(p,s)}$ and $\bar{\psi}_n^{(p,s)}$ occurs at the poles of $\bar{T}_{\rm r}^{(p,s)}.$
 		
 		\item[\text{\rm(c)}] In particular, if $T_{\rm r}^{(p,s)}$ has a pole at $z=\pm z_j$
 		of multiplicity $m_j,$ then we have
 		\begin{equation}\label{ta.5018}
 		\ds\frac{d^k a_n^{(p,s)}}{dz^k}\bigg |_{z=\pm z_j}=0,\qquad 0\le k\le m_{j-1},
 \quad n\in \mathbb{Z}.
 		\end{equation}
 		Similarly, if $\bar{T}_{\rm r}^{(p,s)}$ has a pole at $z=\pm \bar{z}_j$ of multiplicity $\bar{m}_j,$ then we have
 		\begin{equation}\label{ta.5019}
 		\ds\frac{d^k \bar{a}_n^{(p,s)}}{dz^k}\bigg |_{z=\pm \bar{z}_j}=0,\qquad 0\le k\le \bar{m}_{j-1},
 \quad n\in \mathbb{Z}.
 		\end{equation}
 	\end{enumerate}
 	 \end{corollary}

\section{The transformations}
\label{sec:section3}

In this section we relate the linear systems \eqref{1.1}, \eqref{1.2aa}, \eqref{1.2ab} to each other by choosing the potential pairs $(u,v)$ and $(p,s)$ in terms of the potential pair $(q,r)$ in a particular way, namely as
	\begin{equation}\label{x3.1}
u_n=q_n\,\displaystyle\frac{E_{n-1}^{(q,r)}}{D_n^{(q,r)}},
\end{equation}
\begin{equation}\label{x3.2}
v_n=\left(-r_n+r_{n+1}-q_n\,r_n\,r_{n+1}\right)\,\displaystyle\frac{D_{n-1}^{(q,r)}}{E_n^{(q,r)}},
\end{equation}
\begin{equation}\label{x3.3} p_n=\left(q_n-q_{n+1}-q_n\,q_{n+1}\,r_{n+1}\right)\,\displaystyle\frac{E_{n-1}^{(q,r)}}{D_{n+1}^{(q,r)}},
\end{equation}
\begin{equation}\label{x3.4}
s_n=r_{n+1}\,\displaystyle\frac{D_{n}^{(q,r)}}{E_n^{(q,r)}},
\end{equation}
where $D_n^{(q,r)}$ and $E_n^{(q,r)}$ are the quantities defined in \eqref{x2.42} and \eqref{x2.43}, respectively. This helps us to express the Jost solutions and the scattering coefficients for \eqref{1.2aa} and \eqref{1.2ab} in terms of the corresponding quantities for \eqref{1.1}. In this section we also present certain relevant properties of the Jost solutions to \eqref{1.1}, and we express $q_n$ and $r_n$ in terms
 of the values at $z=1$ of the Jost solutions to
\eqref{1.2aa} and \eqref{1.2ab}.
The results presented in this section play a crucial role in solving the direct and inverse scattering problems for \eqref{1.1} by exploiting the techniques for the corresponding problems for \eqref{1.2aa} and \eqref{1.2ab}.

In the next proposition, when \eqref{x3.1}--\eqref{x3.4} are valid we present some relationships among the quantities for the potential pairs $(q,r),$ $(u,v),$ and $(p,s).$

\begin{proposition}
	\label{thm:theorem x3.1}
		Assume that the potentials $q_n$ and $r_n$ appearing in \eqref{1.1} are rapidly decaying and satisfy \eqref{1.1a}. Let the potential pairs $(u,v)$ and $(p,s)$ are related to $(q,r)$ as in \eqref{x3.1}--\eqref{x3.4}. We then have
	\begin{equation}\label{x.101}
	1-u_nv_n=\ds\frac{1}{(1-q_n\,r_n)(1+q_{n}\,r_{n+1})},
		\end{equation}
			\begin{equation}\label{x.102}
		1-p_ns_n=\ds\frac{1}{(1-q_{n+1}\,r_{n+1})(1+q_{n}\,r_{n+1})},
		\end{equation}
	\begin{equation}\label{x.103}
D_n^{(u,v)}=\ds\frac{1}{D_n^{(q,r)}\,E_n^{(q,r)}}, \quad D_\infty^{(u,v)}=\ds\frac{1}{D_\infty^{(q,r)}\,E_\infty^{(q,r)}},
\end{equation}
\begin{equation}\label{x.104}
D_n^{(p,s)}=\ds\frac{1}{D_{n+1}^{(q,r)}\,E_n^{(q,r)}}, \quad D_\infty^{(p,s)}=\ds\frac{1}{D_\infty^{(q,r)}\,E_\infty^{(q,r)}},
\end{equation}			
where we recall that $D_n^{(u,v)}$ and $D_\infty^{(u,v)}$ are as in \eqref{x2.D_n}, $D_n^{(p,s)}$ and $D_\infty^{(p,s)}$ are as in \eqref{x2.D_na}, $D_n^{(q,r)}$ and $D_\infty^{(q,r)}$ are as in \eqref{x2.42}, and $E_n^{(q,r)}$ and $E_\infty^{(q,r)}$ are as in \eqref{x2.43}.
\end{proposition}

\begin{proof}
	We evaluate the left-hand side of \eqref{x.101} with the help of \eqref{x2.42},  \eqref{x2.43}, \eqref{x3.1}, \eqref{x3.2}, and after a brief simplification we establish \eqref{x.101}. Similarly, we obtain \eqref{x.102} with the help of \eqref{x2.42},  \eqref{x2.43}, \eqref{x3.3}, and \eqref{x3.4}. Then, we establish \eqref{x.103} by using \eqref{x.101} in \eqref{x2.D_n}, \eqref{x2.42}, and  \eqref{x2.43}. Similarly, we get \eqref{x.104} by using \eqref{x.102} in \eqref{x2.D_na}, \eqref{x2.42}, and  \eqref{x2.43}.
\end{proof}

The following proposition will be useful in solving the inverse problem for \eqref{1.1} by using the method introduced in (e) of Section~\ref{sec:section8}.

\begin{proposition}
	\label{thm:theorem x3.1a}
	Assume that the potentials $q_n$ and $r_n$ appearing in \eqref{1.1} are rapidly decaying and satisfy \eqref{1.1a}. Let the potential pairs $(u,v)$ and $(p,s)$ be related to $(q,r)$ as in \eqref{x3.1}--\eqref{x3.4}. Then, we have
		\begin{equation}\label{x.601}
	D_n^{(q,r)}=\ds\frac{1}{\ds\prod_{k=-\infty}^{n-1}\left(1+u_{k+1}\,s_{k}\right)},
	\end{equation}
	\begin{equation}\label{x.602}
	E_n^{(q,r)}=\ds\frac{1}{\ds\prod_{k=-\infty}^{n}\left(1-u_k\,s_{k}\right)},
	\end{equation}
	\begin{equation}\label{x.603}
	q_n=u_n\ds\prod_{k=-\infty}^{n-1}\ds\frac{1-u_k\,s_k}{1+u_{k+1}\,s_k},
	\end{equation}
	\begin{equation}\label{x.604}
	r_n=s_{n-1}\ds\prod_{k=-\infty}^{n-1}\ds\frac{1+u_k\,s_{k-1}}{1-u_{k}\,s_k},
	\end{equation}
where $D_n^{(q,r)}$ and $E_n^{(q,r)}$ are the scalar quantities defined in
\eqref{x2.42} and \eqref{x2.43}, respectively.
\end{proposition}

\begin{proof}
From \eqref{x3.1} and \eqref{x3.4}, we obtain
\begin{equation*}
1-u_n\,s_n=1-q_n\,\ds\frac{E_{n-1}^{(q,r)}}{D_n^{(q,r)}}\,r_{n+1}\,\frac{D_n^{(q,r)}}{E_n^{(q,r)}},
\end{equation*}
which, with the help of \eqref{x2.43}, simplifies to
\begin{equation}\label{Y.1}
1-u_n\,s_n=\ds\frac{1}{1+q_n\,r_{n+1}}.
\end{equation}
Similarly, from \eqref{x3.1} and \eqref{x3.4} we get
\begin{equation*}
1+u_n\,s_{n-1}=1+q_n\,\ds\frac{E_{n-1}^{(q,r)}}{D_n^{(q,r)}}\,r_{n}\,\frac{D_{n-1}^{(q,r)}}{E_{n-1}^{(q,r)}},
\end{equation*}
which, with the help of \eqref{x2.42}, simplifies to
\begin{equation}\label{Y.2}
1+u_n\,s_{n-1}=\ds\frac{1}{1-q_n\,r_{n}}.
\end{equation}
From \eqref{Y.1} and \eqref{Y.2}, we respectively get
\begin{equation*}
\ds\prod_{k=-\infty}^{n}\left(1-u_k\,s_k\right)=\ds\frac{1}{\ds\prod_{k=-\infty}^{n}\left(1+q_k\,r_{k+1}\right)},
\end{equation*}
\begin{equation*}
\ds\prod_{k=-\infty}^{n}\left(1+u_k\,s_{k-1}\right)=\ds\frac{1}{\ds\prod_{k=-\infty}^{n}\left(1-q_k\,r_{k}\right)},
\end{equation*}
which yield \eqref{x.602} and \eqref{x.601}, respectively. Finally, by using \eqref{x.601} and \eqref{x.602} in \eqref{x3.1} and \eqref{x3.4}, we obtain \eqref{x.603} and \eqref{x.604}, respectively.
\end{proof}

In the next theorem, when the potential pairs are related to each other as in \eqref{x3.1}--\eqref{x3.4}, the Jost solutions to \eqref{1.1} are
related to the Jost solutions to \eqref{1.2aa}
and also to the Jost solutions to \eqref{1.2ab}.

\begin{theorem}
\label{thm:theorem x3.2}
	Assume that the potentials $q_n$ and $r_n$ appearing in \eqref{1.1} are rapidly decaying and satisfy \eqref{1.1a}, and let the potential pairs $(u,v)$ and $(p,s)$ be related to the potential pair $(q,r)$ as in \eqref{x3.1}--\eqref{x3.4}.
Then, the four Jost solutions
$\psi_n^{(q,r)},$ $\phi_n^{(q,r)},$	$\bar{\psi}_n^{(q,r)},$  $\bar{\phi}_n^{(q,r)}$ to \eqref{1.1} are related to the Jost solutions
$\psi_n^{(u,v)},$ $\phi_n^{(u,v)},$ $\bar{\psi}_n^{(u,v)},$  $\bar{\phi}_n^{(u,v)}$ to \eqref{1.2aa} and the
Jost solutions $\psi_n^{(p,s)},$ $\phi_n^{(p,s)},$ $\bar{\psi}_n^{(p,s)},$  $\bar{\phi}_n^{(p,s)}$ to \eqref{1.2ab}
as
	\begin{equation}\label{x3.5}
\psi_n^{(q,r)}=D_\infty^{(q,r)}
\begin{bmatrix}
\left(1-\ds\frac{1}{z^2}\right) \,\ds\frac{1}{E_{n-1}^{(q,r)}} &0\\
\noalign{\medskip}
\ds\frac{r_n}{E_{n-1}^{(q,r)}} &\ds\frac{1}{D_{n-1}^{(q,r)}}
\end{bmatrix}
\psi_n^{(u,v)},
\end{equation}
\begin{equation}\label{x3.5a}
\psi_n^{(q,r)}=D_\infty^{(q,r)}
\begin{bmatrix}
\ds\frac{1}{E_{n-1}^{(q,r)}} &-\ds\frac{q_n}{D_n^{(q,r)}}\\
\noalign{\medskip}
\ds\frac{r_n}{E_{n-1}^{(q,r)}} &\ds\frac{1}{D_{n-1}^{(q,r)}}
\end{bmatrix}
\psi_n^{(p,s)},
\end{equation}
\begin{equation}\label{x3.7}
\phi_n^{(q,r)}=\ds\frac{1}{1-\ds\frac{1}{z^2}}
\begin{bmatrix}
\left(1-\ds\frac{1}{z^2}\right) \,\ds\frac{1}{E_{n-1}^{(q,r)}} &0\\
\noalign{\medskip}
\ds\frac{r_n}{E_{n-1}^{(q,r)}} &\ds\frac{1}{D_{n-1}^{(q,r)}}
\end{bmatrix}
\phi_n^{(u,v)},
\end{equation}
\begin{equation}\label{x3.7a}
\phi_n^{(q,r)}=\begin{bmatrix}
\ds\frac{1}{E_{n-1}^{(q,r)}} &-\ds\frac{q_n}{D_n^{(q,r)}}\\
\noalign{\medskip}
\ds\frac{r_n}{E_{n-1}^{(q,r)}} &\ds\frac{1}{D_{n-1}^{(q,r)}}
\end{bmatrix}
\phi_n^{(p,s)},
\end{equation}
\begin{equation}\label{x3.6}
\bar{\psi}_n^{(q,r)}=\ds\frac{E_\infty^{(q,r)}}{1-\ds\frac{1}{z^2}}
\begin{bmatrix}
\left(1-\ds\frac{1}{z^2}\right) \,\ds\frac{1}{E_{n-1}^{(q,r)}} &0\\
\noalign{\medskip}
\ds\frac{r_n}{E_{n-1}^{(q,r)}} &\ds\frac{1}{D_{n-1}^{(q,r)}}
\end{bmatrix}
\bar{\psi}_n^{(u,v)},
\end{equation}
\begin{equation}\label{x3.6a}
\bar{\psi}_n^{(q,r)}=E_\infty^{(q,r)}
\begin{bmatrix}
\ds\frac{1}{E_{n-1}^{(q,r)}} &-\ds\frac{q_n}{D_n^{(q,r)}}\\
\noalign{\medskip}
\ds\frac{r_n}{E_{n-1}^{(q,r)}} &\ds\frac{1}{D_{n-1}^{(q,r)}}
\end{bmatrix}
\bar{\psi}_n^{(p,s)},
\end{equation}
\begin{equation}\label{x3.8}
\bar{\phi}_n^{(q,r)}=
\begin{bmatrix}
\left(1-\ds\frac{1}{z^2}\right) \,\ds\frac{1}{E_{n-1}^{(q,r)}} &0\\
\noalign{\medskip}
\ds\frac{r_n}{E_{n-1}^{(q,r)}} &\ds\frac{1}{D_{n-1}^{(q,r)}}
\end{bmatrix}
\bar{\phi}_n^{(u,v)},
\end{equation}
\begin{equation}\label{x3.8a}
\bar{\phi}_n^{(q,r)}=\begin{bmatrix}
\ds\frac{1}{E_{n-1}^{(q,r)}} &-\ds\frac{q_n}{D_n^{(q,r)}}\\
\noalign{\medskip}
\ds\frac{r_n}{E_{n-1}^{(q,r)}} &\ds\frac{1}{D_{n-1}^{(q,r)}}
\end{bmatrix}
\bar{\phi}_n^{(p,s)},
\end{equation}
where we recall that $D_\infty^{(q,r)}$ and $E_\infty^{(q,r)}$ are the constants defined in \eqref{x2.42} and \eqref{x2.43}, respectively.
\end{theorem}

\begin{proof}
We only present the proof for \eqref{x3.5} because the proofs for \eqref{x3.5a}--\eqref{x3.8a} can be obtained in a similar manner. To establish \eqref{x3.5} we let
\begin{equation}\label{x.301}
\psi_{n}^{(q,r)}=\Gamma_{n}^{(q,r)}\,\psi_{n}^{(u,v)},
\end{equation}
where $\Gamma_{n}^{(q,r)}$ is a 2$\times$2 matrix to be determined. Since $\psi_{n}^{(q,r)}$ satisfies \eqref{1.1} and $\psi_{n}^{(u,v)}$ satisfies \eqref{1.2aa}, from \eqref{1.1}, \eqref{1.2aa}, \eqref{x.301} we obtain $\Gamma_{n}^{(q,r)}$ as listed in \eqref{x3.5}. As an alternate proof we remark that the reader can directly verify that each of \eqref{x3.5}--\eqref{x3.8a} is compatible with \eqref{1.1}, \eqref{1.2aa}, \eqref{x2.3}--\eqref{x2.6}, and \eqref{x3.1}--\eqref{x3.4}.
\end{proof}

In the next theorem we relate the scattering coefficients for \eqref{1.1}, \eqref{1.2aa}, \eqref{1.2ab} to each other.

\begin{theorem}
\label{thm:theorem x3.3}
Assume that the potential pair $(q,r)$ is rapidly decaying and satisfy \eqref{1.1a}. Assume also that the potential pairs $(u,v)$ and $(p,s)$ are related to $(q,r)$ as in \eqref{x3.1}--\eqref{x3.4}. Then, the scattering coefficients $T^{(q,r)},$ $\bar{T}^{(q,r)},$  $R^{(q,r)},$ $\bar{R}^{(q,r)},$ $L^{(q,r)},$ $\bar{L}^{(q,r)}$ for \eqref{1.1} are related to the scattering coefficients $T_{\rm l}^{(u,v)},$ $T_{\rm r}^{(u,v)},$
$\bar{T}_{\rm l}^{(u,v)},$  $\bar{T}_{\rm r}^{(u,v)},$ $R^{(u,v)},$ $\bar{R}^{(u,v)},$ $L^{(u,v)},$ $\bar{L}^{(u,v)}$ for \eqref{1.2aa} and $T_{\rm l}^{(p,s)},$ $T_{\rm r}^{(p,s)},$ $\bar{T}_{\rm l}^{(p,s)},$ $\bar{T}_{\rm r}^{(p,s)},$ $R^{(p,s)},$ $\bar{R}^{(p,s)},$ $L^{(p,s)},$ $\bar{L}^{(p,s)}$ for \eqref{1.2ab} as
\begin{equation}\label{x3.9}
T_{\rm l}^{(u,v)}=D_\infty^{(q,r)}\,T^{(q,r)}, \quad T_{\rm l}^{(p,s)}=D_\infty^{(q,r)}\,T^{(q,r)},
\end{equation}
\begin{equation}\label{x3.9a}
T_{\rm r}^{(u,v)}=\ds\frac{1}{E_\infty^{(q,r)}}\,T^{(q,r)}, \quad T_{\rm r}^{(p,s)}=\ds\frac{1}{E_\infty^{(q,r)}}\,T^{(q,r)},
\end{equation}
\begin{equation}\label{x3.10}
\bar{T}_{\rm l}^{(u,v)}=E_\infty^{(q,r)}\,\bar{T}^{(q,r)}, \quad \bar{T}_{\rm l}^{(p,s)}=E_\infty^{(q,r)}\,\bar{T}^{(q,r)},
\end{equation}
\begin{equation}\label{x3.10a}
\bar{T}_{\rm r}^{(u,v)}=\ds\frac{1}{D_\infty^{(q,r)}}\,\bar{T}^{(q,r)}, \quad \bar{T}_{\rm r}^{(p,s)}=\ds\frac{1}{D_\infty^{(q,r)}}\,\bar{T}^{(q,r)},
\end{equation}
\begin{equation}\label{x3.11}
R^{(u,v)}=\left(1-\ds\frac{1}{z^2}\right)\,\ds\frac{D_\infty^{(q,r)}}{E_\infty^{(q,r)}}\,R^{(q,r)}, \quad R^{(p,s)}=\ds\frac{D_\infty^{(q,r)}}{E_\infty^{(q,r)}}\,R^{(q,r)},
\end{equation}
 \begin{equation}\label{x3.11a}
\bar{R}^{(u,v)}=\ds\frac{1}{1-\ds\frac{1}{z^2}}\,\ds\frac{E_\infty^{(q,r)}}{D_\infty^{(q,r)}}\,\bar{R}^{(q,r)}, \quad \bar{R}^{(p,s)}=\ds\frac{E_\infty^{(q,r)}}{D_\infty^{(q,r)}}\,\bar{R}^{(q,r)},
\end{equation}
\begin{equation}\label{x3.12}
L^{(u,v)}=\ds\frac{1}{1-\ds\frac{1}{z^2}}\,L^{(q,r)}, \quad L^{(p,s)}= L^{(q,r)},
\end{equation}
\begin{equation}\label{x3.12a}
\bar{L}^{(u,v)}=\left(1-\ds\frac{1}{z^2}\right)\,\bar{L}^{(q,r)}, \quad \bar{L}^{(p,s)}=\bar{L}^{(q,r)},
\end{equation}
where we recall that $D_\infty^{(q,r)}$ and $E_\infty^{(q,r)}$ are the constants defined in \eqref{x2.42} and \eqref{x2.43}, respectively.
	\end{theorem}

\begin{proof}
We use the asymptotics of $\psi_{n}^{(q,r)},$ $\psi_{n}^{(u,v)},$ $\psi_{n}^{(p,s)}$ given in \eqref{x2.7} and we let $n\to -\infty$ in \eqref{x3.5}, which helps us to establish \eqref{x3.9} and \eqref{x3.9a}, respectively. We establish \eqref{x3.10}--\eqref{x3.12a} in a similar manner.
\end{proof}

When \eqref{x3.1}--\eqref{x3.4} hold, from \eqref{x.501} and \eqref{x3.9}--\eqref{x3.10a} we obtain the result stated in the following corollary.

\begin{corollary}
\label{thm:theorem x3.4}
Assume that the potentials $q_n$ and $r_n$ appearing in \eqref{1.1} are rapidly decaying and satisfy \eqref{1.1a}. Assume further that the potential pairs $(u,v)$ and $(p,s)$ are related to $(q,r)$ as in \eqref{x3.1}--\eqref{x3.4}. Then, the transmission coefficients $T^{(q,r)},$
$T_{\rm l}^{(u,v)},$  $T_{\rm r}^{(u,v)},$  $T_{\rm l}^{(p,s)},$
$T_{\rm r}^{(p,s)}$  have coinciding poles in $0<|z|<1$ and
the coinciding multiplicity for each of those poles. Similarly, $\bar{T}^{(q,r)},$  $\bar{T}_{\rm l}^{(u,v)},$  $\bar{T}_{\rm r}^{(u,v)},$  $\bar{T}_{\rm l}^{(p,s)},$  $\bar{T}_{\rm r}^{(p,s)}$ have their coinciding poles in $|z|>1$
with the coinciding multiplicity for each of those poles.
\end{corollary}

When \eqref{x3.1}--\eqref{x3.4} hold, based on Corollary~\ref{thm:theorem x3.4} we will
use $\{ {\pm z}_{j},m_j\}_{j=1}^N$ to denote the common set of poles in $0<|z|<1$ and their
multiplicities for $T^{(q,r)},$
$T_{\rm l}^{(u,v)},$  $T_{\rm r}^{(u,v)},$  $T_{\rm l}^{(p,s)},$  $T_{\rm r}^{(p,s)},$ and similarly we will use $\{ {\pm \bar{z}}_{j},\bar{m}_j\}_{j=1}^{\bar{N}}$ to denote the common set of poles in $|z|>1$ and their
multiplicities for  $\bar{T}^{(q,r)},$ $\bar{T}_{\rm l}^{(u,v)},$  $\bar{T}_{\rm r}^{(u,v)},$  $\bar{T}_{\rm l}^{(p,s)},$  $\bar{T}_{\rm r}^{(p,s)}.$

 We present some relevant properties of the Jost solutions to \eqref{1.1} in the next theorem, which is the analog of Theorem~\ref{thm:theorem x2.1} that lists the relevant properties of the Jost solutions to \eqref{1.2aa}.

\begin{theorem}
	\label{thm:theorem x3.4a}
	Assume that the potentials $q_n$ and $r_n$ appearing in \eqref{1.1} are rapidly decaying and satisfy \eqref{1.1a}. Then, the corresponding Jost solutions to \eqref{1.1} satisfy the following:
	\begin{enumerate}
	 	
	 	   \item[\text{\rm(a)}] For each $n\in\mathbb{Z}$ the quantities $z^{-n}\,\psi_n^{(q,r)},$ $z^{n}\,\phi_n^{(q,r)},$ $z^{n}\,\bar{\psi}_n^{(q,r)},$  $z^{-n}\,\bar{\phi}_n^{(q,r)}$ are even in $z$ in their respective domains.
		
			\item[\text{\rm(b)}] The quantity $z^{-n}\,\psi_n^{(q,r)}$  is analytic in $|z|<1$ and
		continuous in $|z|\le 1$.
		
		   \item[\text{\rm(c)}] The quantity $z^{n}\,\phi_n^{(q,r)}$ is analytic in $|z|<1$ and
		continuous in $|z|\le 1$.
		
		\item[\text{\rm(d)}] The quantity $z^{n}\,\bar{\psi}_n^{(q,r)}$ is analytic in $|z|>1$ and
		continuous in $|z|\ge 1$.
		
		\item[\text{\rm(e)}] The quantity $z^{-n}\,\bar{\phi}_n^{(q,r)}$ is analytic in $|z|>1$ and
		continuous in $|z|\ge 1$.
		
		\item[\text{\rm(f)}] The Jost solution $\psi_{n}^{(q,r)}$ has the expansion
		\begin{equation}\label{x5.1}
			\psi_{n}^{q,r)}=\sum_{l=n}^{\infty}K_{nl}^{(q,r)}z^l, \qquad |z|\le 1,
		\end{equation}
		with the double-indexed quantities $K_{nl}^{(q,r)}$ for which we have
		\begin{equation}\label{x5.2}
		K_{nn}^{(q,r)}=D_\infty^{(q,r)}
		\begin{bmatrix}
		-\ds\frac{q_n}{D_n^{(q,r)}}\\
		\noalign{\medskip}
		\ds\frac{1}{D_{n-1}^{(q,r)}}
		\end{bmatrix},
		\end{equation}
		\begin{equation*}
%\label{x5.3}
		K_{n(n+2)}^{(q,r)}=D_\infty^{(q,r)}
		\begin{bmatrix}
		\ds\frac{q_n}{D_{n}^{(q,r)}}-\ds\frac{q_{n+1}}
{D_{n+1}^{(q,r)}}-\ds\frac{q_n\left(S_{\infty}^{(q,r)}-S_n^{(q,r)}\right)}{D_n^{(q,r)}}\\
		\noalign{\medskip}
		\ds\frac{S_n^{(q,r)}}{D_{n-1}^{(q,r)}}
		\end{bmatrix},
		\end{equation*}
	with $D_{n}^{(q,r)},$ $D_{\infty}^{(q,r)},$ $S_n^{(q,r)},$ $S_{\infty}^{(q,r)}$ being the scalar quantities defined in \eqref{x2.42}, \eqref{Sn}, \eqref{Sn1}, respectively, and that $K_{nl}^{(q,r)}=0 $ when $n+l$ is odd or  $l<n$.

		\item[\text{\rm(g)}] The Jost solution $\bar{\psi}_{n}^{(q,r)}$ has the expansion		\begin{equation}\label{x5.4}
        \bar{\psi}_{n}^{q,r)}=\sum_{l=n}^{\infty}\bar{K}_{nl}^{(q,r)}\,\ds\frac{1}{z^{l}}, \qquad |z|\ge 1,
		\end{equation}
		with the double-indexed quantities $\bar{K}_{nl}^{(q,r)}$ for which we have
		\begin{equation}\label{x5.5}
		\bar{K}_{nn}^{(q,r)}=\ds\frac{E_\infty^{(q,r)}}{E_{n-1}^{(q,r)}}
		\begin{bmatrix}
		1\\
		\noalign{\medskip}
		r_n
		\end{bmatrix},
		\end{equation}
		\begin{equation*}
%\label{x5.6}
		\bar{K}_{n(n+2)}^{(q,r)}=E_\infty^{(q,r)}
		\begin{bmatrix}
		-\ds\frac{q_n
r_{n+1}}{E_n^{(q,r)}}+\ds\frac{Q_\infty^{(q,r)}-Q_{n-1}^{(q,r)}}{E_{n-1}^{(q,r)}}\\
		\noalign{\medskip}
		\ds\frac{r_{n+1}(1-q_nr_n)}{E_n^{(q,r)}}
+\ds\frac{r_n\left(Q_\infty^{(q,r)}-Q_{n-1}^{(q,r)}\right)}{E_{n-1}^{(q,r)}}
		\end{bmatrix},
		\end{equation*}
		with $E_{n}^{(q,r)},$ $E_\infty^{(q,r)},$ $Q_n^{(q,r)},$ $Q_\infty^{(q,r)}$ being the scalar quantities defined in \eqref{x2.43}, \eqref{Qn}, \eqref{Qn1}, respectively, and that $\bar{K}_{nl}^{(q,r)}=0 $ when $n+l$ is odd or  $l<n$.

		\item[\text{\rm(h)}] For the Jost solution $\phi_{n}^{(q,r)}$ we have the expansion
	\begin{equation*}
%\label{x5.7}
	z^{n}\,\phi_{n}^{(q,r)}=\ds\sum_{l=0}^{\infty}P_{nl}^{(q,r)}\ds z^{l}, \qquad |z|\le 1,
	\end{equation*}
   with the double-indexed quantities $P_{nl}^{(q,r)}$ for which we have
		\begin{equation*}
%\label{x5.8}
		P_{n0}^{(q,r)}=D_{n-1}^{(q,r)}
		\begin{bmatrix}
		1\\
		\noalign{\medskip}
		0
		\end{bmatrix},
		\end{equation*}
		\begin{equation*}
%\label{x5.9}
		P_{n2}^{(q,r)}=D_{n-2}^{(q,r)}
		\begin{bmatrix}
		q_{n-1}r_{n-1}+\left(1-q_{n-1}r_{n-1}\right)S_{n-2}^{(q,r)}\\
		\noalign{\medskip}
		-r_{n-1}
		\end{bmatrix},
		\end{equation*}
		and that $P_{nl}^{(q,r)}=0 $ when $l$ is odd or  $l<0$.	

		\item[\text{\rm(i)}] For the Jost solution $\bar{\phi}_{n}^{(q,r)}$ we have the expansion
	\begin{equation*}
%\label{x5.10}
	z^{-n}\,\bar{\phi}_{n}^{(q,r)}=\sum_{l=0}^{\infty}\bar{P}_{nl}^{(q,r)}\ds\frac{1}{z^{l}}, \qquad |z|\ge 1,
	\end{equation*}
	with the double-indexed quantities $\bar{P}_{nl}^{(q,r)}$ for which we have
		\begin{equation*}
%\label{x5.11}
		\bar{P}_{n0}^{(q,r)}=E_{n-2}^{(q,r)}
		\begin{bmatrix}
		-q_{n-1}\\
		\noalign{\medskip}
		1
		\end{bmatrix},
		\end{equation*}
		\begin{equation*}
%\label{x5.12}
		\bar{P}_{n2}^{(q,r)}=E_{n-2}^{(q,r)}
		\begin{bmatrix}
		q_{n-1}-\ds\frac{q_{n-2}}{1+q_{n-2}r_{n-1}}-q_{n-1}\,Q_{n-3}^{(q,r)}\\
		\noalign{\medskip}
		Q_{n-3}^{(q,r)}
		\end{bmatrix},
		\end{equation*}
		and that $\bar{P}_{nl}^{(q,r)}=0 $ when $l$ is odd or  $l<0$.
		
		\item[\text{\rm(j)}] The scattering coefficients for \eqref{1.1} are even in $z$ in their respective domains. The domain for the reflection coefficients
		is the unit circle $\mathbb{T}$ and the domains for the transmission coefficients consist the union of $\mathbb{T}$ and their regions of extensions.
		
		\item[\text{\rm(k)}] The quantity $1/T^{(q,r)}$  has an extension from $z\in\mathbb{T}$ to $|z|<1$ and that extension is analytic for $|z|<1$ and continuous for $|z|\le 1.$ Similarly, the quantity $1/\bar{T}^{(q,r)}$ has an extension from $z\in\mathbb{T}$ so that it is analytic in $|z|>1$ and continuous in $|z|\ge 1.$
	\end{enumerate}
\end{theorem}

\begin{proof}
The proof is similar to the proof of Theorem~\ref{thm:theorem x2.1} and is obtained with the help of \eqref{1.1} and \eqref{x2.3}--\eqref{x2.6}.
\end{proof}

In the next theorem, at $z=1$ we present the values of the Jost solutions
to \eqref{1.1},
 \eqref{1.2aa},
  \eqref{1.2ab}
when the corresponding potential pairs $(q,r),$ $(u,v),$ $(p,s)$ are related to each other as in
\eqref{x3.1}--\eqref{x3.4}.
These results will be useful in the solution to the inverse problem for
 \eqref{1.1}.

\begin{theorem}
	\label{prop:proposition6.1}
Assume that the potentials $q_n$ and $r_n$ appearing in \eqref{1.1} are rapidly decaying and satisfy \eqref{1.1a}. Assume further that the potential pairs $(u,v)$ and $(p,s)$ are related to $(q,r)$ as in \eqref{x3.1}--\eqref{x3.4}. Then, at $z=1$ the Jost solutions
$\psi_n^{(q,r)}(1),$ $\bar{\psi}_n^{(q,r)}(1),$
$\psi_n^{(u,v)}(1),$ $\bar{\psi}_n^{(u,v)}(1),$ $\psi_n^{(p,s)}(1),$ and  $\bar{\psi}_n^{(p,s)}(1)$  have the values
\begin{equation}\label{6.1aa}
\begin{bmatrix}
\bar{\psi}_n^{(q,r)}(1)&
\psi_n^{(q,r)}(1)\end{bmatrix}=\begin{bmatrix}
1&0\\
\noalign{\medskip}
\ds\sum_{j=n}^{\infty} r_j&1
\end{bmatrix},
\end{equation}
\begin{equation}\label{6.1bb}
\begin{bmatrix}
\bar{\psi}_n^{(u,v)}(1)&
\psi_n^{(u,v)}(1)\end{bmatrix}=\begin{bmatrix}
\ds\frac{E_{n-1}^{(q,r)}}{E_\infty^{(q,r)}}&\ds\frac{E_{n-1}^{(q,r)}}{D_\infty^{(q,r)}}\,\sum_{j=n}^{\infty} q_j\\
\noalign{\medskip}
- r_n\,\ds\frac{D_{n-1}^{(q,r)}}{E_\infty^{(q,r)}}&\ds\frac{D_{n-1}^{(q,r)}}{D_\infty^{(q,r)}}\,\left(1-r_n\,\sum_{j=n}^{\infty} q_j\right)
\end{bmatrix},
\end{equation}
\begin{equation}\label{6.1cc}
\begin{bmatrix}
\bar{\psi}_n^{(p,s)}(1)&
\psi_n^{(p,s)}(1)\end{bmatrix}=\begin{bmatrix}
\ds\frac{E_{n-1}^{(q,r)}}{E_\infty^{(q,r)}}\,\left(1+q_n\,\sum_{j=n+1}^{\infty} r_j\right)&q_n\,\ds\frac{E_{n-1}^{(q,r)}}{D_\infty^{(q,r)}}\\
\noalign{\medskip}
\ds\frac{D_n^{(q,r)}}{E_\infty^{(q,r)}}\,\sum_{j=n+1}^{\infty} r_j&\ds\frac{D_n^{(q,r)}}{D_\infty^{(q,r)}}
\end{bmatrix},
\end{equation}
	where $D_n^{(q,r)},$ $D_\infty^{(q,r)},$ $E_n^{(q,r)},$ $E_\infty^{(q,r)},$are the quantities defined in \eqref{x2.42} and \eqref{x2.43}, respectively.
\end{theorem}

\begin{proof}
One can obtain \eqref{6.1aa} via iteration
by directly solving \eqref{1.1} with $z=1$
and using \eqref{x2.3} and \eqref{x2.5}.
Similarly, one can get \eqref{6.1bb} via iteration by solving \eqref{1.2aa} with $z=1$
and	using \eqref{x2.3}, \eqref{x2.5}, \eqref{x2.42}, \eqref{x2.43}, \eqref{x3.1},
and \eqref{x3.2}.
One can obtain \eqref{6.1cc} in a similar manner.
Alternatively, one can directly verify that
the two columns of \eqref{6.1aa} satisfy \eqref{1.1} with $z=1$
with the respective asymptotics in \eqref{x2.5} and \eqref{x2.3}.
Similarly,
with the help of \eqref{x2.42}, \eqref{x2.43}, \eqref{x3.1}, \eqref{x3.2},
one can directly verify that the
two columns given in \eqref{6.1bb} have the respective asymptotics in \eqref{x2.5} and \eqref{x2.3} and that they also
	satisfy \eqref{1.2aa} with $z=1.$ In a similar way, with the help of \eqref{x2.42}, \eqref{x2.43}, \eqref{x3.3}, and \eqref{x3.4}, one can directly verify
	that the two columns given in  \eqref{6.1cc} have the respective asymptotics in \eqref{x2.5} and \eqref{x2.3} and that they each
	satisfy \eqref{1.2ab} with $z=1.$
\end{proof}

We see that at $z=1$ the Jost solutions appearing on the left-hand sides of
\eqref{6.1aa},
\eqref{6.1bb},
\eqref{6.1cc}
can be expressed by using
\eqref{x2.11}, \eqref{x2.13},
 \eqref{x2.11aa}, \eqref{x2.13aa},
\eqref{x5.1}, \eqref{x5.4} as
\begin{equation}\label{6.1aaaa}
\begin{bmatrix}
\bar{\psi}_n^{(q,r)}(1)&
\psi_n^{(q,r)}(1)\end{bmatrix}=
\begin{bmatrix}
\ds\sum_{l=n}^{\infty} \bar{K}_{nl}^{(q,r)}&\ds\sum_{l=n}^{\infty} K_{nl}^{(q,r)}
\end{bmatrix},
\end{equation}
\begin{equation}\label{6.1bbbb}
\begin{bmatrix}
\bar{\psi}_n^{(u,v)}(1)&
\psi_n^{(u,v)}(1)\end{bmatrix}=
\begin{bmatrix}
\ds\sum_{l=n}^{\infty} \bar{K}_{nl}^{(u,v)}&\ds\sum_{l=n}^{\infty} K_{nl}^{(u,v)}
\end{bmatrix},
\end{equation}
\begin{equation}\label{6.1cccc}
\begin{bmatrix}
\bar{\psi}_n^{(p,s)}(1)&
\psi_n^{(p,s)}(1)\end{bmatrix}=
\begin{bmatrix}
\ds\sum_{l=n}^{\infty} \bar{K}_{nl}^{(p,s)}&\ds\sum_{l=n}^{\infty} K_{nl}^{(p,s)}
\end{bmatrix}.
\end{equation}

For a column vector $\mathbf K$ with two components let use $[\mathbf K]_1$ and $[\mathbf K]_2$ to denote the first and second components, respectively, i.e. we let
\begin{equation}\label{ta001}
\begin{bmatrix}K\end{bmatrix}_1:=\begin{bmatrix}1&0\end{bmatrix}\mathbf K,\quad
\begin{bmatrix}K\end{bmatrix}_2:=\begin{bmatrix}0&1\end{bmatrix}\mathbf K.
\end{equation}
In the next theorem we show how to recover the potentials $q_n$ and $r_n$ from
\eqref{6.1bbbb} and \eqref{6.1cccc}, respectively.

\begin{theorem}
	\label{prop:proposition6.1ff}
Assume that the potentials $q_n$ and $r_n$ appearing in \eqref{1.1} are rapidly decaying and satisfy \eqref{1.1a}. Assume further that the potential pairs $(u,v)$ and $(p,s)$ are related to $(q,r)$ as in \eqref{x3.1}--\eqref{x3.4}. Then,
$q_n$ and $r_n$ are related to
the Jost solutions for $(u,v)$ evaluated at $z=1$ given in
\eqref{6.1bbbb} as
\begin{equation}\label{6.6}
q_n=\ds\frac{D_\infty^{(q,r)}}{E_\infty^{(q,r)}}\,\left(\ds\frac{\begin{bmatrix}\psi_{n}^{(u,v)}
(1)\end{bmatrix}_1}{\begin{bmatrix}\bar{ \psi}_n^{(u,v)}(1)\end{bmatrix}_1}-\ds\frac{\begin{bmatrix}\psi_{n+1}^{(u,v)}(1)
\end{bmatrix}_1}{\begin{bmatrix}\bar{ \psi}_{n+1}^{(u,v)}(1)\end{bmatrix}_1}\right),
\end{equation}
\begin{equation}\label{6.6kk}
r_n=-\ds\frac{E_\infty^{(q,r)}}{D_\infty^{(q,r)}}\,\ds\frac{\begin{bmatrix}\bar{ \psi}_n^{(u,v)}(1)\end{bmatrix}_1\,
\begin{bmatrix}\bar{ \psi}_n^{(u,v)}(1)\end{bmatrix}_2
}{\begin{bmatrix}\bar{ \psi}_n^{(u,v)}(1)\end{bmatrix}_1\,\begin{bmatrix}\psi_{n}^{(u,v)}
(1)\end{bmatrix}_2
-\begin{bmatrix}\bar{ \psi}_n^{(u,v)}(1)\end{bmatrix}_2\,
\begin{bmatrix}\psi_{n}^{(u,v)}
(1)\end{bmatrix}_1}.
\end{equation}
Similarly, $q_n$ and
$r_n$ are related to
the Jost solutions for $(p,s)$ evaluated at $z=1$ given in
\eqref{6.1cccc} as
\begin{equation}\label{6.6hh}
q_n=\ds\frac{D_\infty^{(q,r)}}{E_\infty^{(q,r)}}\,\ds\frac{\begin{bmatrix}\psi_{n}^{(p,s)}
(1)\end{bmatrix}_1\,
\begin{bmatrix}\psi_{n}^{(p,s)}
(1)\end{bmatrix}_2
}{\begin{bmatrix}\bar{ \psi}_n^{(p,s)}(1)\end{bmatrix}_1\,\begin{bmatrix}\psi_{n}^{(p,s)}
(1)\end{bmatrix}_2
-\begin{bmatrix}\bar{ \psi}_n^{(p,s)}(1)\end{bmatrix}_2\,
\begin{bmatrix}\psi_{n}^{(p,s)}
(1)\end{bmatrix}_1},
\end{equation}
\begin{equation}\label{6.8}
r_n=\ds\frac{E_\infty^{(q,r)}}{D_\infty^{(q,r)}}\,\left(\ds\frac{\begin{bmatrix}	\bar{\psi}_{n-1}^{(p,s)}(1)\end{bmatrix}_2}
{\begin{bmatrix}\psi_{n-1}^{(p,s)}(1)\end{bmatrix}_2}
-\ds\frac{\begin{bmatrix}
\bar{\psi}_{n}^{(p,s)}(1)\end{bmatrix}_2}
{\begin{bmatrix}
\psi_{n}^{(p,s)}(1)\end{bmatrix}_2}\right).
\end{equation}
\end{theorem}

\begin{proof}
 From \eqref{6.1bb} we obtain
\begin{equation}\label{6.5}
\ds\frac{D_\infty^{(q,r)}}{E_\infty^{(q,r)}}\,\ds\frac{\begin{bmatrix}
\psi_{n}^{(u,v)}(1)\end{bmatrix}_1}
{\begin{bmatrix}
\bar{ \psi}_n^{(u,v)}(1)\end{bmatrix}_1}
=\sum_{j=n}^{\infty}q_j,
\end{equation}
which yields \eqref{6.6}. Then, using \eqref{6.5} in \eqref{6.1bb} we get
\eqref{6.6kk}.
Similarly, from \eqref{6.1cc} we get
\begin{equation}\label{6.7}
\ds\frac{E_\infty^{(q,r)}}{D_\infty^{(q,r)}}\,\ds\frac{\begin{bmatrix}
\bar{\psi}_{n}^{(p,s)}(1)\end{bmatrix}_2}
{\begin{bmatrix}
\psi_n^{(p,s)}(1)\end{bmatrix}_2}
=\sum_{j=n+1}^{\infty}r_j,
\end{equation}
which yields
\eqref{6.8}. Using \eqref{6.7} in \eqref{6.1cc} we get \eqref{6.6hh}.
\end{proof}

Let us remark that, as seen from \eqref{6.1aa}, we cannot determine $q_n$ from
either side of \eqref{6.1aaaa} even though we obtain $r_n$ as
$$r_n=\begin{bmatrix}
\bar{\psi}_n^{(q,r)}(1)\end{bmatrix}_2
-\begin{bmatrix}
\bar{\psi}_{n+1}^{(q,r)}(1)\end{bmatrix}_2.$$

\section{The bound states}
\label{sec:section4}

In this section we analyze the bound states for each of the three linear systems \eqref{1.1}, \eqref{1.2aa}, \eqref{1.2ab}, and we describe their bound-state data sets
in terms of the bound-state $z$-values, the multiplicity of each bound state, and the bound-state norming constants. We show how the bound-state norming constants are related to the
dependency constants and the transmission coefficients. Using a pair of constant matrix triplets
$(A,B,C)$ and $(\bar{A}, \bar{B}, \bar{C}),$ we describe in an elegant manner
each bound-state data set for any number of bound states with any multiplicities.
In the formulation of the
Marchenko method, we show how to relate the two matrix triplets
to the relevant Marchenko kernels in such a way that the
procedure is generally applicable for both continuous and discrete
linear systems. When the potential pairs $(q,r),$ $(u,v),$ and $(p,s)$ are related to
each other as in \eqref{x3.1}--\eqref{x3.4}, we describe how the corresponding
bound-state data sets are related to each other and also how the corresponding pairs of matrix triplets are related
to each other.

Let us first consider the bound states for \eqref{1.1}. By definition
a bound state for \eqref{1.1} corresponds to a square-summable solution in $n\in\mathbb{Z}$, i.e. a solution $\begin{bmatrix}
\alpha_{n}\\
\beta_n
\end{bmatrix}$ satisfying
\begin{equation}
\label{x4.1}
\sum_{n=-\infty}^{\infty}\left(|\alpha_n|^2+|\beta_n|^2\right)< +\infty.
\end{equation}
The bound states for \eqref{1.2aa} and \eqref{1.2ab} are defined in a similar way, i.e. for each of these two systems a bound state corresponds to a square-summable solution.

Let us introduce the dependency constants related to bound states for each of \eqref{1.1}, \eqref{1.2aa}, \eqref{1.2ab}. For each of these systems, at a bound state at $z=z_j$ the Jost solutions $\phi_{n}$ and $\psi_{n}$ are linearly dependent because $\phi_{n}(z_j)$ decays sufficiently fast as $n\to-\infty$ and $\psi_{n}(z_j)$ decays sufficiently fast as $n\to+\infty$ so that each of these solutions satisfy \eqref{x4.1}. Thus, a bound-state solution is a constant multiple of either of $\phi_{n}(z_j)$ and $\psi_{n}(z_j),$
and we can introduce the double-indexed dependency constant $\gamma_{j0}$ as the constant satisfying
\begin{equation}
\label{ta3001}
\phi_{n}(z_j)=\gamma_{j0}\,\psi_{n}(z_j),\qquad n\in\mathbb{Z}.
\end{equation}
As seen from any of the first equalities in \eqref{ta.5001}, \eqref{ta.5002},
and \eqref{ta.5003}, we observe that \eqref{ta3001} is equivalent to
the vanishing of the Wronskian determinant at $z=z_j$ for all $n\in\mathbb{Z},$ i.e.
\begin{equation*}
\begin{vmatrix}
\phi_{n}(z_j)&\psi_{n}(z_j)
\end{vmatrix}=0,\qquad n\in\mathbb{Z},
\end{equation*}
which is also equivalent to the linear dependence of
the Jost solutions $\phi_n$ and $\psi_n$ at $z=z_j.$

Similarly, at a bound state at $z=\bar{z}_j,$ the Jost solutions
$\bar{\phi}_n$ and $\bar{\psi}_n$ are linearly dependent
and for any of the systems \eqref{1.1}, \eqref{1.2aa}, \eqref{1.2ab},
this can be expressed in some equivalent forms such as
\begin{equation*}
\bar{\phi}_n(\bar{z}_j)=\bar{\gamma}_{j0}\,\bar{\psi}_n(\bar{z}_j),\qquad n\in\mathbb{Z},
\end{equation*}
where $\bar{\gamma}_{j0}$ is the double-indexed dependency constant, and also as
\begin{equation*}
\begin{vmatrix}
\bar{\phi}_{n}(\bar{z}_j)&\bar{\psi}_{n}(\bar{z}_j)
\end{vmatrix}=0,\qquad n\in\mathbb{Z},
\end{equation*}
indicating the linear dependence  of the
Jost solutions $\bar{\phi}_n$ and $\bar{\psi}_n$ at $z=\bar{z}_j.$

If a bound state is not simple, as seen from \eqref{ta.5006}, \eqref{ta.5007},
\eqref{ta.5010}, \eqref{ta.5011}, \eqref{ta.5018}, and \eqref{ta.5019},
the number of constraints is equivalent to the multiplicity of the bound state,
yielding as many dependency constants as the multiplicity of the bound state.
At a bound state at $z=z_j$ with multiplicity $m_j,$ by proceeding as in \cite{busse2008generalized,Ercan2018}, it follows that each of
\eqref{ta.5006}, \eqref{ta.5010},  \eqref{ta.5018} is equivalent to having $m_j$ constraints relating the Jost solutions $\phi_n$ and $\psi_n$ and their $z$-derivatives as

%In the presence of a bound state with multiplicities, we need as many dependency constants
%as the multiplicity of the bound state, and we generalize \eqref{ta3001} as follows.
%For each of \eqref{1.1}, \eqref{1.2aa}, and \eqref{1.2ab} let us introduce the scalar quantity $a_{\rm r}(z)$ in terms of the Wronskian of $\phi_{n}$ and $\psi_{n},$ i.e.
%\begin{equation}\label{x4.101}
%a_{\rm r}(z)=\begin{vmatrix}
%\phi_{n}&\psi_{n}
%\end{vmatrix}.
%\end{equation}
%At a bound state at $z=z_j$ with multiplicity $m_j$ we have
%\begin{equation}\label{x4.102}
%a_{\rm r}(z_j)=\ds\frac{d\,a_{\rm r}(z_j)}{dz}=\cdots=\ds\frac{d^{m_j-1}\,a_{\rm r}(z_j)}{dz^{m_j-1}}=0.
%\end{equation}
%From \eqref{x4.101} and \eqref{x4.102} it follows that $\phi_{n}$ and $\psi_{n}$ and their $z$-derivatives are related to each other as

\begin{equation}\label{x4.103}
\ds\frac{d^k \phi_{n}(z_j)}{dz^k}=\ds\sum_{l=0}^{k}\ds\binom{k}{l}\,\gamma_{j(k-l)}\,\ds\frac{d^l\,\psi_{n}(z_j)}{dz^l},\qquad 0\le k \le m_j-1,
\end{equation}
where $\binom{k}{l}$ is the binomial coefficient
and we refer to the double-indexed scalar quantities $\gamma_{jk}$ as the dependency constants at $z=z_j.$
%We remark that the binomial coefficient
%$\binom{k}{l}$ appearing in \eqref{x4.103} is defined as
%\begin{equation}\label{xxx.1}
%\ds\binom{k}{l}:=\ds\frac{k(k-1)(k-2)\cdots (k-l+1)}{l!}.
%\end{equation}
%Later we will use the definition in \eqref{xxx.1} also when $k$ is a negative
%integer, and hence we avoid using $k!$ on the right-hand side of \eqref{xxx.1}.
Note that \eqref{x4.103} holds for each of the systems \eqref{1.1}, \eqref{1.2aa}, and \eqref{1.2ab}. We can use the appropriate superscripts so that $\gamma_{jk}^{(q,r)},$ $\gamma_{jk}^{(u,v)},$ $\gamma_{jk}^{(p,s)}$ denote the corresponding dependency constants for \eqref{1.1}, \eqref{1.2aa}, \eqref{1.2ab}, respectively. In a similar way, we obtain the double-indexed dependency constants $\bar{\gamma}_{jk}$ at a bound state at $z=\bar{z}_j$ with multiplicity $\bar{m}_j,$ which relate the Jost solutions $\bar{\phi}_n$ and $\bar{\psi}_n$ and their $z$-derivatives as

%by introducing the Wronskian
%\begin{equation}\label{x4.105}
%\bar{a}_{\rm r}(z)=\begin{vmatrix}
%\bar{\phi}_{n}&\bar{\psi}_{n}
%\end{vmatrix}.
%\end{equation}
%At a bound state $z=\bar{z_j}$ we have the sufficiently fast decay for both $\bar{\phi}_{n}(\bar{z}_j)$ as $n\to-\infty$ and $\bar{\psi}_{n}(\bar{z}_j)$ as $n\to+\infty$ so that when $\bar{a}_{\rm r}(\bar{z}_j)=0,$ the two Jost solutions $\bar{\phi}_{n}(\bar{z}_j)$ and $\bar{\psi}_{n}(\bar{z}_j)$  are linearly dependent and hence a bound-state solution is obtained as a constant multiple of $\bar{\phi}_{n}(\bar{z}_j)$ or $\bar{\psi}_{n}(\bar{z}_j)$. In the case of multiplicity $\bar{m}_j$ we have
%\begin{equation}\label{x4.106}
%\bar{a}_{\rm r}(\bar{z}_j)=\ds\frac{d\,\bar{a}_{\rm r}(\bar{z}_j)}{dz}=\cdots=\ds\frac{d^{\bar{m}_j-1}\,\bar{a}_{\rm r}(\bar{z}_j)}{dz^{\bar{m}_j-1}}=0,
%\end{equation}
%and from \eqref{x4.105} and \eqref{x4.106} we get the dependency constant $\bar{\gamma}_{jk}$ via
\begin{equation}\label{x4.107}
\ds\frac{d^k\bar{\phi}_{n}(\bar{z}_j)}{dz^k}=\ds\sum_{l=0}^{k}\ds\binom{k}{l}\,\bar{\gamma}_{j(k-l)}\,\ds\frac{d^l\bar{\psi}_{n}(\bar{z}_j)}{dz^l},\qquad 0\le k \le \bar{m}_j-1.
\end{equation}

Let us also introduce the ``residues" $t_{jk}$ of the right transmission coefficients for each of \eqref{1.1}, \eqref{1.2aa}, and \eqref{1.2ab} when the corresponding right transmission coefficient $T_{\rm r}$ has a pole at $z=z_j$ of order $m_j$. Using the expansion
\begin{equation}\label{x4.8}
T_{\rm r}= \ds\frac{t_{jm_j}}{(z-z_j)^{m_j}}+\ds\frac{t_{j(m_j-1)}}{(z-z_j)^{m_j-1}}+\cdots+\ds\frac{t_{j1}}{(z-z_j)}+O\left(1\right),\qquad z\to z_j,
\end{equation}
we uniquely obtain the residues $t_{jk}$ for $1 \le k \le m_j$ and $1 \le j \le N$. We remark that $t_{jk}^{(q,r)},$  $t_{jk}^{(u,v)},$ $t_{jk}^{(p,s)}$ are defined as in \eqref{x4.8} by using the right transmission coefficients $T^{(q,r)},$ $T_{\rm r}^{(u,v)},$ $T_{\rm r}^{(p,s)}$ corresponding to \eqref{1.1}, \eqref{1.2aa}, \eqref{1.2ab}, respectively. In a similar way we define the ``residues" $\bar{t}_{jk}$ by letting
\begin{equation}\label{x4.9}
\bar{T}_{\rm r}= \ds\frac{\bar{t}_{j\bar{m}_j}}{(z-\bar{z}_j)^{\bar{m}_j}}+\ds\frac{\bar{t}_{j(\bar{m}_j-1)}}{(z-\bar{z}_j)^{\bar{m}_j-1}}+\cdots+\ds\frac{\bar{t}_{j1}}{(z-\bar{z}_j)}+O\left(1\right),\qquad z\to \bar{z}_j.
\end{equation}
Again using \eqref{x4.9} with $\bar{T}^{(q,r)},$ $\bar{T}_{\rm r}^{(u,v)},$ $\bar{T}_{\rm r}^{(p,s)}$ we obtain the residues $\bar{t}_{jk}^{(q,r)},$  $\bar{t}_{jk}^{(u,v)},$  $\bar{t}_{jk}^{(p,s)}$ corresponding to \eqref{1.1}, \eqref{1.2aa}, \eqref{1.2ab}, respectively.

In the next theorem we elaborate on the
bound states for \eqref{1.1}.

\begin{theorem}
\label{thm:theorem x4.1}
Assume that the potentials $q_n$ and $r_n$ appearing in \eqref{1.1} are rapidly decaying and satisfy \eqref{1.1a}. Then, we have the following:
\begin{enumerate}

	\item[\text{\rm(a)}] A bound state for \eqref{1.1} can only occur at a $z$-value for which $T^{(q,r)}$ has a pole in the region $0<|z|<1$ or $\bar{T}^{(q,r)}$ has a pole in the region $|z|>1$.

	\item[\text{\rm(b)}] The number of bound states is finite, i.e. the number of poles of $T^{(q,r)}$ in $0<|z|<1$ and the number of poles of $\bar{T}^{(q,r)}$ in $|z|>1$ each must be finite.
	
	\item[\text{\rm(c)}] A bound state is not necessarily simple, but its multiplicity
must be finite.

	\item[\text{\rm(d)}] Since each of the transmission coefficients $T^{(q,r)}$ and $\bar{T}^{(q,r)}$ are even in $z$ in their respective domains, the bound-state $z$-values are symmetrically located with respect to
the origin of the complex $z$-plane.

		\item[\text{\rm(e)}] At a bound state corresponding to a pole at
$z_j$ for $T^{(q,r)}$ with multiplicity $m_j,$ we have the two vectors
\begin{equation}\label{x4.2}
\begin{bmatrix}
\phi_n^{(q,r)}(z_j)\\\noalign{\medskip}\ds\frac{d\,\phi_{n}^{(q,r)}(z_j)}{dz}
\\\noalign{\medskip}\ds\frac{d^{2}\,\phi_{n}^{(q,r)}(z_j)}{dz^{2}}
\\\noalign{\medskip}\vdots\\
\noalign{\medskip}\ds\frac{d^{m_j-1}\,\phi_{n}^{(q,r)}(z_j)}{dz^{m_j-1}}
\end{bmatrix}, \quad
\begin{bmatrix}
\psi_n^{(q,r)}(z_j)\\
\noalign{\medskip}\ds\frac{d\,\psi_{n}^{(q,r)}(z_j)}{dz}\\
\noalign{\medskip}\ds\frac{d^{2}\,\psi_{n}^{(q,r)}(z_j)}{dz^{2}}\\
\noalign{\medskip}\vdots\\\noalign{\medskip}\ds\frac{d^{m_j-1}
\,\psi_{n}^{(q,r)}(z_j)}{dz^{m_j-1}}
\end{bmatrix},
\end{equation}
related to each other as in \eqref{x4.103} via $m_j$ dependency constants $\gamma_{jk}^{(q,r)}.$ Similarly, at the bound state at $z=\bar{z}_j$ corresponding to a pole of $\bar{T}^{(q,r)}$ in $|z|>1$, we have the two vectors
\begin{equation}\label{x4.3}
\begin{bmatrix}
\bar{\phi}_n^{(q,r)}(\bar{z}_j)\\
\noalign{\medskip}\ds\frac{d\,\bar{\phi}_{n}^{(q,r)}(\bar{z}_j)}{dz}\\
\noalign{\medskip}\ds\frac{d^{2}\,\bar{\phi}_{n}^{(q,r)}(\bar{z}_j)}{dz^{2}}\\
\noalign{\medskip}\vdots\\
\noalign{\medskip}\ds\frac{d^{\bar{m}_j-1}\,\bar{\phi}_{n}^{(q,r)}(\bar{z}_j)}
{dz^{\bar{m}_j-1}}
\end{bmatrix}, \quad
\begin{bmatrix}
\bar{\psi}_n^{(q,r)}(\bar{z}_j)\\
\noalign{\medskip}
\ds\frac{d\,\bar{\psi}_{n}^{(q,r)}(\bar{z}_j)}{dz}\\
\noalign{\medskip}\ds\frac{d^{2}\,\bar{\psi}_{n}^{(q,r)}(\bar{z}_j)}{dz^{2}}\\
\noalign{\medskip}\vdots\\\noalign{\medskip}
\ds\frac{d^{\bar{m}_j-1}\,\bar{\psi}_{n}^{(q,r)}(\bar{z}_j)}{dz^{\bar{m}_j-1}}
\end{bmatrix},
	\end{equation}
	related to each other as in \eqref{x4.107} via $\bar{m}_j$
dependency constants $\bar{\gamma}_{jk}^{(q,r)}.$  We recall that an overbar does not denote complex conjugation and that $\psi_{n}^{(q,r)}(z),$ $\phi_n^{(q,r)}(z),$
$\bar{\psi}_{n}^{(q,r)}(z),$ $\bar{\phi}_n^{(q,r)}(z)$
are the four Jost solutions to \eqref{1.1}.
\end{enumerate}
\end{theorem}

\begin{proof}
By Theorem~\ref{thm:theorem x3.4a} we know that $z^{-n}\,\psi_{n}^{(q,r)}(z)$ and $z^{n}\,\phi_n^{(q,r)}(z)$ have analytic extensions from $z\in \mathbb{T}$ to $|z|<1.$ Since a bound-state solution to \eqref{1.1} must satisfy \eqref{x4.1}, with the help of the first equality in \eqref{ta.5001} and \eqref{ta.5006} we prove that the bound states located in $|z|<1$ occur if and only if the two vectors listed in \eqref{x4.2} are related as stated in (e) and that relation occurs at a pole of $T^{(q,r)}.$ By Theorem~\ref{thm:theorem x3.4a}(j) we know that $T^{(q,r)}$ contains $z$ as $z^{2},$ and hence the bound states occur at the poles of $T^{(q,r)}$  at $z=\pm z_{j}$ for $1\le j\le N$  in $0<|z|<1,$ each having the multiplicity $m_j.$ The finiteness of $N$ and of $m_j$ is already known from Theorem~\ref{thm:theorem x2.6}(e). In a similar way,
with the help of the second equality in \eqref{ta.5001} and \eqref{ta.5007} 
 we show that the bound states of \eqref{1.1} in $|z|>1$ occur at $z=\pm \bar{z}_j,$ where the two vectors listed in \eqref{x4.3} are related to each other as stated in (e) and that $\bar{T}^{(q,r)}$ has a pole at $z=\pm \bar{z}_j$ with multiplicity $\bar{m}_j.$ The number of such $\bar{z}_j$-values denoted by $\bar{N}$ and each multiplicity $\bar{m}_j$ are both finite as a consequence of Theorem~\ref{thm:theorem x2.6}(e).
\end{proof}

In Theorem~\ref{thm:theorem x4.1} and its proof, the bound-state $z$-values and their multiplicities are described by the sets $\{ {\pm z}_{j},m_j\}_{j=1}^N$ and $\{ {\pm \bar{z}}_{j},\bar{m}_j\}_{j=1}^{\bar{N}}$ without using the superscript $(q,r).$ For clarity, one must use $z_j^{(q,r)},$ $m_j^{(q,r)},$ $N^{(q,r)},$ $\bar{z}_j^{(q,r)},$ $\bar{m}_j^{(q,r)},$ $\bar{N}^{(q,r)}$ for \eqref{1.1} and similar notations to describe the bound states for \eqref{1.2aa} and \eqref{1.2ab}. Then, the bound states for \eqref{1.2aa} and \eqref{1.2ab} can be described by the corresponding version of Theorem~\ref{thm:theorem x4.1}.

Let us also remark that from \eqref{x2.26}  and the analog of \eqref{x2.26} for $(p,s),$ we conclude that the bound states for \eqref{1.2aa} and \eqref{1.2ab} can equivalently be described as in Theorem~\ref{thm:theorem x4.1} by using either the left transmission coefficients or the right transmission coefficients. If the three potential pairs $(q,r),$  $(u,v),$  $(p,s)$ are related to each other as in \eqref{x3.1}--\eqref{x3.4} then from Theorem~\ref{thm:theorem x3.3} it follows that
\begin{equation}\label{x4.4}
\begin{cases}
T^{(q,r)}=E_\infty^{(q,r)}\,T_{\rm r}^{(u,v)}=E_\infty^{(q,r)}\,T_{\rm r}^{(p,s)}=\ds\frac{1}{D_\infty^{(q,r)}}\,T_{\rm l}^{(u,v)}=\ds\frac{1}{D_\infty^{(q,r)}}\,T_{\rm l}^{(p,s)},\\
\noalign{\medskip}
\bar{T}^{(q,r)}=D_\infty^{(q,r)}\,\bar{T}_{\rm r}^{(u,v)}=D_\infty^{(q,r)}\,\bar{T}_{\rm r}^{(p,s)}=\ds\frac{1}{E_\infty^{(q,r)}}\,\bar{T}_{\rm l}^{(u,v)}=\ds\frac{1}{E_\infty^{(q,r)}}\,\bar{T}_{\rm l}^{(p,s)},
\end{cases}
\end{equation}
and hence the sets $\{ {\pm z}_{j},m_j\}_{j=1}^N$ and $\{ {\pm \bar{z}}_{j},\bar{m}_j\}_{j=1}^{\bar{N}}$ refer to the common sets of bound states and the corresponding multiplicities for \eqref{1.1}, \eqref{1.2aa}, \eqref{1.2ab}.
In that case, from \eqref{x4.4} it follows that the residues corresponding to \eqref{1.1}, \eqref{1.2aa}, \eqref{1.2ab} are related to each other as
\begin{equation}\label{x4.5}
\begin{cases}
t_{jk}^{(q,r)}=E_\infty^{(q,r)}\,t_{jk}^{(u,v)}=E_\infty^{(q,r)}\,t_{jk}^{(p,s)},\\
\noalign{\medskip}
\bar{t}_{jk}^{(q,r)}=D_\infty^{(q,r)}\,\bar{t}_{jk}^{(u,v)}=D_\infty^{(q,r)}\,\bar{t}_{jk}^{(p,s)}.
\end{cases}
\end{equation}

In the next theorem, when the potential pairs $(q,r),$ $(u,v)$, $(p,s)$ are related to each other as in \eqref{x3.1}--\eqref{x3.4}, we present the relationships among the corresponding dependency constants.
\begin{theorem}
\label{thm:theorem x4.2}
Assume that the potentials $q_n$ and $r_n$ appearing in \eqref{1.1} are rapidly decaying and satisfy \eqref{1.1a}. Assume further that the potential pairs $(u,v)$ and $(p,s)$ are related to $(q,r)$ as in \eqref{x3.1}--\eqref{x3.4}. Then, the corresponding dependency constants $\gamma_{jk}^{(q,r)},$  $\gamma_{jk}^{(u,v)},$ $\gamma_{jk}^{(p,s)}$ are related to each other for $0 \le k \le m_j-1$ and $1 \le j \le N$ as
\begin{equation}\label{x4.6}
\begin{cases}
D_\infty^{(q,r)}\,\gamma_{jk}^{(q,r)}=\gamma_{jk}^{(p,s)},\\
\noalign{\medskip}
D_\infty^{(q,r)}\,\gamma_{jk}^{(q,r)}=\ds\sum_{l=0}^{k}\ds\binom{k}{l}\,
\ds\frac{d^l\,\sigma(z_j)}{dz^l}\,\gamma_{j(k-l)}^{(u,v)},
\end{cases}
\end{equation}
where we have defined
\begin{equation}\label{x4.6a}
\sigma(z):=\ds\frac{1}{1-\ds\frac{1}{z^2}}.
\end{equation} Similarly, the corresponding dependency constants $\bar{\gamma}_{jk}^{(q,r)},$  $\bar{\gamma}_{jk}^{(u,v)},$ $\bar{\gamma}_{jk}^{(p,s)}$  are related to each other as
\begin{equation}\label{x4.7}
\begin{cases}
\bar{\gamma}_{jk}^{(p,s)}=E_\infty^{(q,r)}\,\bar{\gamma}_{jk}^{(q,r)},\\
\noalign{\medskip}
\bar{\gamma}_{jk}^{(u,v)}=E_\infty^{(q,r)}\,\ds\sum_{l=0}^{k}\ds\binom{k}{l}\,
\ds\frac{d^l\,\sigma(\bar{z}_j)}{dz^l}\,\bar{\gamma}_{j(k-l)}^{(q,r)}.
\end{cases}
\end{equation}
\end{theorem}

\begin{proof}
Using \eqref{x3.5a} and \eqref{x3.7a} in the Wronskian determinant on the right-hand side of the first equality in \eqref{ta.5001} we get
\begin{equation}\label{x4.109}
\begin{vmatrix}
\phi_{n}^{(q,r)}&\psi_{n}^{(q,r)}
\end{vmatrix}=D_\infty^{(q,r)}\left(\det\begin{bmatrix}
\Lambda_n^{(q,r)}
\end{bmatrix}\right)\begin{vmatrix}
\phi_{n}^{(p,s)}&\psi_{n}^{(p,s)}
\end{vmatrix},
\end{equation}
where $\Lambda_n^{(q,r)}$ is the coefficient matrix appearing in \eqref{x3.5a}
and \eqref{x3.7a}, i.e.
\begin{equation}\label{x4.110}
\Lambda_n^{(q,r)}:=\begin{bmatrix}
\ds\frac{1}{E_{n-1}^{(q,r)}} &-\ds\frac{q_n}{D_n^{(q,r)}}\\
\noalign{\medskip}
\ds\frac{r_n}{E_{n-1}^{(q,r)}} &\ds\frac{1}{D_{n-1}^{(q,r)}}
\end{bmatrix}.
\end{equation}
 From \eqref{x2.42} and \eqref{x4.110} we see that the determinant of $\Lambda_n^{(q,r)}$ is given by
\begin{equation}\label{x4.111}
\det\begin{bmatrix}
\Lambda_n^{(q,r)}
\end{bmatrix}=\ds\frac{1}{E_{n-1}^{(q,r)}\,D_n^{(q,r)}},
\end{equation}
and hence $\det[
\Lambda_n^{(q,r)}
]\ne 0$ for any integer $n.$ Using \eqref{x4.109}, with the help of the first equality in \eqref{ta.5003}, we obtain
\begin{equation}\label{x4.112}
a_n^{(q,r)}(z)=D_\infty^{(q,r)}\left(\det\begin{bmatrix}
\Lambda_n^{(q,r)}
\end{bmatrix}\right)a_n^{(p,s)}(z), \qquad n\in\mathbb{Z}.
\end{equation}
From \eqref{x4.112} we conclude that \eqref{ta.5006} for the potential pair $(q,r)$ occurs if and only if \eqref{ta.5018} for the potential pair $(p,s)$ occurs. Comparing \eqref{x4.103} for $(q,r)$ with \eqref{x4.103} for $(p,s)$, with the help of \eqref{x3.5a} and \eqref{x3.7a} and the fact the matrix $\Lambda_n^{(q,r)}$ defined in \eqref{x4.110} is invertible, we establish the equality in the first line of \eqref{x4.6}. In a similar way, with the help of \eqref{ta.5007},  \eqref{ta.5019}, and \eqref{x4.107} written for the pairs $(q,r)$ and $(p,s)$ and also using \eqref{x3.6a} and \eqref{x3.8a}  we obtain the equality in the first line of \eqref{x4.7}. The equality in the second line of \eqref{x4.6} is established in a similar manner by using \eqref{ta.5006}, \eqref{ta.5010},  and \eqref{x4.103} written for the pairs $(q,r)$ and $(u,v)$ and also using \eqref{x3.5} and \eqref{x3.7}. The equality in the second line of \eqref{x4.7} is established in a similar manner by using \eqref{ta.5007}, \eqref{ta.5011}, and \eqref{x4.107} written for the pairs $(q,r)$ and $(u,v)$ and also using \eqref{x3.6} and \eqref{x3.8}.
\end{proof}

As expected, for a unique solution to an inverse problem, for each bound state we need to specify a corresponding bound-state norming constant. If a bound state has a multiplicity then we must specify a separate norming constant for each multiplicity. In the case of \eqref{1.1}, \eqref{1.2aa}, and \eqref{1.2ab}, because the bound-state locations occur symmetrically with respect to the origin of the complex $z$-plane, we mention that the bound-state norming constants for those symmetric pairs coincide.

 As a summary, in specifying the bound-state data sets for each of \eqref{1.1}, \eqref{1.2aa}, \eqref{1.2ab}, in addition to providing $\{ {\pm z}_{j},m_j\}_{j=1}^N$ and $\{ {\pm \bar{z}}_{j},\bar{m}_j\}_{j=1}^{\bar{N}}$ we also need to provide the sets of bound-state norming constants $\{\{c_{jk}\}_{k=0}^{m_j-1}\}_{j=1}^{N}$ and $\{\{\bar{c}_{jk}\}_{k=0}^{\bar{m}_j-1}\}_{j=1}^{\bar{N}},$ where the double-indexed quantities $c_{jk}$ and $\bar{c}_{jk}$ denote the norming constants associated with $z_j$ and $\bar{z}_j,$ respectively. Clearly, we must use
  $c_{jk}^{(q,r)}$ and
  $\bar{c}_{jk}^{(q,r)}$ for the norming constants for \eqref{1.1},
  use
  $c_{jk}^{(u,v)}$ and
  $\bar{c}_{jk}^{(u,v)}$
   for the norming constants for \eqref{1.2aa}, and use
  $c_{jk}^{(p,s)}$ and
  $\bar{c}_{jk}^{(p,s)}$
   for the norming constants for \eqref{1.2ab}.
 In the presence of multiplicities it becomes extremely complicated to deal with bound states. That is why in the literature most researchers make the artificial assumption that the bound states are simple.

 The bound states with multiplicities can easily and in an elegant way be handled
 \cite{aktosunSymmetries,aktosun2007exact,AE19, Aktosunkdv,busse2008generalized,Busse,Ercan2018}
 for both continuous and discrete systems by using an appropriate constant matrix triplet $(A,B,C)$ for a KdV-like system or a pair of triplets $(A,B,C)$ and $(\bar{A},\bar{B},\bar{C})$ for an NLS-like system. Let us mention that the potentials appear in the block-diagonal format in the linear system in the KdV-like case and the potentials appear in the block off-diagonal format in the linear system in the NLS-like case. In all these cases, the relevant tool to solve inverse scattering problems is the Marchenko method. In the continuous case the Marchenko method involves a linear integral  equation known as the Marchenko equation or a system of linear integral equations to which we refer as the Marchenko system. In the discrete case the integrals in the Marchenko equations or in the Marchenko systems are simply replaced by the corresponding summations. In either the continuous case or the discrete case, the matrix triplet $(A,B,C)$ in the KdV-like case or the  triplets $(A,B,C)$ and $(\bar{A},\bar{B},\bar{C})$ in the NLS-like case are chosen in such way that the part of the kernel of the Marchenko system related to the bound states is expressed in a simple manner in terms of such matrix triplets.

 In this paper we deal with NLS-like discrete systems, and hence we use the pair of matrix triplets $(A,B,C)$ and $(\bar{A},\bar{B},\bar{C})$. If there is a bound state
 at $z=z_j$ with multiplicity $m_j$
 for $1\le j\le N,$ then
the triplet $(A,B,C)$ can be chosen as
	\begin{equation}
\label{TA.1}
A:=\begin{bmatrix}
A_1&0&\cdots&0\\
0&A_2&\cdots&0\\
\vdots&\vdots&\ddots&\vdots\\
0&0&\cdots&A_N
\end{bmatrix}, \quad B:=\begin{bmatrix}
B_1\\
B_2\\
\vdots\\
B_N
\end{bmatrix},\quad C:=\begin{bmatrix}
C_1&C_2&\cdots&C_N
\end{bmatrix},
\end{equation}
in such a way that $A$ is a block-diagonal matrix, $B$ is a block column vector, and $C$ is a block row vector with
	\begin{equation}\label{A1}
A_j:=\begin{bmatrix}
z_j&1&0&\cdots&0&0\\
0&z_j&1&\cdots&0&0\\
0&0&z_j&\cdots&0&0\\
\vdots&\vdots&\vdots&\ddots&\vdots&\vdots\\
0&0&0&\cdots&z_j&1\\
0&0&0&\dots&0&z_j
\end{bmatrix},\quad
B_j:=\begin{bmatrix}
0\\ \vdots \\
0\\
1
\end{bmatrix},
\end{equation}
\begin{equation}\label{A2}
C_j:=\begin{bmatrix}
c_{j(m_j-1)}&c_{j(m_j-2)}&\cdots&c_{j1}&c_{j0}
\end{bmatrix}.
\end{equation}
	As seen from \eqref{A1}, $A_j$ is an $m_j\times m_j$ matrix in the Jordan canonical form with $z_j$ appearing in the diagonal entries, and $B_j$ is an $m_j\times 1$ matrix with $0$ in the first $(m_j-1)$ entries and $1$ in the $m_j$th entry. As also seen from \eqref{A2} the $1\times m_j$ matrix $C_j$ is constructed from the norming constants $c_{jk}$.
In our paper, the matrix triplet $(A_j,B_j,C_j)$ is chosen to include the contribution
 from both $z=z_j$ and $z=-z_j,$
 and this will be seen from \eqref{R.1} and Theorem~\ref{thm:theorem x4.3ab}(d).

In a similar way, for the bound states at $z=\pm \bar{z}_j$
 with multiplicity $\bar{m}_j$ for $1\le j\le \bar{N},$ the corresponding
triplet $(\bar{A},\bar{B},\bar{C})$ can be chosen as	
\begin{equation}
\label{TA.2}
\bar{A}:=\begin{bmatrix}
\bar{A}_1&0&\cdots&0\\
0&\bar{A}_2&\cdots&0\\
\vdots&\vdots&\ddots&\vdots\\
0&0&\cdots&\bar{A}_{\bar{N}}
\end{bmatrix}, \quad \bar{B}:=\begin{bmatrix}
\bar{B}_1\\
\bar{B}_2\\
\vdots\\
\bar{B}_{\bar{N}}
\end{bmatrix},\quad \bar{C}:=\begin{bmatrix}
\bar{C}_1&\bar{C}_2&\cdots&\bar{C}_{\bar{N}}
\end{bmatrix},
\end{equation}
in such a way that $\bar{A}$ is a block-diagonal matrix, $\bar{B}$ is a block column vector, and $\bar{C}$ is a block row vector with
\begin{equation}\label{A3}
\bar{A}_j:=\begin{bmatrix}
\bar{z}_j&1&0&\cdots&0&0\\
0&\bar{z}_j&1&\cdots&0&0\\
0&0&\bar{z}_j&\cdots&0&0\\
\vdots&\vdots&\vdots&\ddots&\vdots&\vdots\\
0&0&0&\cdots&\bar{z}_j&1\\
0&0&0&\cdots&0&\bar{z}_j
\end{bmatrix}, \quad
\bar{B}_j:=\begin{bmatrix}
0\\ \vdots \\
0\\
1
\end{bmatrix},
\end{equation}
\begin{equation}\label{A4}
\bar{C}_j:=\begin{bmatrix}	\bar{c}_{j(\bar{m}_j-1)}&\bar{c}_{j(\bar{m}_j-2)}&\cdots&\bar{c}_{j1}&\bar{c}_{j0}
\end{bmatrix}.
\end{equation}
As seen from \eqref{A3}, $\bar{A}_j$ is an $\bar{m}_j\times \bar{m}_j$ matrix in the Jordan canonical form with $\bar{z}_j$ appearing in the diagonal entries, and $\bar{B}_j$ is an $m_j\times 1$ matrix with $0$ in the first $\bar{m}_j-1$ entries and $1$ in the $\bar{m}_j$th entry. As also seen from \eqref{A4} the $1\times \bar{m}_j$ matrix $\bar{C}_j$ is constructed from the norming constants $\bar{c}_{jk}.$
In our paper, the matrix triplet $(\bar{A}_j,\bar{B}_j,\bar{C}_j)$ is chosen to
include the contribution
 from both $z=\bar{z}_j$ and $z=-\bar{z}_j,$ and this will be seen from \eqref{R.2} and
Theorem~\ref{thm:theorem x4.3ab}(d).

The Marchenko system associated with either of \eqref{1.2aa} and \eqref{1.2ab}
is given by
\begin{equation}\label{Tx.1}
\begin{split}
\phantom{x}&
\begin{bmatrix}
\bar{K}_{nm}&K_{nm}
\end{bmatrix}+\begin{bmatrix}
0&\bar{\Omega}_{n+m}\\
\noalign{\medskip}
\Omega_{n+m}&0
\end{bmatrix}\\
&\phantom{xxxxxx}+\sum_{l=n+1}^{\infty}\begin{bmatrix}
\bar{K}_{nl}&K_{nl}
\end{bmatrix}\begin{bmatrix}
0&\bar{\Omega}_{l+m}\\
\noalign{\medskip}
\Omega_{l+m}&0
\end{bmatrix}=\begin{bmatrix}
0&0\\
\noalign{\medskip}
0&0
\end{bmatrix},\qquad m>n,
\end{split}
\end{equation}
where we have defined
\begin{equation}\label{Tx.2}
K_{nm}:=\ds\frac{1}{2\pi i}\ds\oint dz\,\psi_{n}\,z^{-m-1},\quad \bar{K}_{nm}:=\ds\frac{1}{2\pi i}\ds\oint dz\,\bar{\psi}_{n}\,z^{m-1},
\end{equation}
\begin{equation}\label{Tx.3}
\begin{cases}
\Omega_{k}:=\hat{R}_{k}+CA^{k-1}B,\quad \bar{\Omega}_{k}:=\hat{\bar{R}}_{k}+\bar{C}(\bar{A})^{-k-1}\bar{B},\qquad
k \text{ \rm{even}},
\\
\noalign{\medskip}
\Omega_{k}:=0,\quad \bar{\Omega}_{k}:=0,\qquad
k \text{ \rm{odd}},
\end{cases}
\end{equation}
with
\begin{equation}\label{Tx.4}
\hat{R}_{k}:=\ds\frac{1}{2\pi i}\ds\oint dz\,R\,z^{k-1},\quad \hat{\bar{R}}_{k}:=\ds\frac{1}{2\pi i}\ds\oint dz\,\bar{R}\,z^{-k-1}.
\end{equation}
We remark that $\psi_n$ and $\bar{\psi}_n$ appearing in \eqref{Tx.2} are the Jost solutions satisfying \eqref{x2.3} and \eqref{x2.5}, respectively, and that the integral in \eqref{Tx.2} denoted by $\oint$ is the contour integral along the unit circle $\mathbb{T}$ in the positive direction. In fact, for the potential pair $(u,v)$ the quantities $K_{nm}$ and $\bar{K}_{nm}$ are the column vectors appearing in \eqref{x2.11} and \eqref{x2.13}, respectively. The scalar quantities $R$ and $\bar{R}$ appearing in \eqref{Tx.4} are the right reflection coefficients, and the matrix triplets $(A,B,C)$ and $(\bar{A},\bar{B},\bar{C})$ appearing in \eqref{Tx.3} are those described in \eqref{TA.1} and \eqref{TA.2}, respectively.

Let us also remark that
$K_{nm}=0$ and $\bar{K}_{nm}=0$ when $n+m$ is odd, and this is already
stated in Theorem~\ref{thm:theorem x2.1},
Corollary~\ref{thm:theorem x2.1a},
and
Theorem~\ref{thm:theorem x3.4a} for the potential pairs
$(u,v),$ $(p,s),$ and $(q,r),$ respectively.
Similarly, we already know that the scattering coefficients
are even in $z$ for each of these three potential pairs.
Hence, from \eqref{Tx.4} we see that
$\hat{R}_{k}=0$ and
$\hat{\bar{R}}_{k}=0$ when $k$ is odd. Thus,
the second line of \eqref{Tx.3} is consistent with
\eqref{Tx.1} and \eqref{Tx.4}.

The derivation of \eqref{Tx.1} is obtained as follows. We can express the Jost solutions $\phi_n$ and $\bar{\phi}_n$ satisfying \eqref{x2.4} and \eqref{x2.6}, respectively, as linear combinations of $\psi_n$ and $\bar{\psi}_n$ as
\begin{equation}\label{Tx.5}
\begin{cases}
\phi_n\,T_{\rm r}=\bar{ \psi}_n+\psi_n\,R,\\
\noalign{\medskip}
\bar{\phi}_n\,\bar{T}_{\rm r}=\psi_n+\bar{\psi}_n\,\bar{R},
\end{cases}
\end{equation}
where 	$T_{\rm r}$ and $\bar{T}_{\rm r}$ are the right transmission coefficients appearing in \eqref{x2.8} and \eqref{x2.9}, respectively. We use the Fourier transform on \eqref{Tx.5}, and for $m>n$ we get
\begin{equation}\label{Tx.6}
\ds\frac{1}{2\pi i}\ds\oint dz\,\phi_{n}\,T_{\rm r}\,z^{m-1}=\ds\frac{1}{2\pi i}\ds\oint dz\,\bar{\psi}_{n}\,z^{m-1}+\ds\frac{1}{2\pi i}\ds\oint dz\,\psi_{n}\,R\,z^{m-1},
\end{equation}
\begin{equation}\label{Tx.7}
\ds\frac{1}{2\pi i}\ds\oint dz\,\bar{\phi}_{n}\,\bar{T}_{\rm r}\,z^{-m-1}=\ds\frac{1}{2\pi i}\ds\oint dz\,\psi_{n}\,z^{-m-1}+\ds\frac{1}{2\pi i}\ds\oint dz\,\bar{\psi}_{n}\,\bar{R}\,z^{-m-1},
\end{equation}
yielding the two columns of \eqref{Tx.1}.

Using the notation of \eqref{ta001}, from \eqref{Tx.1} we get the two uncoupled
scalar equations for $m>n$ as
\begin{equation}\label{ta002}
\begin{cases}
\begin{bmatrix} K_{nm}\end{bmatrix}_1
+\bar{\Omega}_{n+m}-
\ds\sum_{l=n+1}^{\infty}
\ds\sum_{j=n+1}^{\infty}\begin{bmatrix} K_{nj}\end{bmatrix}_1\,
\Omega_{j+l}\,\bar{\Omega}_{l+m}=0,\\
\noalign{\medskip}
\begin{bmatrix}\bar{K}_{nm}\end{bmatrix}_2
+\Omega_{n+m}-
\ds\sum_{l=n+1}^{\infty}
\ds\sum_{j=n+1}^{\infty}\begin{bmatrix}
\bar{K}_{nj}\end{bmatrix}_2\,\bar{\Omega}_{j+l}
\,
\Omega_{l+m}=0,
\end{cases}
\end{equation}
and once the system \eqref{ta002} is solved we also have
\begin{equation}\label{ta1001}
\begin{cases}
\begin{bmatrix}\bar{K}_{nm}\end{bmatrix}_1
=-\ds\sum_{l=n+1}^{\infty}
\begin{bmatrix}
K_{nl}\end{bmatrix}_1\,
\Omega_{l+m}
,\\
\noalign{\medskip}
\begin{bmatrix} K_{nm}\end{bmatrix}_2
=-\ds\sum_{l=n+1}^{\infty}
\begin{bmatrix} \bar{K}_{nl}\end{bmatrix}_2\,
\bar{\Omega}_{l+m}
.
\end{cases}
\end{equation}
Let us recall that
$K_{nm}=0$ and $\bar{K}_{nm}=0$ when $n+m$ is odd, and hence
the lower indices for the summations in
\eqref{ta002} and \eqref{ta1001} actually start with $n+2$ instead of $n+1.$
Nevertheless, we use $n+1$ there instead of
$n+2$ so that \eqref{ta002} and \eqref{ta1001} appear in the standard form as
a generic Marchenko system in the discrete case.
When we use \eqref{ta002} corresponding to \eqref{1.2aa}, we recover the potentials $u_n$ and $v_n$ as
\begin{equation}\label{Tx.8}
u_n=\begin{bmatrix}
K_{n(n+2)}^{(u,v)}
\end{bmatrix}_1,\quad v_n=\begin{bmatrix}
\bar{K}_{n(n+2)}^{(u,v)}
\end{bmatrix}_2,
\end{equation}
which are compatible with \eqref{x2.12} and \eqref{x2.14}, respectively. In the same manner, if we use \eqref{ta002} corresponding to \eqref{1.2ab}, we recover the potentials $p_n$ and $s_n$ as
\begin{equation}\label{Tx.9}
p_n=\begin{bmatrix}
K_{n(n+2)}^{(p,s)}
\end{bmatrix}_1,\quad s_n=\begin{bmatrix}
\bar{K}_{n(n+2)}^{(p,s)}
\end{bmatrix}_2,
\end{equation}
which are compatible with \eqref{x2.12aa} and \eqref{x2.14aa}, respectively.

Next, we describe the construction of the norming constants $c_{jk}$ and
$\bar{c}_{jk}$
in terms of the residues $t_{jk}$ and $\bar{t}_{jk}$
and the dependency constants $\gamma_{jk}$ and $\bar{\gamma}_{jk}.$

 \begin{theorem}
 	\label{thm:theorem x4.3}
 	Assume that the potentials $u_n$ and $v_n$ appearing in \eqref{1.2aa} are rapidly decaying and $1-u_n v_n\ne 0$ for $n\in\mathbb{Z}$. Let us use $\{ {\pm z}_{j},m_j\}_{j=1}^N$ and $\{ {\pm \bar{z}}_{j},\bar{m}_j\}_{j=1}^{\bar{N}}$ to denote the corresponding sets for the bound-state locations and their multiplicities. Then, the norming constants $c_{jk}^{(u,v)}$ appearing in \eqref{A2} are related to the residues $t_{jk}^{(u,v)}$ appearing in \eqref{x4.8} and
 the dependency constants $\gamma_{jk}^{(u,v)}$ appearing in \eqref{x4.103}  as
\begin{equation}\label{A5aa}
 c_{jk}^{(u,v)}=-2\sum_{l=0}^{m_j-1-k}t_{j(k+1+l)}^{(u,v)}
 \,\frac{\gamma_{jl}^{(u,v)}}{l!},\qquad 1\le j\le N,\quad 0\le k\le m_j-1.
 \end{equation}
 	Similarly, the norming constants $\bar{c}_{jk}^{(u,v)}$ appearing in \eqref{A4} are related to the residues $\bar{t}_{jk}^{(u,v)}$ appearing in \eqref{x4.9} and the dependency constants $\bar{\gamma}_{jk}^{(u,v)}$ appearing in \eqref{x4.107} as
 		\begin{equation}\label{A6aa}
 \bar{c}_{jk}^{(u,v)}=2\sum_{l=0}^{\bar{m}_j-1-k}\bar{t}_{j(k+1+l)}^{(u,v)}
 \,\frac{\bar{\gamma}_{jl}^{(u,v)}}{l!},\qquad 1\le j\le \bar{N},\quad 0\le k\le \bar{m}_j-1.
 	\end{equation}
\end{theorem}

\begin{proof}
For notational simplicity,
we outline the proof without using the superscript $(u,v)$ on the relevant quantities. As seen from \eqref{Tx.3} the contribution to the Marchenko kernels $\Omega_k$ and $\bar{\Omega}_k$ from the bound states are given by $CA^{k-1}B$ and $\bar{C}(\bar{A})^{-k-1}\bar{B}$, respectively. Thus, the contribution to the Marchenko kernel $\Omega_k$ from the bound state at $z=z_j$ is $C_j A_j^{k-1}B_j/2$
and the contribution to the Marchenko kernel $\bar{\Omega}_k$ from the bound state at $z=\bar{z}_j$ is $\bar{C_j}(\bar{A_j})^{-k-1}\bar{B_j}/2.$ This indicates that the contribution from the pole at $z=z_j$ of $T_{\rm r}$ to the left-hand side of \eqref{Tx.6} is evaluated as
\begin{equation}\label{R.1}
\ds\frac{1}{2\pi i}\ds\oint dz\,\phi_{n}\,T_{\rm r}\,z^{m-1}=-\ds\frac{1}{2}\,\sum_{l=n}^{\infty}K_{nl}\,C_{j}\,A_{j}^{l+m-1}\,B_{j}.
\end{equation}
Similarly, the contribution to the left-hand side of \eqref{Tx.7} from the pole at $z=\bar{z}_j$ of $\bar{T}_{\rm r}$ is evaluated as
\begin{equation}\label{R.2}
\ds\frac{1}{2\pi i}\ds\oint dz\,\bar{\phi}_{n}\,\bar{T}_{\rm r}\,z^{-m-1}=-\ds\frac{1}{2}\,\sum_{l=n}^{\infty}\bar{K}_{nl}\,\bar{C}_{j}\,(\bar{A}_{j})^{-l-m-1}\,\bar{B}_{j}.
\end{equation}
Using \eqref{x4.8} on the left-hand side of \eqref{R.1} we evaluate the aforementioned contribution as
\begin{equation}\label{R.3}
\ds\frac{1}{2\pi i}\ds\oint dz\,\phi_{n}\,T_{\rm r}\,z^{m-1}=\ds\sum_{k=0}^{m_j-1}\ds\frac{t_{jk}}{k!}\,\ds\frac{d^k(\phi_{n}\,z^{m-1})}{dz^k}\bigg |_{z=z_j}.
\end{equation}
Using \eqref{x4.103} on the right-hand side of \eqref{R.3} we write that right-hand side in terms of the residues $t_{jk}$, the dependency constants $\gamma_{jk}$, and $d^{k}\psi_n(z_j)/dz^{k}.$ Finally, we write the expansion for $d^{k}\psi_n(z_j)/dz^{k}$ in terms of the double-indexed quantities $K_{nl}$ appearing in \eqref{x2.11}. By equating the result to the right-hand side of \eqref{R.1}, we establish \eqref{A5aa}. We establish \eqref{A6aa} in a similar manner by evaluating the left-hand side of \eqref{R.2} with the help of \eqref{x4.9} and then by using \eqref{x4.107} and also by using \eqref{x2.13}.
\end{proof}

As the next corollary indicates, the result of Theorem~\ref{thm:theorem x4.3} also holds for the potential pair $(p,s)$ appearing in \eqref{1.2ab}. A proof is omitted because it is similar to the proof of Theorem~\ref{thm:theorem x4.3}.

\begin{corollary}
	\label{thm:theorem x4.3a}
		Assume that the potentials $p_n$ and $s_n$ appearing in \eqref{1.2ab} are rapidly decaying and $1-p_n s_n\ne 0$ for $n\in\mathbb{Z}$. Let us use $\{ {\pm z}_{j},m_j\}_{j=1}^N$ and $\{ {\pm \bar{z}}_{j},\bar{m}_j\}_{j=1}^{\bar{N}}$ to
 denote the corresponding sets for the bound-state locations and their multiplicities. Then, the norming constants $c_{jk}^{(p,s)}$ and $\bar{c}_{jk}^{(p,s)}$  are related to the corresponding residues $t_{jk}^{(p,s)}$ and $\bar{t}_{jk}^{(p,s)}$  and the dependency constants $\gamma_{jk}^{(p,s)}$ and $\bar{\gamma}_{jk}^{(p,s)}$  as
\begin{equation}\label{T.1}
\begin{cases}
 c_{jk}^{(p,s)}=-2\ds\sum_{l=0}^{m_j-1-k}t_{j(k+1+l)}^{(p,s)}
\,\frac{\gamma_{jl}^{(p,s)}}{l!},\qquad 1\le j\le N,\quad 0\le k\le m_j-1,\\
\noalign{\medskip}
 \bar{c}_{jk}^{(p,s)}=2\ds\sum_{l=0}^{\bar{m}_j-1-k}\bar{t}_{j(k+1+l)}^{(p,s)}
\,\frac{\bar{\gamma}_{jl}^{(p,s)}}{l!},\qquad 1\le j\le \bar{N},\quad 0\le k\le \bar{m}_j-1.
\end{cases}
\end{equation}
\end{corollary}

We note that the norming constants are related to the residues and the dependency constants in the same manner both in Theorem~\ref{thm:theorem x4.3} and Corollary~\ref{thm:theorem x4.3a}. Hence,
without loss of any generality,
for the potential pair $(q,r)$  we can define the norming constants $c_{jk}^{(q,r)}$ and $\bar{c}_{jk}^{(q,r)}$, the respective row vectors $C_j^{(q,r)}$ and $\bar{C}_j^{(q,r)}$ appearing in \eqref{A2} and \eqref{A4}, and the respective row vectors $C^{(q,r)}$ and $\bar{C}^{(q,r)}$  appearing in \eqref{TA.1} and \eqref{TA.2} in Corollary~\ref{thm:theorem x4.3a}. The result is stated next.

\begin{definition}
	\label{thm:theorem x4.3b}
Assume that the potentials $q_n$ and $r_n$ appearing in \eqref{1.1} are rapidly decaying and satisfy \eqref{1.1a}. Then, the corresponding norming constants $c_{jk}^{(q,r)}$ and $\bar{c}_{jk}^{(q,r)}$  are related to the residues $t_{jk}^{(q,r)}$ and $\bar{t}_{jk}^{(q,r)}$  and the dependency constants $\gamma_{jk}^{(q,r)}$ and $\bar{\gamma}_{jk}^{(q,r)}$  as
\begin{equation}\label{T.2}
\begin{cases}
c_{jk}^{(q,r)}:=-2\ds\sum_{l=0}^{m_j-1-k}t_{j(k+1+l)}^{(q,r)}
\,\frac{\gamma_{jl}^{(q,r)}}{l!},\qquad 1\le j\le N,\quad 0\le k\le m_j-1,\\
\noalign{\medskip}
\bar{c}_{jk}^{(q,r)}:=2\ds\sum_{l=0}^{\bar{m}_j-1-k}\bar{t}_{j(k+1+l)}^{(q,r)}
\,\frac{\bar{\gamma}_{jl}^{(q,r)}}{l!},\qquad 1\le j\le \bar{N},\quad 0\le k\le \bar{m}_j-1.
\end{cases}
\end{equation}
\end{definition}

If the potential pairs $(q,r),$ $(u,v),$ and $(p,s)$ are related
as in \eqref{x3.1}--\eqref{x3.4}, then the corresponding residues are related as in \eqref{x4.5} and the corresponding dependency constants are related as in \eqref{x4.6} and \eqref{x4.7}.
In the next theorem, when \eqref{x3.1}--\eqref{x3.4} hold, we show how the corresponding bound-state norming constants are related to each other.

\begin{theorem}
	\label{thm:theorem x4.4}
	Assume that the potentials $q_n$ and $r_n$ appearing in \eqref{1.1} are rapidly decaying and satisfy \eqref{1.1a}. Assume further that the potential pairs $(u,v)$ and $(p,s)$ are related to $(q,r)$ as in \eqref{x3.1}--\eqref{x3.4}. Then, the corresponding bound-state norming constants $c_{jk}^{(u,v)}$ and $\bar{c}_{jk}^{(u,v)}$ are related to $c_{jk}^{(p,s)}$ and $\bar{c}_{jk}^{(p,s)}$ as
	\begin{equation}\label{R.4}
	\begin{cases}
	C_j^{(u,v)}=C_j^{(p,s)}\left(I-A_j^{-2}\right), \qquad 1 \le j \le N,\\
	\noalign{\medskip}
	\bar{C}_j^{(p,s)}=\bar{C}_j^{(u,v)}\left[I-(\bar{A}_j)^{-2}\right], \qquad 1 \le j \le \bar{N},
	\end{cases}
	\end{equation}
where $(A_j, B_j, C_j)$ and $(\bar{A}_j, \bar{B}_j, \bar{C}_j)$ are the matrix triplets appearing in \eqref{A1}, \eqref{A2}, \eqref{A3}, and \eqref{A4}. Consequently, we have
	\begin{equation}\label{R.5}
\begin{cases}
C^{(u,v)}=C^{(p,s)}\left(I-A^{-2}\right),\\
\noalign{\medskip}
\bar{C}^{(p,s)}=\bar{C}^{(u,v)}\left[I-(\bar{A})^{-2}\right],
\end{cases}
\end{equation}
where $(A, B, C)$ and $(\bar{A}, \bar{B}, \bar{C})$ are the matrix triplets appearing in \eqref{TA.1} and \eqref{TA.2}, respectively. Similarly, the norming constants $c_{jk}^{(q,r)}$ and $\bar{c}_{jk}^{(q,r)}$ are related to the norming constants $c_{jk}^{(p,s)}$ and $\bar{c}_{jk}^{(p,s)}$ as
\begin{equation}\label{T.3}
\begin{cases}
C_j^{(q,r)}=\ds\frac{E_\infty^{(q,r)}}{D_\infty^{(q,r)}}\,C_j^{(p,s)}, \qquad 1 \le j \le N,\\
\noalign{\medskip}
\bar{C}_j^{(q,r)}=\ds\frac{D_\infty^{(q,r)}}{E_\infty^{(q,r)}}\,\bar{C}_j^{(p,s)}, \qquad 1 \le j \le \bar{N},
\end{cases}
\end{equation}
and hence we also have
\begin{equation}\label{T.4}
\begin{cases}
C^{(p,s)}=\ds\frac{D_\infty^{(q,r)}}{E_\infty^{(q,r)}}
\,C^{(q,r)},\\
\noalign{\medskip}
\bar{C}^{(p,s)}=\ds\frac{E_\infty^{(q,r)}}{D_\infty^{(q,r)}}\,\bar{C}^{(q,r)},
\end{cases}
\end{equation}
where $D_\infty^{(q,r)}$ and $E_\infty^{(q,r)}$ are the constants
appearing in \eqref{x2.42} and \eqref{x2.43}, respectively.
Consequently, we get
\begin{equation}\label{T.4aaa}
\begin{cases}
C^{(u,v)}=\ds\frac{D_\infty^{(q,r)}}{E_\infty^{(q,r)}}
\,C^{(q,r)}\,\left(I-A^{-2}\right),\\
\noalign{\medskip}
\bar{C}^{(u,v)}=\ds\frac{E_\infty^{(q,r)}}{D_\infty^{(q,r)}}
\,\bar{C}^{(q,r)}\left[I-(\bar{A})^{-2}\right]^{-1}.
\end{cases}
\end{equation}
\end{theorem}

\begin{proof}
We will provide only the proof of the first line of \eqref{R.4} because the second line of \eqref{R.4} can be proved in a similar manner. We note that the first line of \eqref{x4.5} yields
\begin{equation}\label{R.6}
t_{jk}^{(p,s)}=t_{jk}^{(u,v)},
\end{equation}
and from \eqref{x4.6} we have
\begin{equation}\label{R.7}
\gamma_{jk}^{(p,s)}=\ds\sum_{l=0}^{k}\ds\binom{k}{l}\,\ds\frac{d^l \sigma(z_j)}{dz^l}\,\gamma_{j(k-l)}^{(u,v)},
\end{equation}
where we recall that $\sigma(z)$ is the scalar quantity defined \eqref{x4.6a}. For the matrix $A_j$ defined in \eqref{A1} we have
\begin{equation}\label{R.8}
A_j^{-2}=\begin{bmatrix}
\frac{1}{z_j^{2}}&-\frac{2}{z_j^{3}}&\frac{3}{z_j^{4}}&\cdots&\frac{(-1)^{m_j}(m_j-1)}{z_j^{m_j}}&\frac{(-1)^{m_j+1}m_j}{z_j^{m_j+1}}\\
\noalign{\medskip}
0&\frac{1}{z_j^{2}}&-\frac{2}{z_j^{3}}&\cdots&\frac{(-1)^{m_j-1}(m_j-2)}{z_j^{m_j-1}}&\frac{(-1)^{m_j}(m_j-1)}{z_j^{m_j}}\\
\noalign{\medskip}
0&0&\frac{1}{z_j^{2}}&\cdots&\frac{(-1)^{m_j-2}(m_j-3)}{z_j^{m_j-2}}&\frac{(-1)^{m_j-1}(m_j-2)}{z_j^{m_j-1}}\\
\vdots&\vdots&\vdots&\ddots&\vdots&\vdots\\
\noalign{\medskip}
0&0&0&\cdots&\frac{1}{z_j^{2}}&-\frac{2}{z_j^{3}}\\
\noalign{\medskip}
0&0&0&\cdots&0&\frac{1}{z_j^{2}}
\end{bmatrix}.
\end{equation}
Using \eqref{x4.6a} and \eqref{R.6}--\eqref{R.8} on the right-hand side of
the first line of \eqref{T.1}, we establish the first line of \eqref{R.4}. By using the summation over all the bound states, i.e. summing over $1 \le j \le N,$ from the first line of \eqref{R.4} we obtain the first line of \eqref{R.5}. In a similar manner, the second line of \eqref{R.4} yields the second line of \eqref{R.5}.  Finally, the proof of \eqref{T.3} and \eqref{T.4} are obtained by using \eqref{x4.5}, the first lines of \eqref{x4.6} and \eqref{x4.7}, and \eqref{T.2}, and by comparing the result with \eqref{T.1}.
\end{proof}

Recall that we use $t_{jk},$ $\gamma_{jk},$ and $c_{jk}$ to denote the residues, the dependency constants, and the norming constants, respectively, corresponding to a bound state at $z=z_j$ with multiplicity $m_j$ for each of the linear systems \eqref{1.1}, \eqref{1.2aa}, and \eqref{1.2ab}. In the next theorem we compare those quantities with the corresponding quantities related to the bound state at   $z=-z_j.$ We also show that the contributions to the Marchenko kernels from $z=z_j$ and from $z=-z_j$ are equal.

\begin{theorem}
	\label{thm:theorem x4.3ab}
	For each of the linear systems  \eqref{1.1}, \eqref{1.2aa}, and \eqref{1.2ab}, as indicated in Theorem~\ref{thm:theorem x4.1}, let the bound states and their multiplicities be described by the sets $\{ {\pm z}_{j},m_j\}_{j=1}^N$ and $\{ {\pm \bar{z}}_{j},\bar{m}_j\}_{j=1}^{\bar{N}}.$ Let the residues $t_{jk}$ and $\bar{t}_{jk}$ be defined as in \eqref{x4.8} and \eqref{x4.9}; the dependency constants $\gamma_{jk}$ and $\bar{\gamma}_{jk}$ be defined as in \eqref{x4.103} and \eqref{x4.107}, respectively;
and the norming constants $c_{jk}$ and $\bar{c}_{jk}$ be defined as in \eqref{A5aa} and \eqref{A6aa}, respectively, or equivalently as in \eqref{T.1} or \eqref{T.2}. We have the following:
	\begin{enumerate}
		\item[\text{\rm(a)}] Let $t_{jk}|_{z=z_j}$ and  $t_{jk}|_{z=-z_j}$ denote the residues at $z=z_j$ and $z=-z_j,$ respectively. Similarly, let $\bar{t}_{jk}|_{z=\bar{z}_j}$ and  $\bar{t}_{jk}|_{z=-\bar{z}_j}$ denote the residues at $z=\bar{z}_j$ and $z=-\bar{z}_j,$ respectively. We then have
		\begin{equation}\label{W.1}
		\begin{cases}
		t_{jk}\bigr|_{z=-z_j}=(-1)^{k}\,t_{jk}\bigr|_{z=z_j},\qquad 1 \le k \le m_j,\\
		\noalign{\medskip}
		\bar{t}_{jk}\bigr|_{z=-\bar{z}_j}=(-1)^{k}\,\bar{t}_{jk}\bigr|_{z=\bar{z}_j},\qquad 1 \le k \le \bar{m}_j.
		\end{cases}
		\end{equation}
		
			\item[\text{\rm(b)}] Let $\gamma_{jk}|_{z=z_j}$ and
$\gamma_{jk}|_{z=-z_j}$ denote the dependency constants at $z=z_j$ and $z=-z_j,$
respectively. Similarly, let $\bar{\gamma}_{jk}|_{z=\bar{z}_j}$ and  $\bar{\gamma}_{jk}|_{z=-\bar{z}_j}$ denote the dependency constants at $z=\bar{z}_j$ and $z=-\bar{z}_j,$ respectively. We then have
		\begin{equation}\label{W.2}
		\begin{cases}
		\gamma_{jk}\bigr|_{z=-z_j}=(-1)^{k}\,\gamma_{jk}\bigr|_{z=z_j},
\qquad 0 \le k \le m_j-1,\\
		\noalign{\medskip}	\bar{\gamma}_{jk}\bigr|_{z=-\bar{z}_j}=(-1)^{k}\,\bar{\gamma}_{jk}\bigr|_{z=\bar{z}_j},\qquad 0 \le k \le \bar{m}_j-1.
		\end{cases}
		\end{equation}
		
		  \item[\text{\rm(c)}] Let $c_{jk}|_{z=z_j}$ and  $c_{jk}|_{z=-z_j}$ denote the norming constants at $z=z_j$ and $z=-z_j,$ respectively. Similarly, let $\bar{c}_{jk}|_{z=\bar{z}_j}$ and  $\bar{c}_{jk}|_{z=-\bar{z}_j}$ denote the norming constants at $z=\bar{z}_j$ and $z=-\bar{z}_j,$ respectively. We then have
		\begin{equation}\label{W.3}
		\begin{cases}
		c_{jk}\bigr|_{z=-z_j}=(-1)^{k}\,c_{jk}\bigr|_{z=z_j},\qquad 0 \le k \le m_j-1,\\
		\noalign{\medskip}
	\bar{c}_{jk}\bigr|_{z=-\bar{z}_j}=(-1)^{k}\,\bar{c}_{jk}\bigr|_{z=\bar{z}_j},\qquad 0 \le k \le \bar{m}_j-1.
		\end{cases}
		\end{equation}
		
		\item[\text{\rm(d)}] The contribution to the Marchenko kernel $\Omega_{n+m}$ from $z=z_j$ is equal to the contribution from $z=-z_j.$ Similarly, the contribution to the Marchenko kernel $\bar{\Omega}_{n+m}$ from $z=\bar{z}_j$ is equal to the contribution from $z=-\bar{z}_j.$
\end{enumerate}
	\end{theorem}

\begin{proof}
The proof of (a) is obtained as follows. We know that the transmission coefficients $T_{\rm r}$ for each of these three linear systems contain $z$ as $z^{2}.$ From  \eqref{x4.8}, using $T_{\rm r}(-z)=T_{\rm r}(z),$ as $z\to -z_j$ we obtain
\begin{equation*}
T_{\rm r}(z)= \ds\frac{(-1)^{m_j}\,t_{jm_j}}{(z+z_j)^{m_j}}+
\ds\frac{(-1)^{m_j-1}\,t_{j(m_j-1)}}{(z+z_j)^{m_j-1}}
+\cdots+\ds\frac{(-1)\,t_{j1}}{(z+z_j)}+O\left(1\right),
\end{equation*}
which yields the first line of \eqref{W.1}. The second line of \eqref{W.1} is obtained from \eqref{x4.9} by proceeding in a similar manner. This completes the proof of (a). Let us now prove (b). From Theorem~\ref{thm:theorem x2.1}(a) and its analogs in Corollary~\ref{thm:theorem x2.1a} and Theorem~\ref{thm:theorem x3.4a}(a), we get
\begin{equation}\label{W.5}
\begin{cases}
\psi_n(-z)=(-1)^{n}\,\psi_n(z),\quad \phi_n(-z)=(-1)^{n}\,\phi_n(z),\\
		\noalign{\medskip}
\bar{\psi}_n(-z)=(-1)^{n}\,\bar{\psi}_n(z),\quad \bar{\phi}_n(-z)=(-1)^{n}\,\bar{\phi}_n(z).
\end{cases}
\end{equation}
Using the first line of \eqref{W.5} in \eqref{x4.103} we determine the dependency constant $\gamma_{jk}|_{z=-z_j}$ and establish the first line of \eqref{W.2}. The second line of \eqref{W.2} is obtained in a similar way by using the second line of \eqref{W.5} in \eqref{x4.107}. This completes the proof of (b). To prove (c) we proceed as follows. Using the first lines of \eqref{W.1} and \eqref{W.2} in \eqref{A5aa} we determine the norming constant $c_{jk}\bigr|_{z=-z_j}$ and establish the first line of \eqref{W.3}. The second line of \eqref{W.3} is proved in a similar way by using the second lines of \eqref{W.1} and \eqref{W.2} in \eqref{A6aa}. This completes the proof of (c). Let us finally prove (d). The right-hand side of \eqref{R.1} is  the contribution to the Marchenko kernel $\Omega_{n+m}$ from the bound state at $z=z_j.$ Using \eqref{R.1} and \eqref{R.3}, with the help of \eqref{A1} and \eqref{A2}, we evaluate the contribution to the Marchenko kernel $\Omega_{n+m}$ from the bound state $z=-z_j$ and we obtain
\begin{equation}\label{W.6}
\ds\frac{1}{2\pi i}\ds\oint dz\,\phi_{n}\,T_{\rm r}\,z^{m-1}=-\ds\frac{1}{2}\,\sum_{l=n}^{\infty}K_{nl}\left(C_{j}\,A_{j}^{l+m-1}\,B_{j}\right)
\Bigr|_{z_j\mapsto -z_j},
\end{equation}
where $z_j\mapsto -z_j$ is used to indicate the substitution of
$-z_j$ for $z_j.$
 From \eqref{A2} and the first line of \eqref{W.3} we get
 \begin{equation}\label{W.7aa}
C_j\bigr|_{z=-z_j}=C_j\bigr|_{z=z_j}\,
\text{\rm{diag}}\left\{(-1)^{m_j},(-1)^{m_j-1},\cdots,(-1)^1
\right\}
,
\end{equation}
where $\text{\rm{diag}}$ is used to denote
the diagonal matrix.
Similarly, from \eqref{A1}, for any integer $n$ we get
\begin{equation}\label{W.7bb}
\begin{split}
A_j^n\,& B_j\bigr|_{z=-z_j}\\
&=
\text{\rm{diag}}\left\{(-1)^{n-(m_j-1)},(-1)^{n-(m_j-2)},\cdots,(-1)^n
\right\}\,A_j^n\,B_j\bigr|_{z=z_j}.
\end{split}
\end{equation}
Using \eqref{W.7aa} and \eqref{W.7bb},
when $n+m$ and $n+l$ are both even integers in \eqref{W.6},
we confirm that the right-hand side of \eqref{W.6} is equal to the
right-hand side of \eqref{R.1}. Hence the
contribution to the Marchenko kernel $\Omega_{n+m}$ from $z=z_j$ is equal to the contribution from $z=-z_j.$ Similarly, we prove that the contribution to $\bar{\Omega}_{n+m}$ from $z=-\bar{z}_j$ is equal to the contribution from $z=\bar{z}_j.$
\end{proof}

Let us remark on the simplicity and elegance of the use of matrix triplets in dealing with bound states with multiplicities. The formulas in \eqref{R.5} are very simple compared to the formulas written for the individual bound-state norming constants. In fact, to extract the formulas for $c_{jk}^{(u,v)}$ from the first line of \eqref{R.4} we postmultiply that first line by a column vector with $m_j$ components so that we get
\begin{equation*}
c_{jk}^{(u,v)}=C_j^{(p,s)}\left(I-A_j^{-2}\right)e_{m_j-1-k},\qquad 0 \le k \le m_j-1,
\end{equation*}
where we use $e_l$ for the column vector with $m_j$ components with all the entries $0$ except $1$ in the $l$th entry. In a similar way, from the second line of \eqref{R.4} we obtain
\begin{equation*}
\bar{c}_{jk}^{(p,s)}=\bar{C}_j^{(u,v)}\left[I-(\bar{A}_j)^{-2}\right]\bar{e}_{\bar{m}_j-1-k},\qquad 0 \le k \le \bar{m}_j-1,
\end{equation*}
where we use $\bar{e}_l$ for the column vector with $\bar{m}_j$ components having $1$ at the $l$th entry and $0$ elsewhere.

Let us also remark that the fundamental result given in \eqref{R.5} is compatible with \eqref{Tx.3}, from which we obtain the Marchenko kernels $\Omega_k^{(u,v)}$ and $\bar{\Omega}_k^{(u,v)}$ and the Marchenko kernels $\Omega_k^{(p,s)}$ and $\bar{\Omega}_k^{(p,s)}.$ When \eqref{x3.1}--\eqref{x3.4} hold, we see that \eqref{x3.11} and \eqref{x3.11a}, respectively, yield
\begin{equation}\label{R.10}
\begin{cases}
R^{(u,v)}=\left(1-\ds\frac{1}{z^2}\right)\,R^{(p,s)},\\
\noalign{\medskip}
\bar{R}^{(p,s)}=\left(1-\ds\frac{1}{z^2}\right)\,\bar{R}^{(u,v)}.
\end{cases}
\end{equation}
Using \eqref{R.10} in \eqref{Tx.4} we obtain
\begin{equation}\label{R.11}
\begin{cases}
\hat{R}_k^{(u,v)}=\hat{R}_k^{(p,s)}-\hat{R}_{k-2}^{(p,s)},\\
\noalign{\medskip}
\hat{\bar{R}}_k^{(p,s)}=\hat{\bar{R}}_k^{(u,v)}-\hat{\bar{R}}_{k+2}^{(p,s)}.
\end{cases}
\end{equation}
 From \eqref{Tx.3} and \eqref{R.11} we see that in the absence of bound states we have
\begin{equation}\label{R.12}
\begin{cases}
\Omega_k^{(u,v)}=\Omega_k^{(p,s)}-\Omega_{k-2}^{(p,s)},\\
\noalign{\medskip}
\bar{\Omega}_k^{(p,s)}=\bar{\Omega}_k^{(u,v)}-\bar{\Omega}_{k+2}^{(u,v)}.
\end{cases}
\end{equation}
In fact, \eqref{R.12} holds even in the presence of bound state. Then, comparing \eqref{R.11} and \eqref{R.12} we get
\begin{equation*}
\begin{cases}
C^{(u,v)} A^{k-1} B=C^{(p,s)} A^{k-1} B-C^{(p,s)} A^{k-3} B,\\
\noalign{\medskip}
\bar{C}^{(p,s)} (\bar{A})^{-k-1} B=\bar{C}^{(u,v)}
 (\bar{A})^{-k-1}\bar{B}-\bar{C}^{(u,v)} (\bar{A})^{-k-3} \bar{B},
\end{cases}
\end{equation*}
which yield the important result given in \eqref{R.5}.
Let us also mention that from
\eqref{R.12} we get
\begin{equation}\label{R.12aaaa}
\begin{cases}
\Omega_k^{(p,s)}=\ds\sum_{l=0}^\infty
\Omega_{k-2l}^{(u,v)},\\
\noalign{\medskip}
\bar{\Omega}_k^{(u,v)}=\ds\sum_{l=0}^\infty \bar{\Omega}_{k+2l}^{(p,s)}.
\end{cases}
\end{equation}

In the next theorem we show that, when the potential pairs $(u,v)$ and $(p,s)$ are related to each other as in \eqref{x3.1}--\eqref{x3.4}, their corresponding Marchenko systems hold simultaneously.

\begin{theorem}
	\label{thm:theorem x4.5}
	Assume that the potentials $q_n$ and $r_n$ appearing in \eqref{1.1} are rapidly decaying and satisfy \eqref{1.1a}. Assume further that the potential pair $(u,v)$ appearing in \eqref{1.2aa} and the potential pair $(p,s)$ appearing in \eqref{1.2ab} are related to $(q,r)$ as in \eqref{x3.1}--\eqref{x3.4}. Then, the Marchenko system related to
\eqref{1.2aa} holds if and only if the Marchenko system related to
\eqref{1.2ab} holds.
\end{theorem}

\begin{proof}
The proof is lengthy and it involves some fine estimates. Let us define
\begin{equation}\label{H.1}
W_{nm}^{(u,v)}:=K_{nm}^{(u,v)}+\ds\sum_{l=n}^{\infty}\bar{K}_{nl}^{(u,v)}\,\bar{\Omega}_{l+m}^{(u,v)},
\end{equation}
\begin{equation}\label{H.2}
\bar{W}_{nm}^{(u,v)}:=\bar{K}_{nm}^{(u,v)}+\ds\sum_{l=n}^{\infty}K_{nl}^{(u,v)}\,\Omega_{l+m}^{(u,v)},
\end{equation}
\begin{equation*}
W_{nm}^{(p,s)}:=K_{nm}^{(p,s)}+\ds\sum_{l=n}^{\infty}\bar{K}_{nl}^{(p,s)}\,\bar{\Omega}_{l+m}^{(p,s)},
\end{equation*}
\begin{equation}\label{H.4}
\bar{W}_{nm}^{(p,s)}:=\bar{K}_{nm}^{(p,s)}+\ds\sum_{l=n}^{\infty}K_{nl}^{(p,s)}\,\Omega_{l+m}^{(p,s)}.
\end{equation}
As seen from \eqref{Tx.1} and the first two equations in \eqref{x2.12} and \eqref{x2.14} we need to prove the equivalence of the Marchenko system
\begin{equation}\label{M.1}
\begin{cases}
\bar{W}_{nm}^{(u,v)}=0,\qquad m>n,\\
\noalign{\medskip}
W_{nm}^{(u,v)}=0,\qquad m>n,
\end{cases}
\end{equation}
and the Marchenko system
\begin{equation}\label{M.2}
\begin{cases}
\bar{W}_{nm}^{(p,s)}=0,\qquad m>n,\\
\noalign{\medskip}
W_{nm}^{(p,s)}=0,\qquad m>n.
\end{cases}
\end{equation}
We provide the proof by relating $\bar{W}_{nm}^{(u,v)}$ to $\bar{W}_{nm}^{(p,s)}$. The relation between $W_{nm}^{(u,v)}$ and $W_{nm}^{(p,s)}$ can be established in a similar manner and hence that proof will be omitted. Using \eqref{x3.5} and \eqref{x3.5a} we relate $\psi_{n}^{(u,v)}$and $\psi_{n}^{(p,s)}$ to each other and apply $\oint dz\,z^{-m-1}/(2\pi i)$ on the resulting equality. Similarly, using \eqref{x3.6} and \eqref{x3.6a} we relate $\bar{\psi}_{n}^{(u,v)}$and $\bar{\psi}_{n}^{(p,s)}$ to each other and apply $\oint dz\,z^{m-1}/(2\pi i)$ on the resulting equality. Then, with the help of \eqref{x3.1} and \eqref{x3.4} we obtain
\begin{equation*}
\begin{split}
\begin{bmatrix}
1&0\\
\noalign{\medskip}
s_{n-1}&1
\end{bmatrix}K_{nm}^{(u,v)}-\begin{bmatrix}
1&0\\
\noalign{\medskip}
0&0
\end{bmatrix}K_{n(m+2)}^{(u,v)}=\begin{bmatrix}
1&-u_n\\
\noalign{\medskip}
s_{n-1}&1
\end{bmatrix}K_{nm}^{(p,s)},
\end{split}
\end{equation*}
\begin{equation*}
\begin{split}
\begin{bmatrix}
1&0\\
\noalign{\medskip}
s_{n-1}&1
\end{bmatrix}\bar{K}_{nm}^{(u,v)}=\begin{bmatrix}
1&-u_n\\
\noalign{\medskip}
s_{n-1}&1
\end{bmatrix}\bar{K}_{nm}^{(p,s)}-\begin{bmatrix}
1&0\\
\noalign{\medskip}
0&0
\end{bmatrix}\bar{K}_{n(m-2)}^{(p,s)},
\end{split}
\end{equation*}
where we have also used \eqref{Tx.2} for $(u,v)$ and $(p,s).$
Using \eqref{H.1} and the first line of \eqref{R.12}, after some simplifications, for $m>n$ we write $\bar{W}_{nm}^{(u,v)}$ appearing in \eqref{H.2} as
\begin{equation}\label{M.5}
\begin{split}
\bar{W}_{nm}^{(u,v)}=\begin{bmatrix}
1&-u_n\\
\noalign{\medskip}
s_{n-1}&1+u_ns_{n-1}
\end{bmatrix}\bar{W}_{nm}^{(p,s)}-\begin{bmatrix}
0&0\\
\noalign{\medskip}
s_{n-1}&1
\end{bmatrix}\bar{W}_{n(m-2)}^{(p,s)}.
\end{split}
\end{equation}
From \eqref{M.5}, when $m>n+2$ we conclude that the first line of \eqref{M.1} holds if and only if the first line of \eqref{M.2} holds. We must analyze the case $m=n+2$ separately because of the appearance of $\bar{W}_{n(m-2)}^{(p,s)}$ in  \eqref{M.5}. Toward that goal we apply $\oint dz\,z^{n-1}/(2\pi i)$ on both sides of the first line of \eqref{Tx.5} with the potential pair $(p,s)$. We then get
\begin{equation}\label{M.11}
\begin{split}
\ds\frac{1}{2\pi i}\ds\oint dz\,\phi_{n}^{(p,s)}\,T_{\rm r}^{(p,s)}\,z^{n-1}=&\ds\frac{1}{2\pi i}\ds\oint dz\,\bar{\psi}_{n}^{(p,s)}\,z^{n-1}\\&+\ds\frac{1}{2\pi i}\ds\oint dz\,\psi_{n}^{(p,s)}\,R^{(p,s)}\,z^{n-1}.
\end{split}
\end{equation}
In this case, besides the bound-state poles of $T_{\rm r}^{(p,s)}$ in $0<|z|<1,$ also the point at $z=0$ contributes to the integral on left-hand side of \eqref{M.11}.
With the help of
\eqref{Tx.2},
\eqref{Tx.4},
\eqref{R.1},
 from \eqref{M.11} we get
\begin{equation}\label{M.12}
\begin{split}
\left[z^{n}\phi_n^{(p,s)}\right]\bigg |_{z=0}T_{\rm r}^{(p,s)}(0)-\sum_{l=n}^{\infty}&K_{nl}^{(p,s)}\,C^{(p,s)}A^{l+n-1}B\\&=\bar{K}_{nn}^{(p,s)}+\ds\sum_{l=n}^{\infty}K_{nl}^{(p,s)}\,\hat{R}_{l+n}^{(p,s)}.
\end{split}
\end{equation}
Using the analogs of
\eqref{x2.16}, \eqref{x2.17}, and \eqref{x2.30b} for the potential pair
$(p,s),$ we have
\begin{equation}\label{M.13}
\left[z^{n}\phi_n^{(p,s)}\right]\bigg |_{z=0}T_{\rm r}^{(p,s)}(0)=\ds\frac{D_\infty^{(p,s)}}{D_{n-1}^{(p,s)}}\begin{bmatrix}
1\\-s_{n-1}
\end{bmatrix}.
\end{equation}
With the help of \eqref{M.13} and the first equality in \eqref{Tx.3} we write \eqref{M.12} as
\begin{equation*}
\bar{K}_{nn}^{(p,s)}+\ds\sum_{l=n}^{\infty}K_{nl}^{(p,s)}\,\Omega_{l+n}^{(p,s)}=\ds\frac{D_\infty^{(p,s)}}{D_{n-1}^{(p,s)}}\begin{bmatrix}
1\\-s_{n-1}
\end{bmatrix},
\end{equation*}
which, with the help of \eqref{H.4}, is seen to be equivalent to
\begin{equation}\label{M.6}
\bar{W}_{nn}^{(p,s)}=\ds\frac{D_\infty^{(p,s)}}{D_{n-1}^{(p,s)}}\begin{bmatrix}
1\\-s_{n-1}
\end{bmatrix}.
\end{equation}
Because of \eqref{M.6} we see that the second term on the right-hand side of  \eqref{M.5} vanishes when $m=n+2$. Consequently, we conclude  that the first lines of \eqref{M.1} and  \eqref{M.2} hold simultaneously also when $m=n+2$.
\end{proof}

\section{The solution to the direct problem}
\label{sec:section5}

In this section, when the potentials $q_n$ and $r_n$ appearing in \eqref{1.1} are rapidly decaying and satisfy \eqref{1.1a}, we provide the solution to the direct scattering problem for \eqref{1.1}, i.e. the determination of the scattering coefficients and the bound-state information for \eqref{1.1} when the potential pair $(q,r)$ is given. The steps in the solution to the direct problem are outlined as follows:
\begin{enumerate}
	\item[\text{\rm(a)}] Using $(q_n,r_n)$ in \eqref{1.1}, we solve \eqref{1.1} with the asymptotic conditions \eqref{x2.3}--\eqref{x2.6} and uniquely construct the four Jost solutions $\psi_{n}^{(q,r)},$ $\phi_{n}^{(q,r)},$ $\bar{\psi}_{n}^{(q,r)},$ $\bar{\phi}_{n}^{(q,r)}$.

	\item[\text{\rm(b)}] We recover the scattering coefficients $T^{(q,r)},$ $\bar{T}^{(q,r)},$ $R^{(q,r)},$ $\bar{R}^{(q,r)},$ $L^{(q,r)},$ $\bar{L}^{(q,r)}$ by using the asymptotics in \eqref{x2.7}--\eqref{x2.9} of the already constructed four Jost solutions.
	
	\item[\text{\rm(c)}] Next, we determine the poles and their multiplicities for the transmission coefficient $T^{(q,r)}$ in $0<|z|<1$ and the poles and their multiplicities for the transmission coefficient $\bar{T}^{(q,r)}$ in $|z|>1$. Note that such poles occur in pairs. We use the notation that the poles of $T^{(q,r)}$ in $0<|z|<1$ occur at $z=\pm z_{j}$ and the multiplicity of the pole at each of
$z=z_j$ and $z=-z_j$
is $m_j$ for $1\le j \le N$. Thus, the  set $\{{\pm z}_{j}, m_j \}_{j=1}^N$ is uniquely determined from the poles of $T^{(q,r)}$ in $0<|z|<1$. In a similar way, we use $\bar{T}^{(q,r)}$ to determine its poles in $|z|>1$ and the multiplicity of each pole. We use the notation that the poles in $|z|>1$ occur when $z=\pm \bar{z}_{j}$ for $1\le j \le \bar{N}$ and the multiplicity of the pole at each of $z=\bar{z}_{j}$ and $z=-\bar{z}_{j}$ is $\bar{m}_{j}.$ Thus, we also obtain the set $\{{\pm \bar{z}}_{j}, \bar{m}_j \}_{j=1}^{\bar{N}}.$
	
	\item[\text{\rm(d)}] With the help of \eqref{x4.8} with $T^{(q,r)}$, we determine the residues $t_{jk}^{(q,r)}$ for $1 \le j \le N$ and $1 \le k \le m_j.$ Similarly, with the help of \eqref{x4.9} with $\bar{T}^{(q,r)},$ we obtain the residues $\bar{t}_{jk}^{(q,r)}$ for $1 \le j \le \bar{N}$ and $1 \le k \le \bar{m}_j.$
	
	\item[\text{\rm(e)}] Using \eqref{x4.103} for the potential pair $(q,r)$, we determine the dependency constants $\gamma_{jk}^{(q,r)}$ for $1 \le j \le N$ and $0 \le k \le m_j-1.$ Similarly, using \eqref{x4.107} with the potential pair $(q,r)$ we obtain the dependency constants $\bar{\gamma}_{jk}^{(q,r)}$ for $1 \le j \le \bar{N}$ and $0 \le k \le \bar{m}_j-1.$
	
	\item[\text{\rm(f)}] Using the constructed residues $t_{jk}^{(q,r)}$ and $\bar{t}_{jk}^{(q,r)}$ and the dependency constants
$\gamma_{jk}^{(q,r)}$ and  $\bar{\gamma}_{jk}^{(q,r)},$ from \eqref{T.2} we obtain the bound-state norming constants $c_{jk}^{(q,r)}$ and $\bar{c}_{jk}^{(q,r)}.$
Note that we also get the triplets
$(A,B,C^{(q,r)})$ and $(\bar{A},\bar{B},\bar{C}^{(q,r)})$ via
\eqref{TA.1}--\eqref{A4}.
\end{enumerate}

\section{The Marchenko system}
\label{sec:section6}

In this section we introduce the linear system \eqref{Z.0} resembling \eqref{Tx.1}, and we call it the Marchenko system for \eqref{1.1}. The system \eqref{Z.0} uses as input the scalar quantities $\Omega_k^{(q,r)}$ and $\bar{\Omega}_k^{(q,r)},$ which are defined as in \eqref{Tx.3}, and hence it is appropriate that we refer to \eqref{Z.0} as the Marchenko system for \eqref{1.1}. We also describe how to obtain the potentials
$q_n$ and $r_n$ from the solution to the Marchenko system \eqref{Z.0}.

The formulation of the Marchenko system for \eqref{1.1} is a significant step in the analysis of inverse problems related to integrable evolution equations. We expect that our method of formulating \eqref{Z.0} can be applied on some other linear systems, both in the continuous and discrete cases, for which a directly relevant Marchenko theory has not yet been established.

Even though a Marchenko system such as \eqref{Z.0} directly related to \eqref{1.1} is desirable, such a system does not yet seem to exist in the literature. We obtain \eqref{Z.0} by exploiting the connection between \eqref{1.1} and \eqref{1.2ab}. The only slight difference from the standard Marchenko theory is that the formulas for $q_n$ and $r_n$  are expressed in terms of the solution to the Marchenko system \eqref{Z.0} not as in \eqref{Tx.8} or \eqref{Tx.9} but as in \eqref{Z.7}
and \eqref{Z.7a}.

In the next theorem we present the derivation of the Marchenko system for \eqref{1.1}.

\begin{theorem}
	\label{thm:theorem 6.1}
	Assume that the potentials $q_n$ and $r_n$ appearing in \eqref{1.1} are rapidly decaying and satisfy \eqref{1.1a}. Then, the Marchenko system given in \eqref{Tx.1} holds with the relevant quantities listed in \eqref{Tx.2}--\eqref{Tx.4} all related to \eqref{1.1}, i.e. we have
	\begin{equation}\label{Z.0}
	\begin{split}
	&\begin{bmatrix}
	\bar{M}_{nm}^{(q,r)}&M_{nm}^{(q,r)}
	\end{bmatrix}+\begin{bmatrix}
	0&\bar{\Omega}_{n+m}^{(q,r)}\\
	\noalign{\medskip}
	\Omega_{n+m}^{(q,r)}&0
	\end{bmatrix}\\&\phantom{xxx}+\sum_{l=n+1}^{\infty}\begin{bmatrix}
	\bar{M}_{nl}^{(q,r)}&M_{nl}^{(q,r)}
	\end{bmatrix}\begin{bmatrix}
	0&\bar{\Omega}_{l+m}^{(q,r)}\\
	\noalign{\medskip}
	\Omega_{l+m}^{(q,r)}&0
	\end{bmatrix}=\begin{bmatrix}
	0&0\\
	\noalign{\medskip}
	0&0
	\end{bmatrix},\qquad m>n.
	\end{split}
	\end{equation}
	Here, the scalar quantities $\Omega_k^{(q,r)}$ and $\bar{\Omega}_k^{(q,r)}$ are related to the scattering data set for \eqref{1.1} as in \eqref{Tx.3}, i.e.
	\begin{equation}\label{F.1}
\begin{cases}
	\Omega_{k}^{(q,r)}:=\hat{R}_{k}^{(q,r)}+C^{(q,r)}A^{k-1}B,
\qquad
k \text{ \rm{even}},
\\
\noalign{\medskip} \bar{\Omega}_{k}^{(q,r)}:=\hat{\bar{R}}_{k}^{(q,r)}+\bar{C}^{(q,r)}(\bar{A})^{-k-1}\bar{B},
\qquad
k \text{ \rm{even}},
\\
\noalign{\medskip}
\Omega_{k}^{(q,r)}:=0,\quad \bar{\Omega}_{k}^{(q,r)}:=0,\qquad
k \text{ \rm{odd}},
\end{cases}
	\end{equation}
	with $\hat{R}_{k}^{(q,r)}$ and $\hat{\bar{R}}_{k}^{(q,r)}$ being related to the reflection coefficients $R^{(q,r)}$ and $\bar{R}^{(q,r)}$ as in \eqref{Tx.4} and the matrix triplets $(A, B, C^{(q,r)})$ and $(\bar{A}, \bar{B}, \bar{C}^{(q,r)})$ are as in \eqref{TA.1} and \eqref{TA.2}, respectively.
\end{theorem}

\begin{proof}
A direct proof can be given by using the procedure described in \eqref{Tx.5}--\eqref{Tx.7}. We present an alternate proof, and this is done by
exploiting the connection between
\eqref{1.1} and \eqref{1.2ab}
when the potential pairs $(q,r)$ and $(p,s)$ are related as in
\eqref{x3.3} and \eqref{x3.4}.
Starting with the Marchenko system \eqref{Tx.1} with the relevant quantities all related to the potential pair $(p,s)$ of \eqref{1.2ab}, we transform that Marchenko system and the relevant quantities so that they are all related to the potential pair $(q,r)$ of \eqref{1.1}. From \eqref{x3.5a} and \eqref{x3.6a} we see that
\begin{equation}\label{Z.1}
\psi_{n}^{(p,s)}=\ds\frac{1}{D_\infty^{(q,r)}}\left(\Lambda_n^{(q,r)}\right)^{-1}\psi_{n}^{(q,r)},\quad \bar{\psi}_{n}^{(p,s)}=\ds\frac{1}{E_\infty^{(q,r)}}\left(\Lambda_n^{(q,r)}\right)^{-1}\bar{\psi}_{n}^{(q,r)},
\end{equation}
where $\Lambda_n^{(q,r)}$ is the matrix defined in \eqref{x4.110} and the quantities $D_\infty^{(q,r)}$ and $E_\infty^{(q,r)}$ are the scalar constants appearing in \eqref{x2.42} and \eqref{x2.43}, respectively. From \eqref{x4.111} we know that the matrix $\Lambda_n^{(q,r)}$ is invertible for all $n\in\mathbb{Z}.$ Thus, with the help of \eqref{Tx.2} and \eqref{Z.1} we conclude that
\begin{equation}\label{Z.2}
K_{nm}^{(p,s)}=\ds\frac{1}{D_\infty^{(q,r)}}\left(\Lambda_n^{(q,r)}\right)^{-1}K_{nm}^{(q,r)},\quad \bar{K}_{nm}^{(p,s)}=\ds\frac{1}{E_\infty^{(q,r)}}\left(\Lambda_n^{(q,r)}\right)^{-1}\bar{K}_{nm}^{(q,r)}.
\end{equation}
From the second equalities in \eqref{x3.11} and \eqref{x3.11a}, with the help of \eqref{Tx.4} we obtain
\begin{equation}\label{Z.3}
\hat{R}^{(p,s)}=\ds\frac{D_\infty^{(q,r)}}{E_\infty^{(q,r)}}\,\hat{R}^{(q,r)},\quad \hat{\bar{R}}^{(p,s)}=\ds\frac{E_\infty^{(q,r)}}{D_\infty^{(q,r)}}\,\hat{\bar{R}}^{(q,r)}.
\end{equation}
From the second equality in \eqref{x3.9a} it follows that the poles of $T_{\rm r}^{(p,s)}$ and $T^{(q,r)}$ coincide, and from the second equality in \eqref{x3.10a} we see that the poles of $\bar{T}_{\rm r}^{(p,s)}$ and $\bar{T}^{(q,r)}$ coincide. Hence, the matrices $A,$ $\bar{A},$ $B,$ $\bar{B}$ appearing in \eqref{Tx.3} are common to the potential pairs $(p,s)$ and $(q,r).$
Using \eqref{T.4} and \eqref{Z.3} in \eqref{Tx.3} we conclude that
\begin{equation}\label{Z.5}
\Omega_k^{(p,s)}=\ds\frac{D_\infty^{(q,r)}}{E_\infty^{(q,r)}}\,\Omega_k^{(q,r)},\quad \bar{\Omega}_k^{(p,s)}=\ds\frac{E_\infty^{(q,r)}}{D_\infty^{(q,r)}}\,\bar{\Omega}_k^{(q,r)}.
\end{equation}
With the help of the first equalities in
\eqref{x2.12aa} and \eqref{x2.14aa},
we observe that the Marchenko system \eqref{Tx.1} related to the potential pair $(p,s)$ is equivalent to the system given in \eqref{M.2}. Using \eqref{Z.2} and \eqref{Z.5} we transform \eqref{M.2} into
\begin{equation}\label{Z.6}
\begin{cases}
\bar{K}_{nm}^{(q,r)}+\ds\sum_{l=n}^{\infty}K_{nl}^{(q,r)}\,\Omega_{l+m}^{(q,r)}=0,\qquad m>n,\\
K_{nm}^{(q,r)}+\ds\sum_{l=n}^{\infty}\bar{K}_{nl}^{(q,r)}\,\bar{\Omega}_{l+m}^{(q,r)}=0,\qquad m>n.
\end{cases}
\end{equation}
The system in \eqref{Z.6} can be written in the matrix form as
\begin{equation}\label{Z.x1}
\begin{split}
&\begin{bmatrix}
\bar{K}_{nm}^{(q,r)}&K_{nm}^{(q,r)}
\end{bmatrix}+\begin{bmatrix}
\bar{K}_{nn}^{(q,r)}&K_{nn}^{(q,r)}
\end{bmatrix}\begin{bmatrix}
0&\bar{\Omega}_{n+m}^{(q,r)}\\
\noalign{\medskip}
\Omega_{n+m}^{(q,r)}&0
\end{bmatrix}\\&\phantom{xxx}+\sum_{l=n+1}^{\infty}\begin{bmatrix}
\bar{K}_{nl}^{(q,r)}&K_{nl}^{(q,r)}
\end{bmatrix}\begin{bmatrix}
0&\bar{\Omega}_{l+m}^{(q,r)}\\
\noalign{\medskip}
\Omega_{l+m}^{(q,r)}&0
\end{bmatrix}=\begin{bmatrix}
0&0\\
\noalign{\medskip}
0&0
\end{bmatrix},\qquad m>n.
\end{split}
\end{equation}
The matrix $[
\bar{K}_{nn}^{(p,s)}\quad K_{nn}^{(p,s)}
],$ as seen from the first equalities in \eqref{x2.12aa} and \eqref{x2.14aa}, is equal to the $2\times 2$ identity matrix, and hence the second term on the left-hand side of \eqref{Tx.1} does not contain the matrix $[
\bar{K}_{nn}^{(p,s)}\quad K_{nn}^{(p,s)}
].$ However, the matrix $[
\bar{K}_{nn}^{(q,r)}\quad K_{nn}^{(q,r)}
]$ appearing in \eqref{Z.x1} is not equal to the identity matrix. From \eqref{x5.2} and \eqref{x5.5} it follows that
\begin{equation}\label{Z.x2}
\begin{bmatrix}
\bar{K}_{nn}^{(q,r)}&K_{nn}^{(q,r)}
\end{bmatrix}=\Lambda_n^{(q,r)}\begin{bmatrix}
E_\infty^{(q,r)}&0\\
\noalign{\medskip}
0&D_\infty^{(q,r)}
\end{bmatrix},
\end{equation}
where we recall that $\Lambda_n^{(q,r)}$ is the invertible matrix appearing in \eqref{x4.110}. Hence, the matrix on the left-hand side of \eqref{Z.x2} is invertible and we have
\begin{equation}
\label{Z.x3}
\begin{bmatrix}
\bar{K}_{nn}^{(q,r)}&K_{nn}^{(q,r)}
\end{bmatrix}^{-1}=\begin{bmatrix}
\ds\frac{1}{E_\infty^{(q,r)}}&0\\
\noalign{\medskip}
0&\ds\frac{1}{D_\infty^{(q,r)}}
\end{bmatrix}\left(\Lambda_n^{(q,r)}\right)^{-1}.
\end{equation}
Premultiplying both sides of \eqref{Z.x1} by $[
\bar{K}_{nn}^{(q,r)}\quad K_{nn}^{(q,r)}
]^{-1},$ we obtain \eqref{Z.0}, where we have defined
\begin{equation}\label{x6.00}
\begin{bmatrix}
\bar{M}_{nm}^{(q,r)}&M_{nm}^{(q,r)}
\end{bmatrix}:=\begin{bmatrix}
\bar{K}_{nn}^{(q,r)}&K_{nn}^{(q,r)}
\end{bmatrix}^{-1}\begin{bmatrix}
\bar{K}_{nm}^{(q,r)}&K_{nm}^{(q,r)}
\end{bmatrix}.
\end{equation}
\end{proof}

Note that \eqref{x6.00} implies that
$M_{nm}^{(q,r)}=0$ and
$\bar{M}_{nm}^{(q,r)}=0$
when $n+m$ is odd because
we have $K_{nm}^{(q,r)}=0$ and
$\bar{K}_{nm}^{(q,r)}=0$
when $n+m$ is odd as stated in Theorem~\ref{thm:theorem x3.4a}.

We can uncouple the Marchenko system \eqref{Z.0} as in \eqref{ta002} and \eqref{ta1001}.
Hence, without a proof we state the result in the next corollary.

\begin{corollary}
\label{thm:theorem 6.1a}
Assume that the potentials $q_n$ and $r_n$ appearing in \eqref{1.1} are rapidly decaying and satisfy \eqref{1.1a}. Then, the Marchenko system  \eqref{Z.0} is equivalent to the uncoupled system, for $m>n$, given by
\begin{equation}\label{6.1a}
\begin{cases}
\begin{bmatrix} M_{nm}^{(q,r)}\end{bmatrix}_1
+\bar{\Omega}_{n+m}^{(q,r)}-
\ds\sum_{l=n+1}^{\infty}
\ds\sum_{j=n+1}^{\infty}\begin{bmatrix} M_{nj}^{(q,r)}\end{bmatrix}_1\,
\Omega_{j+l}^{(q,r)}\,\bar{\Omega}_{l+m}^{(q,r)}=0,\\
\noalign{\medskip}
\begin{bmatrix}\bar{M}_{nm}^{(q,r)}\end{bmatrix}_2
+\Omega_{n+m}^{(q,r)}-
\ds\sum_{l=n+1}^{\infty}
\ds\sum_{j=n+1}^{\infty}\begin{bmatrix}
\bar{M}_{nj}^{(q,r)}\end{bmatrix}_2\,\bar{\Omega}_{j+l}^{(q,r)}
\,
\Omega_{l+m}^{(q,r)}=0,
\end{cases}
\end{equation}
and with $[\bar{M}_{nm}^{(q,r)}]_1$ and $[M_{nm}^{(q,r)}]_2$ obtained from the solution to \eqref{6.1a} as
\begin{equation*}
\begin{cases}
\begin{bmatrix}\bar{M}_{nm}^{(q,r)}\end{bmatrix}_1
=-\ds\sum_{l=n+1}^{\infty}
\begin{bmatrix}
M_{nl}^{(q,r)}\end{bmatrix}_1\,
\Omega_{l+m}^{(q,r)}
,\\
\noalign{\medskip}
\begin{bmatrix} M_{nm}^{(q,r)}\end{bmatrix}_2
=-\ds\sum_{l=n+1}^{\infty}
\begin{bmatrix} \bar{M}_{nl}^{(q,r)}\end{bmatrix}_2\,
\bar{\Omega}_{l+m}^{(q,r)}
,
\end{cases}
\end{equation*}
where we recall that
$[\cdot]_1$ and $[\cdot]_2$ denote the first and second components of the relevant column vectors, as indicated in \eqref{ta001}.
\end{corollary}

In the next theorem we describe the recovery of $q_n$ and $r_n$ from the solution to the Marchenko system \eqref{Z.0}.

\begin{theorem}
\label{thm:theorem 6.2}
Assume that the potentials $q_n$ and $r_n$ appearing in \eqref{1.1} are rapidly decaying and satisfy \eqref{1.1a}. Then, $q_n$ and $r_n$ are recovered from the solution to the Marchenko system given in \eqref{Z.0} via
\begin{equation}\label{Z.7}
q_n=\ds\frac{\ds\sum_{l=n}^{\infty}
\begin{bmatrix}M_{nl}^{(q,r)}\end{bmatrix}_1\,
\ds\sum_{k=n}^{\infty}
\begin{bmatrix}M_{nk}^{(q,r)}\end{bmatrix}_2}
{\ds\sum_{l=n}^{\infty}
\begin{bmatrix}\bar{M}_{nl}^{(q,r)}\end{bmatrix}_1\,
\ds\sum_{k=n}^{\infty}
\begin{bmatrix}M_{nk}^{(q,r)}\end{bmatrix}_2
-\ds\sum_{l=n}^{\infty}
\begin{bmatrix}M_{nl}^{(q,r)}\end{bmatrix}_1\,
\ds\sum_{k=n}^{\infty}
\begin{bmatrix}
\bar{M}_{nk}^{(q,r)}\end{bmatrix}_2},
\end{equation}
\begin{equation}\label{Z.7a}
r_n=\ds\frac{\ds\sum_{l=n-1}^{\infty}
\begin{bmatrix}\bar{M}_{(n-1)l}^{(q,r)}\end{bmatrix}_2}
{\ds \sum_{l=n-1}^{\infty}
\begin{bmatrix}M_{(n-1)l}^{(q,r)}\end{bmatrix}_2}
-\ds\frac{\ds\sum_{l=n}^{\infty}
\begin{bmatrix}\bar{M}_{nl}^{(q,r)}\end{bmatrix}_2}
{\ds \sum_{l=n}^{\infty}
\begin{bmatrix}
M_{nl}^{(q,r)}\end{bmatrix}_2},
\end{equation}
where $[\cdot]_1$ and $[\cdot]_2$ denote the first and second components of the relevant column vectors, as indicated in \eqref{ta001}.
\end{theorem}

\begin{proof}
Using \eqref{x3.5a}, \eqref{x3.6a}, and \eqref{Tx.2} we get	
\begin{equation}\label{x6.15}
\begin{bmatrix}
\bar{K}_{nl}^{(q,r)}&K_{nl}^{(q,r)}
\end{bmatrix}=\Lambda_n^{(q,r)}\begin{bmatrix}
\bar{K}_{nl}^{(p,s)}&K_{nl}^{(p,s)}
\end{bmatrix}\begin{bmatrix}
E_\infty^{(q,r)}&0\\
\noalign{\medskip}
0&D_\infty^{(q,r)}
\end{bmatrix},
\end{equation}
where $\Lambda_n^{(q,r)}$ is the invertible matrix in \eqref{x4.110}; $D_\infty^{(q,r)}$ and $E_\infty^{(q,r)}$ are the scalar constants appearing in \eqref{x2.42} and \eqref{x2.43}, respectively; $K_{nl}^{(q,r)}$ and $\bar{K}_{nl}^{(q,r)}$ are the column vectors in \eqref{x5.1} and \eqref{x5.4}, respectively; $K_{nl}^{(p,s)}$ and $\bar{K}_{nl}^{(p,s)}$ are the column vectors in \eqref{x2.11aa} and \eqref{x2.13aa}, respectively.
With the help of \eqref{Z.x3}, from \eqref{x6.15} for $l\ge n$ we get
\begin{equation}\label{x6.17}
\begin{split}
&\begin{bmatrix}
\bar{K}_{nn}^{(q,r)}&K_{nn}^{(q,r)}
\end{bmatrix}^{-1}\begin{bmatrix}
\bar{K}_{nl}^{(q,r)}&K_{nl}^{(q,r)}
\end{bmatrix}
\\&\phantom{xxxxxxxxx}=\begin{bmatrix}
E_\infty^{(q,r)}&0\\
\noalign{\medskip}
0&D_\infty^{(q,r)}
\end{bmatrix}^{-1}\begin{bmatrix}
\bar{K}_{nl}^{(p,s)}&K_{nl}^{(p,s)}
\end{bmatrix}\begin{bmatrix}
E_\infty^{(q,r)}&0\\
\noalign{\medskip}
0&D_\infty^{(q,r)}
\end{bmatrix}.
\end{split}
\end{equation}
We note that the left-hand side of \eqref{x6.17} is equal to the left-hand side of \eqref{x6.00}. Using the summation with $l\ge n,$ from \eqref{x6.17} we obtain
\begin{equation}\label{x6.18}
\begin{split}
&\ds\sum_{l=n}^{\infty}\begin{bmatrix}
\bar{M}_{nl}^{(q,r)}&M_{nl}^{(q,r)}
\end{bmatrix}
\\
&\phantom{xxxxxxx}=\begin{bmatrix}
E_\infty^{(q,r)}&0\\
\noalign{\medskip}
0&D_\infty^{(q,r)}
\end{bmatrix}^{-1}\sum_{l=n}^{\infty}\begin{bmatrix}
\bar{K}_{nl}^{(p,s)}&K_{nl}^{(p,s)}
\end{bmatrix}\begin{bmatrix}
E_\infty^{(q,r)}&0\\
\noalign{\medskip}
0&D_\infty^{(q,r)}
\end{bmatrix}.
\end{split}
\end{equation}
We remark that the summation on the right-hand side in \eqref{x6.18} is related
to $[
\bar{\psi}_n^{(p,s)}\quad \psi_n^{(p,s)}
]$ evaluated at $z=1$, as seen from \eqref{x2.11aa} and
\eqref{x2.13aa}. With the help of \eqref{6.1cc} and \eqref{6.1cccc}
we express the right-hand side of \eqref{x6.18} in terms of the matrix
on the right-hand side of \eqref{6.1cc}, and we get
\begin{equation}\label{x6.19}
\ds\sum_{l=n}^{\infty}\begin{bmatrix}
\bar{M}_{nl}^{(q,r)}&M_{nl}^{(q,r)}
\end{bmatrix}=\begin{bmatrix}
\ds\frac{E_{n-1}^{(q,r)}}{E_\infty^{(q,r)}}\,\left(1+q_n\,
\sum_{j=n+1}^{\infty} r_j\right)&q_n\,\ds\frac{E_{n-1}^{(q,r)}}{E_\infty^{(q,r)}}\\
\noalign{\medskip}
\ds\frac{D_n^{(q,r)}}{D_\infty^{(q,r)}}\,\sum_{j=n+1}^{\infty} r_j
&\ds\frac{D_n^{(q,r)}}{D_\infty^{(q,r)}}
\end{bmatrix}.
\end{equation}
Using the notation of \eqref{ta001}, from the $(2,1)$ and $(2,2)$
entries in \eqref{x6.19} we obtain \eqref{Z.7a}.
Then, from the $(1,1)$ and $(1,2)$ entries in \eqref{x6.19}
and using \eqref{Z.7a}, we obtain \eqref{Z.7}.
\end{proof}

\section{The alternate Marchenko system}
\label{sec:section7}

In this section we derive the pair of scalar Marchenko equations given in \eqref{6.22d} and \eqref{6.23}, which resembles the uncoupled Marchenko system given in \eqref{ta002}. We refer to the uncoupled system composed of  \eqref{6.22d} and \eqref{6.23} as the alternate Marchenko system. Such a system is the discrete analog of the Marchenko system given in $(6.22)$ and $(6.23)$ of \cite{AE19} in the continuous case.  In this section we also show that the potentials $q_n$ and $r_n$
are recovered as in \eqref{6.21} and \eqref{6.22} from the solution to the alternate Marchenko system.

We remark that the uncoupled alternate Marchenko equation \eqref{6.22d} is closely related to the system \eqref{1.2aa} with the potential pair $(u,v),$ and hence we use the superscript $(u,v)$ in the quantities appearing in \eqref{6.22d}. Similarly, the uncoupled alternate Marchenko equation \eqref{6.23} involves the quantities closely related to \eqref{1.2ab} with the potential pair $(p,s),$ and hence we use the superscript $(p,s)$ in the quantities appearing in \eqref{6.23}.
Our alternate Marchenko equations \eqref{6.22d} and \eqref{6.23}
and our recovery formulas
\eqref{6.21} and \eqref{6.22} are closely related to (4.12c), (4.12d),
(4.21a), and (4.21b), respectively, of \cite{tsuchida2010new}.
We remark that Tsuchida in \cite{tsuchida2010new} assumes that the bound states are
all simple, and we also mention that, contrary to our own
 \eqref{6.22d} and \eqref{6.23}, Tsuchida's
(4.12c) and (4.12d) in \cite{tsuchida2010new} lack the appropriate symmetry
for a standard Marchenko system apparent in \eqref{ta002} in the discrete case.

Let us make a comparison between the alternate Marchenko system used in this section and the Marchenko system introduced in Section~\ref{sec:section6}. The Marchenko system \eqref{Z.0} has the same
standard form used in other inverse problems arising in applications, but the recovery of the potentials $q_n$ and $r_n$ from the solution to \eqref{Z.0} is not standard, i.e. the recovery is not of the form given in \eqref{Tx.8} or \eqref{Tx.9}. On the other hand, certain terms in the alternate Marchenko system
 involve some discrete spacial derivatives and hence the alternate Marchenko system slightly differs from the standard Marchenko system \eqref{Tx.1}. However, the recovery of the potentials $q_n$ and $r_n$ is similar to recovery described in \eqref{Tx.8} and \eqref{Tx.9}, which are used as the standard recovery formulas for other standard Marchenko systems.

Inspired by \eqref{6.6} and \eqref{6.8} we define the scalar quantities  $\mathscr{K}_{nm}^{(u,v)}$ and $\bar{\mathscr{K}}_{nm}^{(p,s)},$ respectively, as
\begin{equation}\label{6.15}
\mathscr{K}_{nm}^{(u,v)}:=\ds\frac{\ds\sum_{l=m}^{\infty}
\begin{bmatrix}K_{nl}^{(u,v)}\end{bmatrix}_1 }
{\ds\sum_{l=n}^{\infty}
\begin{bmatrix}
\bar{K}_{nl}^{(u,v)}\end{bmatrix}_1 },\qquad m\ge n,
\end{equation}
\begin{equation}\label{6.17}
\bar{\mathscr{K}}_{nm}^{(p,s)}:=\ds\frac{\ds\sum_{l=m}^{\infty}
\begin{bmatrix}\bar{K}_{nl}^{(p,s)}\end{bmatrix}_2}
{\ds\sum_{l=n}^{\infty}
\begin{bmatrix}K_{nl}^{(p,s)}\end{bmatrix}_2},\qquad m\ge n,
\end{equation}
where we use the notation of \eqref{ta001} and recall that $K_{nl}^{(u,v)}$ and $\bar{K}_{nl}^{(u,v)}$ satisfy (f) and (g) of Theorem~\ref{thm:theorem x2.1}, and similarly, $K_{nl}^{(p,s)}$ and $\bar{K}_{nl}^{(p,s)}$ satisfy (a) and (b) of Corollary~\ref{thm:theorem x2.1a}.
We remark that the $m$-dependence of $\mathscr{K}_{nm}^{(u,v)}$ and $\bar{\mathscr{K}}_{nm}^{(p,s)}$ occurs only in the numerators in \eqref{6.15} and \eqref{6.17}. When $m=n,$ with the help of \eqref{6.1bbbb}, \eqref{6.1cccc}, \eqref{6.15}, and \eqref{6.17} we obtain
\begin{equation}\label{6.15c}
\mathscr{K}_{nn}^{(u,v)}=\ds\frac{\ds\sum_{l=n}^{\infty}
\begin{bmatrix}K_{nl}^{(u,v)}\end{bmatrix}_1}
{\ds\sum_{l=n}^{\infty}
\begin{bmatrix}\bar{K}_{nl}^{(u,v)}\end{bmatrix}_1}
=\ds\frac{\begin{bmatrix}\psi_{n}^{(u,v)}(1)\end{bmatrix}_1}
{\begin{bmatrix}\bar{\psi}_{n}^{(u,v)}(1)\end{bmatrix}_1},
\end{equation}
\begin{equation}\label{6.17c}
\bar{\mathscr{K}}_{nn}^{(p,s)}=\ds\frac{\ds\sum_{l=n}^{\infty}
\begin{bmatrix}\bar{K}_{nl}^{(p,s)}\end{bmatrix}_2}
{\ds\sum_{l=n}^{\infty}
\begin{bmatrix}K_{nl}^{(p,s)}\end{bmatrix}_2 }
=\ds\frac{\begin{bmatrix}\bar{\psi}_{n}^{(p,s)}(1)\end{bmatrix}_2}
{\begin{bmatrix}\psi_{n}^{(p,s)}(1)\end{bmatrix}_2}.
\end{equation}
Comparing \eqref{6.6}, \eqref{6.8}, \eqref{6.15c}, and \eqref{6.17c} we observe that the potentials $q_n$ and $r_n$ are recovered from $\mathscr{K}_{nm}^{(u,v)}$ and $\bar{\mathscr{K}}_{nm}^{(p,s)},$ respectively, as
\begin{equation}\label{6.21}
q_n=\ds\frac{D_\infty^{(q,r)}}{E_\infty^{(q,r)}}
\left(\mathscr{K}_{nn}^{(u,v)}-\mathscr{K}_{(n+1)(n+1)}^{(u,v)}\right),
\end{equation}
\begin{equation}\label{6.22}
r_n=\ds\frac{E_\infty^{(q,r)}}{D_\infty^{(q,r)}}
\left(\bar{\mathscr{K}}_{(n-1)(n-1)}^{(p,s)}-\bar{\mathscr{K}}_{nn}^{(p,s)}\right),
\end{equation}
where we recall that $D_\infty^{(q,r)}$ and $E_\infty^{(q,r)}$ are the constants
appearing in \eqref{x2.42} and \eqref{x2.43}, respectively.

In the next theorem we show that the scalar quantities $\mathscr{K}_{nm}^{(u,v)}$ and $\bar{\mathscr{K}}_{nm}^{(p,s)}$ given in \eqref{6.15} and \eqref{6.17}
satisfy the respective alternate Marchenko equations, for $m>n,$ given by
\begin{equation}\label{6.22d}
\begin{split}
&\mathscr{K}_{nm}^{(u,v)}+\bar{G}_{n+m}^{(u,v)}
\\ &
\phantom{xxx}
+\sum_{l=n+1}^{\infty}\sum_{j=n+1}^{\infty}\left(
\mathscr{K}_{n(j+1)}^{(u,v)}-\mathscr{K}_{nj}^{(u,v)}\right)G_{j+l}^{(u,v)}\left(\bar{G}_{l+m}^{(u,v)}
-\bar{G}_{l+m-1}^{(u,v)}\right)=0,
\end{split}
\end{equation}
\begin{equation}\label{6.23}
\begin{split}
&\bar{\mathscr{K}}_{nm}^{(p,s)}+G_{n+m}^{(p,s)}
\\ &
\phantom{xxx}
+\sum_{l=n+1}^{\infty}\sum_{j=n+1}^{\infty}\left(
\bar{\mathscr{K}}_{n(j+1)}^{(p,s)}-\bar{\mathscr{K}}_{nj}^{(p,s)}\right)\bar{G}_{j+l}^{(p,s)}\left(G_{l+m}^{(p,s)}
-G_{l+m-1}^{(p,s)}\right)=0,
\end{split}
\end{equation}
where  we have defined
\begin{equation}\label{6.24}
G_n^{(u,v)}:=\sum_{j=n}^{\infty}\Omega_j^{(u,v)},\quad \bar{G}_n^{(u,v)}:=\sum_{j=n}^{\infty}\bar{\Omega}_j^{(u,v)},
\end{equation}
\begin{equation}\label{6.25}
G_n^{(p,s)}:=\sum_{j=n}^{\infty}\Omega_j^{(p,s)},\quad \bar{G}_n^{(p,s)}:=\sum_{j=n}^{\infty}\bar{\Omega}_j^{(p,s)},
\end{equation}
with the scalar functions $\Omega^{(u,v)}_j,$ $\bar{\Omega}^{(u,v)}_j,$ $\Omega^{(p,s)}_j,$ $\bar{\Omega}^{(p,s)}_j$ defined as in \eqref{Tx.3} for the potentials pairs $(u,v)$ and $(p,s),$ respectively.

\begin{theorem}
	\label{thm:theorem6.3}
	Assume that the potentials $q_n$ and $r_n$ appearing in \eqref{1.1} are rapidly decaying and satisfy \eqref{1.1a}. Assume further that the potential pairs $(u,v)$ and $(p,s)$ are related to $(q,r)$ as in \eqref{x3.1}--\eqref{x3.4}. Let  $\mathscr{K}_{nm}^{(u,v)}$ and $\bar{\mathscr{K}}_{nm}^{(p,s)}$ be the scalar quantities defined as in \eqref{6.15} and \eqref{6.17}, respectively, and let $G^{(u,v)}_n,$ $\bar{G}^{(u,v)}_n,$  $G^{(p,s)}_n,$ $\bar{G}^{(p,s)}_n$ be the quantities defined in \eqref{6.24} and \eqref{6.25}. Then, $\mathscr{K}_{nm}^{(u,v)}$ and $\bar{\mathscr{K}}_{nm}^{(p,s)}$ satisfy the alternate Marchenko system given in \eqref{6.22d} and \eqref{6.23}, respectively.
\end{theorem}

\begin{proof}
In the notation of \eqref{ta001}, the $(1,2)$ entry in the matrix Marchenko system \eqref{Tx.1} for the potential pair $(u,v)$ is given by
\begin{equation}\label{6.31a}
\begin{bmatrix}K_{nk}^{(u,v)}\end{bmatrix}_1
+\bar{\Omega}_{n+k}^{(u,v)}
+\sum_{l=n+1}^{\infty}\begin{bmatrix}
\bar{K}_{nl}^{(u,v)}\end{bmatrix}_1
\bar{\Omega}_{l+k}^{(u,v)}=0,
\qquad k>n.
\end{equation}
Adding and subtracting $\bar{\Omega}_{n+k}^{(u,v)}$ to $\bar{\Omega}_{l+k}^{(u,v)}$ in \eqref{6.31a}, we obtain
\begin{equation}\label{6.43}
\begin{split}
\begin{bmatrix}K_{nk}^{(u,v)}\end{bmatrix}_1
+&\bar{\Omega}_{n+k}^{(u,v)}+\sum_{l=n+1}^{\infty}
\begin{bmatrix}\bar{K}_{nl}^{(u,v)}\end{bmatrix}_1
\bar{\Omega}_{n+k}^{(u,v)}\\&+\sum_{l=n+1}^{\infty}
\begin{bmatrix}\bar{K}_{nl}^{(u,v)}\end{bmatrix}_1
\left(\bar{\Omega}_{l+k}^{(u,v)}-\bar{\Omega}_{n+k}^{(u,v)}\right)=0.
\end{split}
\end{equation}
Using $[\bar{K}_{nn}^{(u,v)}]_1=1,$ as seen from the first equality in \eqref{x2.14},
we combine the second and third terms on the left-hand side of \eqref{6.43} to obtain
\begin{equation}\label{6.44}
\begin{split}
\begin{bmatrix}K_{nk}^{(u,v)}\end{bmatrix}_1
+&\bar{\Omega}_{n+k}^{(u,v)}
\sum_{l=n}^{\infty}
\begin{bmatrix}\bar{K}_{nl}^{(u,v)}\end{bmatrix}_1
\\&+\sum_{l=n+1}^{\infty}
\begin{bmatrix}\bar{K}_{nl}^{(u,v)}\end{bmatrix}_1
\left(\bar{\Omega}_{l+k}^{(u,v)}-\bar{\Omega}_{n+k}^{(u,v)}\right)=0.
\end{split}
\end{equation}
From \eqref{6.1bbbb} we see that the summation in the second term on the left-hand side of \eqref{6.44} is equal to $[\bar{\psi}_n^{(u,v)}(1)]_1,$ and hence by dividing \eqref{6.44} by that term we get
\begin{equation}\label{6.45}
\ds\frac{\begin{bmatrix}K_{nk}^{(u,v)}\end{bmatrix}_1}
{\begin{bmatrix}\bar{\psi}_{n}^{(u,v)}(1)\end{bmatrix}_1}
+\bar{\Omega}_{n+k}^{(u,v)}
+\sum_{l=n+1}^{\infty}
\ds\frac{\begin{bmatrix}\bar{K}_{nl}^{(u,v)}\end{bmatrix}_1}
{\begin{bmatrix}\bar{\psi}_{n}^{(u,v)}(1)\end{bmatrix}_1}
\left(\bar{\Omega}_{l+k}^{(u,v)}
-\bar{\Omega}_{n+k}^{(u,v)}\right)=0.
\end{equation}
Using the $(1,1)$ entry in the Marchenko system \eqref{Tx.1} for the potential pair $(u,v),$  we write \eqref{6.45} as
\begin{equation}\label{6.46}
\begin{split}
\ds\frac{\begin{bmatrix}K_{nk}^{(u,v)}\end{bmatrix}_1}
{\begin{bmatrix}\bar{\psi}_{n}^{(u,v)}(1)\end{bmatrix}_1}
&+\bar{\Omega}_{n+k}^{(u,v)}
\\&
-\ds\sum_{l=n+1}^{\infty}\ds\sum_{j=n+1}^{\infty}
\ds\frac{\begin{bmatrix}K_{nj}^{(u,v)}\end{bmatrix}_1}
{\begin{bmatrix}\bar{\psi}_{n}^{(u,v)}(1)\end{bmatrix}_1}
\,\Omega_{j+l}^{(u,v)}\left(\bar{\Omega}_{l+k}^{(u,v)}-\bar{\Omega}_{n+k}^{(u,v)}\right)=0.
\end{split}
\end{equation}
Taking the summation for $k\ge m$ in \eqref{6.46} and using \eqref{6.24} we get
\begin{equation}\label{6.48}
\begin{split}
\sum_{k=m}^{\infty}
&\ds\frac{\begin{bmatrix}K_{nk}^{(u,v)}\end{bmatrix}_1}
{\begin{bmatrix}\bar{\psi}_{n}^{(u,v)}(1)\end{bmatrix}_1}
+\bar{G}_{n+m}^{(u,v)}
\\&-\sum_{l=n+1}^{\infty}
\ds\sum_{j=n+1}^{\infty}
\ds\frac{\begin{bmatrix}K_{nj}^{(u,v)}\end{bmatrix}_1}
{\begin{bmatrix}\bar{\psi}_{n}^{(u,v)}(1)\end{bmatrix}_1}\,
\Omega_{j+l}^{(u,v)}
\left(\bar{G}_{l+m}^{(u,v)}-\bar{G}_{n+m}^{(u,v)}\right)=0.
\end{split}
\end{equation}
Further, using \eqref{6.1bbbb}, \eqref{6.15}, and \eqref{6.24} in \eqref{6.48}, for $m>n$ we obtain
\begin{equation}\label{6.49}
\begin{split}
&\mathscr{K}_{nm}^{(u,v)}+\bar{G}_{n+m}^{(u,v)}
\\ &
-\sum_{l=n+1}^{\infty}\sum_{j=n+1}^{\infty}\left(\mathscr{K}_{nj}^{(u,v)}-
\mathscr{K}_{n(j+1)}^{(u,v)}\right)\left(G_{l+j}^{(u,v)}-G_{l+j+1}^{(u,v)}\right)
\left(\bar{G}_{l+m}^{(u,v)}-\bar{G}_{n+m}^{(u,v)}\right)=0.
\end{split}
\end{equation}
It is lengthy but straightforward to show that
\begin{equation}\label{6.53a}
\begin{split}
\sum_{l=n+1}^{\infty}\left(G_{l+j}^{(u,v)}-G_{l+j+1}^{(u,v)}\right)
&\left(\bar{G}_{l+m}^{(u,v)}-\bar{G}_{n+m}^{(u,v)}\right)\\&= \sum_{l=n+1}^{\infty}G_{j+l}^{(u,v)}\left(\bar{G}_{l+m}^{(u,v)}
-\bar{G}_{l+m-1}^{(u,v)}\right).
\end{split}
\end{equation}
Finally, using \eqref{6.53a} in \eqref{6.49} we obtain \eqref{6.22d}. The derivation of \eqref{6.23} is similarly obtained with the help of the $(2,1)$ and $(2,2)$ entries of \eqref{Tx.1} for the potential pair $(p,s).$
\end{proof}

\section{The solution to the inverse problem}
\label{sec:section8}

In this section we describe various methods to recover the potentials $q_n$ and $r_n$ when the scattering data set for \eqref{1.1} is available. We recall that the scattering data set consists of the scattering coefficients and the bound-state information. Because of Theorem~\ref{thm:theorem x2.4}, the four scattering coefficients $T^{(q,r)},$ $\bar{T}^{(q,r)},$ $R^{(q,r)},$ $\bar{R}^{(q,r)}$ contain all the information about the scattering coefficients for \eqref{1.1}. Similarly, because of Theorem~\ref{thm:theorem x4.1}, \eqref{TA.1}--\eqref{A4}, and \eqref{T.2}, the matrix triplets $(A, B, C^{(q,r)})$ and $(\bar{A}, \bar{B}, \bar{C}^{(q,r)})$ contain all the information related to the bound states of \eqref{1.1}. We let
\begin{equation}\label{8.1}
\mathbf{D}^{(q,r)}:= \{T^{(q,r)}, \bar{T}^{(q,r)}, R^{(q,r)}, \bar{R}^{(q,r)}, (A, B, C^{(q,r)}), (\bar{A}, \bar{B}, \bar{C}^{(q,r)})\},
\end{equation}
and refer to $\mathbf{D}^{(q,r)}$ as the scattering data set for \eqref{1.1}. Let us mention
that the relevant constants $D_\infty^{(q,r)}$ and $E_\infty^{(q,r)}$ are obtained from $T^{(q,r)}$ and $\bar{T}^{(q,r)}$
via \eqref{x.1}, and hence $D_\infty^{(q,r)}$ and $E_\infty^{(q,r)}$  are known if $\mathbf{D}^{(q,r)}$ is known.

Using the theory developed in Sections~\ref{sec:section2}--\ref{sec:section7}, we are able to solve the inverse problem for \eqref{1.1} in various ways, and we outline below some of those methods.

\begin{enumerate}

\item [\text{\rm(a)}] \textbf{The standard Marchenko method.} In this method, using the scattering data set $\mathbf{D}^{(q,r)}$ described in \eqref{8.1}, we construct the scalar quantities $\Omega_k^{(q,r)}$  and $\bar{\Omega}_k^{(q,r)}$ defined in \eqref{F.1} and use them as input to the Marchenko system \eqref{Z.0}. It can be proved in the standard way that \eqref{Z.0} is uniquely solvable
     via iteration. From the solution
$[\bar{M}_{nm}^{(q,r)}\quad M_{nm}^{(q,r)}]$ to \eqref{Z.0} we recover $q_n$ and $r_n$ via \eqref{Z.7} and \eqref{Z.7a}, respectively.

\item [\text{\rm(b)}] \textbf{The alternate Marchenko method.} In this method, using the scattering data set $\mathbf{D}^{(q,r)},$ we first obtain the constants $D_\infty^{(q,r)}$ and $E_\infty^{(q,r)}$  via \eqref{x.1}
    and also obtain $\Omega_k^{(q,r)}$  and $\bar{\Omega}_k^{(q,r)}$
   defined in \eqref{F.1}. Then, we
    construct the scalar quantities $\Omega_k^{(p,s)}$  and $\bar{\Omega}_k^{(p,s)}$ via \eqref{Z.5}.
    Moreover, using \eqref{R.12}, \eqref{R.12aaaa}, and \eqref{Z.5}
    we construct $\Omega_k^{(u,v)}$  and $\bar{\Omega}_k^{(u,v)}$ as
\begin{equation}
\label{ta4001}
\Omega_k^{(u,v)}=\ds\frac{D_\infty^{(q,r)}}{E_\infty^{(q,r)}}\left(\Omega_k^{(q,r)}-
\Omega_{k-2}^{(q,r)}\right)
,\quad \bar{\Omega}_k^{(u,v)}=\ds\frac{E_\infty^{(q,r)}}{D_\infty^{(q,r)}}\,
\ds\sum_{l=0}^\infty \bar{\Omega}_{k+2l}^{(q,r)}.
\end{equation}
 Next, we use \eqref{6.24} to obtain $G_k^{(u,v)}$ and $\bar{G}_k^{(u,v)}$ and use \eqref{6.25} to get $G_k^{(p,s)}$ and $\bar{G}_k^{(p,s)}.$ Using $G_k^{(u,v)}$ and $\bar{G}_k^{(u,v)}$ as input to the uncoupled alternate Marchenko equation \eqref{6.22d}, we obtain $\mathscr{K}_{nm}^{(u,v)}.$ Similarly, using $G_k^{(p,s)}$ and $\bar{G}_k^{(p,s)}$ as input to the uncoupled alternate Marchenko equation \eqref{6.23}, we obtain $\bar{\mathscr{K}}_{nm}^{(p,s)}.$ Finally, we recover the potentials $q_n$ and $r_n$ via \eqref{6.21} and \eqref{6.22}, respectively.

\item [\text{\rm(c)}] \textbf{Inversion with the help of the Marchenko system for \eqref{1.2aa}.} In this method, from the scattering data set $\mathbf{D}^{(q,r)}$ we first obtain the constants $D_\infty^{(q,r)}$ and
    $E_\infty^{(q,r)}$ via \eqref{x.1}
    and also obtain $\Omega_k^{(q,r)}$  and $\bar{\Omega}_k^{(q,r)}$
    defined in \eqref{F.1}.
    Then, we get $\Omega_k^{(u,v)}$  and $\bar{\Omega}_k^{(u,v)}$ via \eqref{ta4001}. Using  $\Omega_k^{(u,v)}$  and $\bar{\Omega}_k^{(u,v)}$ as input to the Marchenko system \eqref{Tx.1}, we obtain $K_{nm}^{(u,v)}$ and $\bar{K}_{nm}^{(u,v)}.$ Next, using \eqref{6.1bbbb} we recover the $2\times 2$ matrix
$[\bar{\psi}_n^{(u,v)}(1)\quad 	\psi_n^{(u,v)}(1)]$ from $K_{nm}^{(u,v)}$ and $\bar{K}_{nm}^{(u,v)}.$ Finally, we use \eqref{6.6} and \eqref{6.6kk} to recover the potentials $q_n$ and $r_n$, respectively.

\item [\text{\rm(d)}] \textbf{Inversion with the help of the Marchenko system for \eqref{1.2ab}.} In this method, using  \eqref{x.1}  we first obtain the constants $D_\infty^{(q,r)}$ and $E_\infty^{(q,r)}$
    and also obtain $\Omega_k^{(q,r)}$  and $\bar{\Omega}_k^{(q,r)}$
    defined in \eqref{F.1}
    from the scattering data set $\mathbf{D}^{(q,r)}.$  Then, we get  $\Omega_k^{(p,s)}$  and $\bar{\Omega}_k^{(p,s)}$ via \eqref{Z.5}. Next, using  $\Omega_k^{(p,s)}$  and $\bar{\Omega}_k^{(p,s)}$ as input to the Marchenko system \eqref{Tx.1}, we obtain
$[\bar{K}_{nm}^{(p,s)}\quad K_{nm}^{(p,s)}].$ Then, via \eqref{6.1cccc} we get
$[\bar{\psi}_n^{(p,s)}(1)\quad \psi_n^{(p,s)}(1)].$ Finally, we use \eqref{6.6hh} and \eqref{6.8} to recover the potentials $q_n$ and $r_n$, respectively.

\item [\text{\rm(e)}] \textbf{Inversion by first recovering the potentials $u_n$ and $s_n$.} In this method, from the scattering data set $\mathbf{D}^{(q,r)}$  we first obtain the constants $D_\infty^{(q,r)}$ and $E_\infty^{(q,r)}$  via  \eqref{x.1}
    and also obtain $\Omega_k^{(q,r)}$  and $\bar{\Omega}_k^{(q,r)}$
    defined in \eqref{F.1}.
    Then, we construct $\Omega_k^{(u,v)}$ and $\bar{\Omega}_k^{(u,v)}$
    via \eqref{ta4001} and
    also construct $\Omega_k^{(p,s)}$ and $\bar{\Omega}_k^{(p,s)}$ via \eqref{Z.5}. Next, using  $\Omega_k^{(u,v)}$ and $\bar{\Omega}_k^{(u,v)}$ as input in the uncoupled Marchenko equation given in the first line of \eqref{ta002} related to $(u,v)$, we obtain $[K_{nm}^{(u,v)}]_1,$ from which we recover $u_n$ as in the first equality in \eqref{Tx.8}. Similarly, using  $\Omega_k^{(p,s)}$ and $\bar{\Omega}_k^{(p,s)}$ as input in the uncoupled Marchenko equation given in the second line of \eqref{ta002} related to $(p,s)$, we
    obtain $[\bar{K}_{nm}^{(p,s)}]_2,$ from which we recover $s_n$ as in the second equality in \eqref{Tx.9}. Finally, we use \eqref{x.603} and \eqref{x.604} with input $(u_n, s_n)$ and recover the potentials $q_n$ and $r_n.$

\end{enumerate}

\section{The time evolution of the scattering data}
\label{sec:section9}

In this section we consider an application of our results
in integrable semi-discrete systems, and we
provide the solution to the nonlinear system \eqref{1.2a} via the method of the inverse scattering transform. This is done by describing the time evolution of the scattering data for \eqref{1.1} and determining the corresponding time-evolved potentials $q_n$ and $r_n.$ Hence, each of the methods to solve the inverse problem for \eqref{1.1} presented in Section~\ref{sec:section8} can be used to solve \eqref{1.2a} if we replace the scattering data set $\mathbf{D}^{(q,r)}$ appearing in \eqref{8.1} with its time-evolved version.
In this section, we also present certain solution formulas for \eqref{1.2a}
expressed explicitly in terms of the matrix triplets
$(A, B, C)$ and $(\bar{A}, \bar{B}, \bar{C})$ for the linear system
\eqref{1.1}. Such solution formulas
correspond to the reflectionless scattering data for \eqref{1.1}, in which case
the corresponding Marchenko integral system for
\eqref{1.1} has separable kernels
and hence is solved in closed form by using standard linear algebraic
methods.

Let us mention
\cite{tsuchida2010new}
that the system \eqref{1.2a} is the semi-discrete analog of the nonlinear system
\begin{equation}\label{1.2b}
\begin{cases}
iq_t+q_{xx}-i(q\,r\,q)_x=0,\\
\noalign{\medskip}
ir_t-r_{xx}-i(r\,q\,r)_x=0,\\
\end{cases}
\end{equation}
where $q(x,t)$ and $r(x,t)$ are the continuous analogs of $q_n$ and $r_n$ when the latter quantities depend on both $n$ and $t$. The nonlinear system \eqref{1.2b} is known as the derivative NLS system or  the Kaup-Newell system.

It is already known \cite{ablowitz149inverse,tsuchida2010new} that \eqref{1.2a} can be derived by imposing the compatibility condition
\begin{equation}\label{1.2c}
\dot{\mathcal{X}}_n+\mathcal{X}_n\,\mathcal{T}_{n+1}-\mathcal{T}_n\,\mathcal{X}_n=0,
\end{equation}
where $(\mathcal{X}_n,\mathcal{T}_n)$ is the AKNS pair with $\mathcal{X}_n$ being the 2$\times$2 coefficient matrix appearing in \eqref{1.1} and $\mathcal{T}_n$ is the 2$\times$2 matrix given by
\begin{equation*}
\mathcal{T}_n=\begin{bmatrix}
\ds\frac{-i(z^2-1)[1+(z^2+1)\,q_{n-1}\,r_n]}{z^2(1+q_{n-1}\,r_n)} & \ds\frac{i(z^2-1)q_{n-1}}{1+q_{n-1}\,r_n}-\ds\frac{i(z^2-1)q_n}{z^2(1-q_n\,r_n)}\\
\noalign{\medskip}
\ds\frac{-ir_{n-1}}{1-q_{n-1}\,r_{n-1}}+\ds\frac{i\,z^2\,r_n}{1+q_{n-1}\,r_n}& \ds\frac{i(z^2-1)}{1+q_{n-1}\,r_n}
\end{bmatrix},
\end{equation*}
which plays the key role in the time evolution of the potential pair $(q_n,r_n).$
We recall that an overdot denotes the derivative with respect to $t.$
Let us remark that the AKNS pair for a given nonlinear system is not unique. One can use the transformation
\begin{equation*}
\begin{cases}
\Psi_n \mapsto \tilde{\Psi}_n:= \mathcal{G}_n\Psi_n,\\
\noalign{\medskip}
\mathcal{X}_n\mapsto\tilde{\mathcal{X}_n}:= \mathcal{G}_n\,\mathcal{X}_n\,\mathcal{G}_{n+1}^{-1}, \\
\noalign{\medskip}
\mathcal{T}_n\mapsto\tilde{\mathcal{T}_n}:= \dot{\mathcal{G}}_n\,\mathcal{G}_n^{-1}+\mathcal{G}_n\mathcal{T}_n\,\mathcal{G}_{n}^{-1},
\end{cases}
\end{equation*}
for any appropriate invertible matrix $\mathcal{G}_n$, and the corresponding compatibility condition
\begin{equation*}
\dot{\mathcal{\tilde{X}}}_n+\mathcal{\tilde{X}}_n\,
\mathcal{\tilde{T}}_{n+1}-\mathcal{\tilde{T}}_n\,\mathcal{\tilde{X}}_n=0,
\end{equation*}
yields the same integrable nonlinear system that \eqref{1.2c} yields. Since the choice of $\mathcal{X}_n$ is not unique, instead of analyzing the linear system
\begin{equation*}
\Psi_n=\mathcal{X}_n\,\Psi_{n+1},
\end{equation*}
one can alternatively analyze the system
\begin{equation*}
\tilde{\Psi}_n=\mathcal{\tilde{X}}_n\,\tilde{\Psi}_{n+1}.
\end{equation*}

The linear system \eqref{1.2aa} is associated with the integrable nonlinear system
\begin{equation}\label{x.8}
\begin{cases}
i\dot{u}_n+u_{n-1}-2u_n+u_{n+1}-u_{n-1}\,u_n\,v_n-u_n\,u_{n+1}\,v_n=0,\\
\noalign{\medskip}
i\dot{v}_n-v_{n-1}+2v_n-v_{n+1}+u_n\,v_{n-1}\,v_n+u_n\,v_n\,v_{n+1}=0.
\end{cases}
\end{equation}
The AKNS pair $(\mathcal{X}_n,\mathcal{T}_n)$ for \eqref{x.8}
consists of the matrix $\mathcal{X}_n$ appearing as the coefficient matrix in \eqref{1.2aa} and
the matrix $\mathcal{T}_n$ given by
$$\mathcal{T}_n=
\begin{bmatrix}
i(z^2-1)-i\,u_{n-1}\,v_n&i\, z^2\, u_{n-1}-i\, u_n
\\
\noalign{\medskip}
-\ds\frac{i}{z^2}\,v_{n-1}+i\, v_n&
i\left(1-\ds\frac{1}{z^2}\right)+iu_n\, v_{n-1}
\end{bmatrix}.$$
Similarly, the linear system \eqref{1.2ab} is associated with the integrable system \eqref{x.8} with $(u_n, v_n)$ replaced by $(p_n, s_n)$ there.

In the following theorem we summarize the time evolution of the scattering data for \eqref{1.2aa}. A proof is omitted because the time evolution of the scattering coefficients is described in
\cite{tsuchida2010new} and the time evolution of
the norming constants for simple bound states
described in \cite{tsuchida2010new} is readily
generalized to the case of non-simple bound states and hence
to the time evolution of the matrix triplets.

\begin{theorem}
	\label{thm:theorem x9.1}
Assume that the potentials $u_n$ and $v_n$ appearing in \eqref{1.2aa} and \eqref{x.8} are rapidly decaying and $1-u_n v_n\ne 0$ for $n\in\mathbb{Z}$. Then, the corresponding reflection coefficients evolve in time as
\begin{equation*}
\begin{cases}
R^{(u,v)}\mapsto R^{(u,v)}\,e^{-it(z-z^{-1})^2},\quad \bar{R}^{(u,v)}
\mapsto \bar{R}^{(u,v)}\,e^{it(z-z^{-1})^2},\\
\noalign{\medskip}
L^{(u,v)}\mapsto L^{(u,v)}\,e^{it(z-z^{-1})^2},\quad \bar{L}^{(u,v)}
\mapsto \bar{L}^{(u,v)}\,e^{-it(z-z^{-1})^2},
\end{cases}
\end{equation*}
and the transmission coefficients $T_{\rm l}^{(u,v)},$ $T_{\rm r}^{(u,v)},$ $\bar{T}_{\rm l}^{(u,v)},$ $\bar{T}_{\rm r}^{(u,v)}$ do not change in time. Furthermore, in the corresponding matrix triplets $(A, B, C^{(u,v)})$ and $(\bar{A}, \bar{B}, \bar{C}^{(u,v)})$ describing the bound-state data for \eqref{1.2aa}, the row vectors $C^{(u,v)}$ and $\bar{C}^{(u,v)}$ evolve in time as
\begin{equation}\label{x.9a}
C^{(u,v)}\mapsto C^{(u,v)}\,e^{-it(A-A^{-1})^2},\quad \bar{C}^{(u,v)}
\mapsto \bar{C}^{(u,v)}\,e^{it [\bar{A}-(\bar{A})^{-1}]^2},
\end{equation}
and the matrices $A,$ $\bar{A},$ $B,$ $\bar{B}$ do not change in time. Moreover, the constant $D_\infty^{(u,v)}$ appearing in \eqref{x2.D_n} does not change in time, either.
\end{theorem}

We remark that Theorem~\ref{thm:theorem x9.1} holds in the same way for the system \eqref{1.2ab} with the potential pair $(p,s).$
Next, we present the time evolution of the scattering data for \eqref{1.1}.

\begin{theorem}
	\label{thm:theorem x9.2}
	Assume that the potentials $q_n$ and $r_n$ appearing in \eqref{1.1} and \eqref{1.2a} are rapidly decaying  and satisfy \eqref{1.1a}. Then, the corresponding reflection coefficients evolve in time as
	\begin{equation}\label{x.9b}
	\begin{cases}
	R^{(q,r)}\mapsto R^{(q,r)}\,e^{-it(z-z^{-1})^2},\quad \bar{R}^{(q,r)}
	\mapsto \bar{R}^{(q,r)}\,e^{it(z-z^{-1})^2},\\
	\noalign{\medskip}
	L^{(q,r)}\mapsto L^{(q,r)}\,e^{it(z-z^{-1})^2},\quad \bar{L}^{(q,r)}
	\mapsto \bar{L}^{(q,r)}\,e^{-it(z-z^{-1})^2},
	\end{cases}
	\end{equation}
	and the corresponding transmission coefficients $T^{(q,r)}$ and $\bar{T}^{(q,r)}$  do not change in time. Furthermore, in the matrix triplets $(A, B, C^{(q,r)})$ and $(\bar{A}, \bar{B}, \bar{C}^{(q,r)})$ describing the bound-state data for \eqref{1.1}, the row vectors $C^{(q,r)}$ and $\bar{C}^{(q,r)}$ evolve in time according to
	\begin{equation}\label{x.9ab}
	C^{(q,r)}\mapsto C^{(q,r)}\,e^{-it(A-A^{-1})^2},\quad \bar{C}^{(q,r)}
	\mapsto \bar{C}^{(q,r)}\,e^{it [\bar{A}-(\bar{A})^{-1} ]^2}.
	\end{equation}
Moreover, neither of constants $D_\infty^{(q,r)}$ and $E_\infty^{(q,r)}$  appearing in \eqref{x2.42} and \eqref{x2.43}, respectively, changes in time.
\end{theorem}

\begin{proof}
Let us first prove that $D_\infty^{(q,r)}$  does not change in time, i.e. we have $\dot{D}_\infty^{(u,v)}=0,$ where we recall that we use an overdot to denote the time derivative. From the second equality in \eqref{x2.42} and the fact that $D_\infty^{(q,r)}\ne 0,$ we see that $\dot{D}_\infty^{(q,r)}=0$ if and only if $\dot{D}_\infty^{(q,r)}/D_\infty^{(q,r)}=0,$ which is equivalent to having
\begin{equation}\label{x.9ac}
\ds\sum_{n=-\infty}^{\infty}\ds\frac{\dot{q}_n\,r_n+q_n\,\dot{r}_n}{1-q_n\,r_n}=0.
\end{equation}
In order to prove that \eqref{x.9ac} holds, we multiply the first line of \eqref{1.2a} with $r_n$ and the second line of \eqref{1.2a}  with $q_n$ and then we add the resulting equations. Using the summation over $n,$ after some straightforward simplifications, we get
\begin{equation}\label{x.9ad}
\ds\sum_{n=-\infty}^{\infty}\ds\frac{\dot{q}_n\,r_n+q_n\,\dot{r}_n}{1-q_n\,r_n}=
\ds\sum_{n=-\infty}^{\infty}\left(\Delta_{n+1}-\Delta_n\right),
\end{equation}
where we have let
\begin{equation}\label{x.9ae}
\Delta_n:=i\left[\ds\frac{1}{1+q_{n-1}\,r_n}-1+\ds\frac{q_n\,r_{n-1}}{(1-q_{n-1}\,r_{n-1})(1-q_n\,r_n)}\right].
\end{equation}
Since the potentials $q_n$ and $r_n$ are rapidly decaying as $n\to\pm\infty$ and satisfy \eqref{1.1a}, from \eqref{x.9ae} we see that $\Delta_n$ is well defined and rapidly decaying as $n\to\pm\infty.$ Hence, the telescoping series on the right-hand side of \eqref{x.9ad} converges to $0,$ which completes the proof that $\dot{D}_\infty^{(q,r)}=0.$ The proof of $\dot{E}_\infty^{(q,r)}=0$ is obtained in a similar manner by establishing that $\dot{E}_\infty^{(q,r)}/E_\infty^{(q,r)}=0,$ which is equivalent to having
\begin{equation}\label{x.9af}
\ds\sum_{n=-\infty}^{\infty}\ds\frac{\dot{q}_n\,r_{n+1}+q_n\,\dot{r}_{n+1}}{1+q_n\,r_{n+1}}=0.
\end{equation}
In order to prove \eqref{x.9af}, we replace $n$ by $n+1$ in the second line of \eqref{1.2a} and multiply the resulting equation by $q_n,$ and then we add to that equation the first line of \eqref{1.2a} multiplied by $r_{n+1}.$ Then, a summation over $n$, after some straightforward simplifications, yields
\begin{equation}\label{x.9ah}
\ds\sum_{n=-\infty}^{\infty}\ds\frac{\dot{q}_n\,r_{n+1}+q_n\,\dot{r}_{n+1}}{1+q_n\,r_{n+1}}=
\ds\sum_{n=-\infty}^{\infty}\left(\Theta_{n+1}-\Theta_n\right),
\end{equation}
where we have let
\begin{equation*}
\Theta_n:=i\left[\ds\frac{1}{1-q_n\,r_n}-1-\ds\frac{q_{n-1}\,r_{n+1}}
{(1+q_{n-1}\,r_n)(1+q_n\,r_{n+1})}\right].
\end{equation*}
From the properties of $q_n$ and $r_n,$ it follows that $\Theta_n$ is well defined and rapidly decaying as $n\to\pm\infty.$ Thus, the telescoping series in \eqref{x.9ah} is convergent to $0,$ which establishes the proof that $\dot{E}_\infty^{(q,r)}=0.$ When the potential pairs $(q,r),$ $(u,v),$ $(p,s)$ are related to each other as in \eqref{x3.1}--\eqref{x3.4}, we have the matrices $A,$ $\bar{A},$ $B,$ $\bar{B}$ appearing in \eqref{Tx.3} are common and the scattering coefficients for \eqref{1.1}, \eqref{1.2aa}, \eqref{1.2ab} are related as described in Theorem~\ref{thm:theorem x3.3}. Thus, with the help of Theorem~\ref{thm:theorem x3.3}, Theorem~\ref{thm:theorem x9.1}, and the fact that $\dot{D}_\infty^{(q,r)}=0$  and $\dot{E}_\infty^{(q,r)}=0,$ we conclude that the transmission coefficients $T^{(q,r)}$ and $\bar{T}^{(q,r)}$ do not change in time and the reflection coefficients evolve as in \eqref{x.9b}. Furthermore, from  \eqref{T.4aaa}, \eqref{x.9a}, and the fact that  $\dot{D}_\infty^{(q,r)}=0$  and $\dot{E}_\infty^{(q,r)}=0,$ we obtain \eqref{x.9ab}.
\end{proof}

Next we consider explicit solutions to the integrable systems \eqref{1.2a} and \eqref{x.8} by using the method of \cite{aktosun2007exact,Aktosunkdv}. Such explicit solutions correspond to the zero reflection coefficients and the time-evolved scattering data sets. From Theorem~\ref{thm:theorem x9.1} and Theorem~\ref{thm:theorem x9.2} we see that the matrix triplets corresponding to \eqref{1.2a} and \eqref{x.8} have similar time evolutions described as
\begin{equation}\label{x9.1}
(A,B,C)\mapsto (A,B,C\,\mathcal{E}) ,\quad (\bar{A},\bar{B},\bar{C})\mapsto (\bar{A},\bar{B},\bar{C}\,\bar{\mathcal{E}}),
\end{equation}
where we have defined
\begin{equation}\label{x9.2}
\mathcal{E}:=e^{-it(A-A^{-1})^2},\quad \bar{\mathcal{E}}:= e^{it [\bar{A}-(\bar{A})^{-1} ]^2}.
\end{equation}
Let us remark that \eqref{Tx.1} and \eqref{Tx.3} for the potential pair $(u,v)$ and \eqref{Z.0} and \eqref{F.1} for the potential
pair $(q,r)$ are similar, and hence
the solution to \eqref{Z.0} is obtained in a similar way the solution to \eqref{Tx.1} is obtained.

Our goal now is to present the corresponding explicit solutions to  \eqref{Tx.1} and  \eqref{Z.0}  when their Marchenko kernels are given by
\begin{equation}\label{x9.3}
\Omega_{n+m}=C\mathcal{E}A^{n+m-1}B,\quad \bar{\Omega}_{n+m}=\bar{C}\bar{\mathcal{E}}(\bar{A})^{-n-m-1}\bar{B}.
\end{equation}
Note that we impose no restriction on the values of $N,$ $\bar{N},$ $z_j,$ $\bar{z}_j,$ $m_j,$ $\bar{m}_j$ in the matrix triplets $(A, B, C)$ and $(\bar{A}, \bar{B}, \bar{C})$ appearing in \eqref{TA.1} and  \eqref{TA.2}, respectively. Hence, this method yields an enormous number of explicit solutions to each of \eqref{Tx.1} and  \eqref{Z.0}.
From \eqref{x9.3} we see that the Marchenko kernels $\Omega_{n+m}$ and $\bar{\Omega}_{n+m}$ are separable in $n$ and $m,$ i.e. we can write them as the matrix products given by
\begin{equation}\label{x9.4}
\Omega_{n+m}= \left[C\,A^{n}\right]
\left[\mathcal{E}\,A^{m-1}\,B\right],\quad
\bar{\Omega}_{n+m}=\left[\bar{C}\,(\bar{A})^{-n}\right]
\left[\bar{\mathcal{E}}(\bar{A})^{-m-1}\,\bar{B}\right],
\end{equation}
where we have used the fact that the matrices $A$ and $\mathcal{E}$ commute and that the matrices
$\bar{A}$ and $\bar{\mathcal{E}}$ commute.

Before we present our explicit solutions to \eqref{Tx.1} and \eqref{Z.0}, we introduce some auxiliary quantities. In terms of the positive integers $m_j,$ $N,$ $\bar{m}_j,$ $\bar{N}$ appearing in \eqref{TA.1}--\eqref{A4}, we introduce the positive integers $\mathscr{N}$ and $\bar{\mathscr{N}}$ as
\begin{equation}\label{x9.5}
\mathscr{N}:=\ds\sum_{j=1}^{N}m_j,\quad \bar{\mathscr{N}}:=\ds\sum_{j=1}^{\bar{N}}\bar{m}_j.
\end{equation}
From the results in Section~\ref{sec:section4} it follows that $2(\mathscr{N}+\bar{\mathscr{N}})$ corresponds to the total number of bound states including the multiplicities for \eqref{1.1} and \eqref{1.2aa}. In terms of the matrix triplets $(A, B, C)$ and $(\bar{A}, \bar{B}, \bar{C})$ let us introduce the $\mathscr{N}\times \bar{\mathscr{N}}$ matrix $\Upsilon$ and the $\bar{\mathscr{N}}\times \mathscr{N}$ matrix $\bar{\Upsilon}$ as
\begin{equation}\label{x9.6}
\Upsilon:=\ds\sum_{k=0}^{\infty}A^{k}\,B\,\bar{C}\,(\bar{A})^{-k},\quad \bar{\Upsilon}:=\ds\sum_{k=0}^{\infty}(\bar{A})^{-k}\,\bar{B}\,C\,A^{k}.
\end{equation}
In terms of the two matrix triplets let us also define
the $\mathscr{N}\times \mathscr{N}$ matrix $U_n$ and
the $\bar{\mathscr{N}}\times \bar{\mathscr{N}}$ matrix
$\bar{U}_n$ as
\begin{equation}\label{x9.12}
U_n:=I-\bar{\mathcal{E}}\,(\bar{A})^{-n-2}\,
\bar{\Upsilon}\,\mathcal{E}\,A^{2n+1}\,\Upsilon\, (\bar{A})^{-n-1},
\end{equation}
\begin{equation}\label{x9.13}
\bar{U}_n:=I-\mathcal{E}A^{n}
\Upsilon\bar{\mathcal{E}}(\bar{A})^{-2n-3}\,\bar{\Upsilon}\,A^{n+1},
\end{equation}
where we recall that the $\mathscr{N}\times \mathscr{N}$ matrix
$\mathcal{E}$ and
the $\bar{\mathscr{N}}\times \bar{\mathscr{N}}$
matrix $\bar{\mathcal{E}}$ are defined in \eqref{x9.2}.

In the next proposition we elaborate on the matrices
$\Upsilon$ and $\bar{\Upsilon}.$

\begin{proposition}
	\label{thm:theorem x9.3}
	Let  $(A, B, C)$ and $(\bar{A}, \bar{B}, \bar{C})$ be the matrix triplets appearing in \eqref{TA.1}--\eqref{A4} with $\left|z_j\right|<1$ for $1\le j \le N$ and $\left|\bar{z}_j\right|>1$ for $1\le j \le \bar{N}.$ Then, the matrices $\Upsilon$ and $\bar{\Upsilon}$ defined in \eqref{x9.6} are the unique solutions to the respective linear systems
	\begin{equation}\label{x9.7}
\Upsilon-A\Upsilon(\bar{A})^{-1}=B\,\bar{C},\quad
\bar{ \Upsilon}-(\bar{A})^{-1}\,\bar{\Upsilon}A=\bar{B}\,C.
	\end{equation}
\end{proposition}

\begin{proof}
By premultiplying the first equality in \eqref{x9.6} by $A$ and postmultiplying it by $(\bar{A})^{-1}$ and subtracting the resulting matrix equality from the original equality, we obtain the first linear system in \eqref{x9.7}. The second equality in \eqref{x9.7} is similarly obtained from the second equality in \eqref{x9.6}.
The existence and uniqueness of the solutions to the two matrix systems in \eqref{x9.7} can be analyzed as in Theorem~$18.2$ of \cite{Dym}. Given the matrix triplets $(A, B, C)$ and $(\bar{A}, \bar{B}, \bar{C}),$ we have the unique solutions $\Upsilon$ and $\bar{\Upsilon}$ to \eqref{x9.7} if and only if the product of an eigenvalue of  $A$ and an eigenvalue of  $(\bar{A})^{-1}$ is never equal to $1.$ The satisfaction of the latter condition directly follows from the fact that  $\left|z_j\right|<1$ for $1\le j \le N$ and $\left|\bar{z}_j\right|>1$ for $1\le j \le \bar{N}.$ Thus, the solutions $\Upsilon$ and $\bar{\Upsilon}$ to \eqref{x9.7} are unique and given by \eqref{x9.6}.
\end{proof}

Next, we present the explicit solution formula for the Marchenko system \eqref{Tx.1} corresponding to the Marchenko kernels given in \eqref{x9.4} for the potential
pair $(u,v).$

\begin{theorem}
\label{thm:theorem x9.4}
Using the time-evolved reflectionless Marchenko kernels $\Omega_{n+m}^{(u,v)}$ and $\bar{\Omega}_{n+m}^{(u,v)}$ that have the form as
in \eqref{x9.4} for the potential pair
$(u,v),$ the corresponding Marchenko system \eqref{Tx.1}, in the notation of \eqref{ta001}, has the solution given by
\begin{equation}\label{x9.8}
\begin{bmatrix}K_{nm}^{(u,v)}\end{bmatrix}_1
=-\bar{C}^{(u,v)}\,(\bar{A})^{-n}\,(U_n^{(u,v)})^{-1}\,\bar{\mathcal{E}}\,(\bar{A})^{-m-1}\,\bar{B},
\end{equation}
\begin{equation}\label{x9.9}
\begin{bmatrix}K_{nm}^{(u,v)}\end{bmatrix}_2
=C^{(u,v)}\,A^{n}\,(\bar{U}_n^{(u,v)})^{-1}\,\mathcal{E}\,A^{n}\,
\Upsilon^{(u,v)}\,\bar{\mathcal{E}}\,(\bar{A})^{-m-n-2}\,\bar{B},
\end{equation}
\begin{equation}\label{x9.10}
\begin{bmatrix}\bar{K}_{nm}^{(u,v)}\end{bmatrix}_1
=\bar{C}^{(u,v)}\,(\bar{A})^{-n}\,(U_n^{(u,v)})^{-1}\,
\bar{\mathcal{E}}\,(\bar{A})^{-n-2}\,
\bar{\Upsilon}^{(u,v)}\,\mathcal{E}\,(A)^{n+m}\,B,
\end{equation}
\begin{equation}\label{x9.11}
\begin{bmatrix}\bar{K}_{nm}^{(u,v)}\end{bmatrix}_2
=-C^{(u,v)}\,A^{n}\,(\bar{U}_n^{(u,v)})^{-1}\,\mathcal{E}\,A^{m-1}\,B,
\end{equation}
where $\Upsilon^{(u,v)}$ and $\bar{\Upsilon}^{(u,v)}$ are the matrices appearing in \eqref{x9.6} for the potential pair
$(u,v)$ and the matrices $U_n^{(u,v)}$ and $\bar{U}_n^{(u,v)}$ are defined as in
\eqref{x9.12} and \eqref{x9.13} for the potential pair
$(u,v).$
\end{theorem}

\begin{proof}
For simplicity, we suppress the superscript $(u,v)$ in the proof.
We already know that the Marchenko system \eqref{Tx.1} is equivalent to the combination of the uncoupled system \eqref{ta002} and the system \eqref{ta1001}. To obtain \eqref{x9.8} we proceed as follows. Using \eqref{x9.4} as input to the first line of \eqref{ta002} we get
\begin{equation}\label{x9.14}
\begin{split}
&\begin{bmatrix} K_{nm}\end{bmatrix}_1
+\bar{C}\,(\bar{A})^{-n}\,\bar{\mathcal{E}}\,(\bar{A})^{-m-1}\,\bar{B}
\\&
\phantom{xxx}
-\ds\sum_{l=n+1}^{\infty}
\ds\sum_{j=n+1}^{\infty}\begin{bmatrix} K_{nj}\end{bmatrix}_1\,
C\,\mathcal{E}\,A^{j+l-1}\,
B\,\bar{C}\,(\bar{A})^{-l}\,
\bar{\mathcal{E}}\,(\bar{A})^{-m-1}\,\bar{B}=0,
\end{split}
\end{equation}
where we have used the fact that $\mathcal{E}$ and $A$ commute and $\bar{\mathcal{E}}$ and $\bar{A}$ commute. From \eqref{x9.14} we see that $[K_{nm}]_1$ has the form
\begin{equation}\label{x9.15}
\begin{bmatrix} K_{nm}\end{bmatrix}_1=H_n\,\bar{\mathcal{E}}(\bar{A})^{-m-1}\bar{B},
\end{equation}
where $H_n$ satisfies
\begin{equation}\label{x9.16}
H_n\left(I-\ds\sum_{l=n+1}^{\infty}
\ds\sum_{j=n+1}^{\infty}
\bar{\mathcal{E}}\,(\bar{A})^{-j-1}\,\bar{B}\,C\mathcal{E}\,A^{j+l-1}\,
B\,\bar{C}\,(\bar{A})^{-l}
\right)=-\bar{C}\,(\bar{A})^{-n}.
\end{equation}
Using \eqref{x9.6} on the left-hand side of \eqref{x9.16}, we write \eqref{x9.16} as
\begin{equation}\label{x9.17}
H_n\,U_n=-\bar{C}\,(\bar{A})^{-n},
\end{equation}
where $U_n$ is the matrix defined in \eqref{x9.12}. From \eqref{x9.17} we get
\begin{equation}\label{x9.18}
H_n=-\bar{C}\,(\bar{A})^{-n}\,(U_n)^{-1},
\end{equation}
and using \eqref{x9.18} in \eqref{x9.15} we obtain \eqref{x9.8}. The solution formula for $\begin{bmatrix}\bar{K}_{nm}\end{bmatrix}_2$ appearing in \eqref{x9.11} is obtained in a similar manner by using the second line of \eqref{ta002}. Then, using \eqref{x9.4} and \eqref{x9.11} in the second line of \eqref{ta1001} we obtain the formula for $[K_{nm}]_2$ given in \eqref{x9.9}. Similarly, by using \eqref{x9.4} and \eqref{x9.8} in the first line of \eqref{ta1001}, we obtain the formula for $[\bar{K}_{nm}]_1$ given in \eqref{x9.10}.
\end{proof}

We remark that the result
 of Theorem~\ref{thm:theorem x9.4} remains
 valid for the Marchenko system \eqref{Z.0} because of the resemblance between \eqref{Tx.1} and \eqref{Z.0}
and the fact that
\eqref{x9.1} and \eqref{x9.2} have the same appearance for
the potential pairs $(u,v)$ and $(q,r).$ So, without a proof we state
the result in the next corollary.

\begin{corollary}
\label{thm:theorem x9.5}
Using the time-evolved reflectionless Marchenko kernels $\Omega_{n+m}^{(q,r)}$ and $\bar{\Omega}_{n+m}^{(q,r)}$ that have the form as
in \eqref{x9.4} for the potential pair
$(q,r),$ the corresponding Marchenko system \eqref{Z.0}, in the notation of \eqref{ta001}, has the solution given by
\begin{equation}\label{x9.8aa}
\begin{bmatrix}M_{nm}^{(q,r)}\end{bmatrix}_1
=-\bar{C}^{(q,r)}\,(\bar{A})^{-n}\,(U_n^{(q,r)})^{-1}\,\bar{\mathcal{E}}\,(\bar{A})^{-m-1}\,\bar{B},
\end{equation}
\begin{equation}\label{x9.9aa}
\begin{bmatrix}M_{nm}^{(q,r)}\end{bmatrix}_2
=C^{(q,r)}\,A^{n}\,(\bar{U}_n^{(q,r)})^{-1}\,\mathcal{E}\,A^{n}\,
\Upsilon^{(q,r)}\,\bar{\mathcal{E}}\,(\bar{A})^{-n-m-2}\,\bar{B},
\end{equation}
\begin{equation}\label{x9.10aa}
\begin{bmatrix}\bar{M}_{nm}^{(q,r)}\end{bmatrix}_1
=\bar{C}^{(q,r)}\,(\bar{A})^{-n}\,(U_n^{(q,r)})^{-1}\,
\bar{\mathcal{E}}\,(\bar{A})^{-n-2}\,
\bar{\Upsilon}^{(q,r)}\,\mathcal{E}\,(A)^{n+m}\,B,
\end{equation}
\begin{equation}\label{x9.11aa}
\begin{bmatrix}\bar{M}_{nm}^{(q,r)}\end{bmatrix}_2
=-C^{(q,r)}\,A^{n}\,(\bar{U}_n^{(q,r)})^{-1}\,\mathcal{E}\,A^{m-1}\,B,
\end{equation}
where $\Upsilon^{(q,r)}$ and $\bar{\Upsilon}^{(q,r)}$ are the matrices appearing in \eqref{x9.6} for the potential pair
$(q,r)$ and the matrices $U_n^{(q,r)}$ and $\bar{U}_n^{(q,r)}$ are defined as
\eqref{x9.12} and \eqref{x9.13} for the potential pair
$(q,r).$
\end{corollary}

In the next proposition, when  \eqref{x3.1}--\eqref{x3.4} hold, we show how some relevant quantities for the potential pairs $(u,v)$ and $(p,s)$ are related to the corresponding quantities for the potential pair $(q,r).$
These results will enable us to obtain explicit solutions to
the nonlinear system \eqref{1.2a} by using the input data directly related to
the potential pair $(q,r).$

\begin{proposition}
	\label{thm:theorem9.5a}
	Assume that the potentials $q_n$ and $r_n$ appearing in \eqref{1.1}  are rapidly decaying and satisfy \eqref{1.1a}. Assume further that the potential pairs $(u,v)$ and $(p,s)$ are related to the potential pair $(q,r)$ as in \eqref{x3.1}--\eqref{x3.4}. Then, we have the following:

	\begin{enumerate}

		\item [\text{\rm(a)}] The matrices $\Upsilon$ and $\bar{ \Upsilon}$ appearing in \eqref{x9.6} corresponding to the potential pairs $(u,v)$ and $(p,s)$ are related to those with the potential pair $(q,r)$ as
		\begin{equation}\label{9.4}
		\begin{cases}
	   \Upsilon^{(u,v)}=\ds\frac{E_\infty^{(q,r)}}{D_\infty^{(q,r)}}
        \,\Upsilon^{(q,r)}\left[I-(\bar{A})^{-2}\right]^{-1},\\
        \noalign{\medskip}
  	   \bar{\Upsilon}^{(u,v)}=	\ds\frac{D_\infty^{(q,r)}}{E_\infty^{(q,r)}}
	   \,\bar{ \Upsilon}^{(q,r)}\,\left(I-A^{-2}\right),
	   \end{cases}
    	\end{equation}
	 \begin{equation}\label{9.5}
	\Upsilon^{(p,s)}=\ds\frac{E_\infty^{(q,r)}}{D_\infty^{(q,r)}}\,\Upsilon^{(q,r)},\quad
	\bar{\Upsilon}^{(p,s)}=	\ds\frac{D_\infty^{(q,r)}}{E_\infty^{(q,r)}}\,\bar{ \Upsilon}^{(q,r)},
	\end{equation}
	where $D_\infty^{(q,r)}$ and $E_\infty^{(q,r)}$ are the constants appearing in \eqref{x2.42} and \eqref{x2.43}, respectively.
		
		\item [\text{\rm(b)}] The matrices $U_n$ and $\bar{U}_n$ appearing in \eqref{x9.12} and \eqref{x9.13}, respectively,  corresponding to the potential pairs $(u,v)$ and $(p,s)$ are related to the quantities relevant to the potential pair $(q,r)$ as
		\begin{equation}\label{9.6}
\begin{split}
		&U_n^{(u,v)}
=I\\
&\phantom{x}+\bar{\mathcal{E}}\,(\bar{A})^{-n-2}\,
\bar{\Upsilon}^{(q,r)}\,\mathcal{E}\,A^{2n-1}\left(I-A^{2}\right)\Upsilon^{(q,r)} \, (\bar{A})^{-n-1}\left[I-(\bar{A})^{-2}\right]^{-1},
\end{split}
		\end{equation}
		\begin{equation}\label{9.7}
\begin{split}
		&
		\bar{U}_n^{(u,v)}
=I\\
&\phantom{xx}+\mathcal{E}\,A^{n}\,
		\Upsilon^{(q,r)}\,\bar{\mathcal{E}}\,(\bar{A})^{-2n-3}
\left[I-(\bar{A})^{-2}\right]^{-1}\bar{ \Upsilon}^{(q,r)}\,A^{n-1}\left(I-A^{2}\right),
\end{split}
		\end{equation}
		\begin{equation}\label{9.8}
		U_n^{(p,s)}=U_n^{(q,r)},\quad \bar{U}_n^{(p,s)}=\bar{U}_n^{(q,r)},
		\end{equation}
		where $\Upsilon^{(q,r)}$ and $\bar{\Upsilon}^{(q,r)}$ are the matrices appearing in \eqref{x9.6} for the potential pair $(q,r).$
	\end{enumerate}
\end{proposition}

\begin{proof}
Using \eqref{T.4aaa} in \eqref{x9.6} we get \eqref{9.4}. Similarly, using \eqref{T.4} in \eqref{x9.6} we have \eqref{9.5}. Next, using \eqref{9.4} in \eqref{x9.12} we obtain \eqref{9.6}. Then, using \eqref{9.4} in \eqref{x9.13} we get \eqref{9.7}. Finally, using \eqref{9.5} in \eqref{x9.12} and \eqref{x9.13} we obtain \eqref{9.8}.
\end{proof}

In the next theorem we present the explicit solution formulas for the alternate Marchenko equations \eqref{6.22d} and \eqref{6.23} corresponding to the time-evolved reflectionless scattering data expressed in terms of the matrix triplets $(A,B,C^{(u,v)}),$ $(\bar{A},\bar{B},\bar{C}^{(u,v)}),$ $(A,B,C^{(p,s)}),$ and $(\bar{A},\bar{B},\bar{C}^{(p,s)}).$

\begin{theorem}
	\label{thm:theorem x9.6}
	Using as input the time-evolved reflectionless Marchenko kernels $\Omega_{n+m}^{(u,v)}$ and $\bar{\Omega}_{n+m}^{(u,v)}$ that have the form as
	in \eqref{x9.4} for the potential pair $(u,v),$ the corresponding alternate Marchenko equation \eqref{6.22d} has the explicit solution given by
	\begin{equation}\label{9.41}
	\mathscr{K}_{nm}^{(u,v)}
	=-\bar{C}^{(u,v)}\,(\bar{A})^{-n-1}
\left[I-(\bar{A})^{-1}\right]^{-1}(V_n^{(u,v)})^{-1}\,\bar{\mathcal{E}}\,(\bar{A})^{-m}\,\bar{B},
	\end{equation}
	where the $\bar{\mathscr{N}}\times \bar{\mathscr{N}}$ matrix $V_n^{(u,v)}$ is defined as
	\begin{equation*}
\begin{split}
	&V_n^{(u,v)}:=I\\
&\phantom{xx}+\bar{\mathcal{E}}\,(\bar{A})^{-n-1}
\left[I-(\bar{A})^{-1}\right]	\bar{\Upsilon}^{(u,v)}\,\mathcal{E}\,A^{2n+1}\,(I-A)^{-1}\Upsilon^{(u,v)}\,(\bar{A})^{-n-1},
\end{split}
	\end{equation*}
	with $\bar{\mathscr{N}}$ being the positive integer defined in \eqref{x9.5} and with $\Upsilon^{(u,v)}$ and $\bar{\Upsilon}^{(u,v)}$ being the matrices appearing in \eqref{x9.6} for the potential pair
	$(u,v).$ Similarly, using as input the time-evolved reflectionless Marchenko kernels $\Omega_{n+m}^{(p,s)}$ and $\bar{\Omega}_{n+m}^{(p,s)}$ that have the form as
	in \eqref{x9.4} for the potential pair $(p,s),$ the corresponding alternate Marchenko equation \eqref{6.23} has the explicit solution given by
	\begin{equation}\label{9.43}
	\bar{\mathscr{K}}_{nm}^{(p,s)}
	=-C^{(p,s)}\,A^{n-1}(I-A)^{-1}\,(\bar{V}_n^{(p,s)})^{-1}\,\mathcal{E}\,A^{m}\,B,
	\end{equation}
		where $\bar{V}_n^{(p,s)}$ is the $\mathscr{N}\times \mathscr{N}$ matrix defined as
	\begin{equation*}
	\bar{V}_n^{(p,s)}:=I+\mathcal{E}\,A^{n+1}\,(I-A)
	\Upsilon^{(p,s)}\,\bar{\mathcal{E}}\,(\bar{A})^{-2n-3}
\left[I-(\bar{A})^{-1}\right]^{-1}\bar{\Upsilon}^{(p,s)}\,A^{n-1},
	\end{equation*}
with $\mathscr{N}$ being the positive integer defined in \eqref{x9.5} and with $\Upsilon^{(p,s)}$ and $\bar{\Upsilon}^{(p,s)}$ being the matrices appearing in \eqref{x9.6} for the potential pair $(p,s).$
\end{theorem}
\begin{proof}
Using \eqref{x9.3} with the potential pair $(u,v),$ from \eqref{6.24} we obtain
\begin{equation*}
G_n^{(u,v)}=\sum_{k=n}^{\infty}C^{(u,v)}\mathcal{E}A^{k-1}B,
\end{equation*}
which is equivalent to
\begin{equation}\label{9.45}
G_n^{(u,v)}=C^{(u,v)}\,\mathcal{E}A^{n-1}\,(I-A)^{-1}\,B.
\end{equation}
Let us remark that $(I-A)^{-1}$ is well defined because $|z_j|<1$ for $1\le j \le N,$ as seen from \eqref{A1} and Theorem~\ref{thm:theorem x4.1}. In the same way, from  \eqref{6.25} and \eqref{x9.3} we get
\begin{equation*}
\bar{G}_n^{(u,v)}=\sum_{k=n}^{\infty}\bar{C}^{(u,v)}\,\bar{\mathcal{E}}\,(\bar{A})^{-k-1}\,\bar{B},
\end{equation*}
or equivalently
\begin{equation}\label{9.46}
\bar{G}_n^{(u,v)}=\bar{C}^{(u,v)}\bar{\mathcal{E}}(\bar{A})^{-n-1}\left[I-(\bar{A})^{-1}\right]^{-1}\bar{B}.
\end{equation}
Using \eqref{9.45} and \eqref{9.46} in \eqref{6.22d} and by proceeding in a similar way as in the proof of Theorem~\ref{thm:theorem x9.4}, we obtain \eqref{9.41}. The explicit solution given in \eqref{9.43} is obtained similarly by using in \eqref{6.23} the analogs of \eqref{9.45} and \eqref{9.46} for the potential pair $(p,s).$
\end{proof}

Let us remark that, using \eqref{x9.8} and \eqref{x9.11} in \eqref{Tx.8}, respectively, we obtain an explicit solution formula for the nonlinear system \eqref{x.8}, where $u_n$ and $v_n$ are expressed explicitly in terms of the matrix triplets $(A,B,C^{(u,v)})$ and $(\bar{A},\bar{B},\bar{C}^{(u,v)})$ as
\begin{equation}\label{9.52aaa}
\begin{cases}
u_n=-\bar{C}^{(u,v)}(\bar{A})^{-n}\left(U_n^{(u,v)}\right)^{-1}\,\bar{\mathcal{E}}\,(\bar{A})^{-n-3}\bar{B},\\
 \noalign{\medskip}
v_n=-C^{(u,v)}\,A^{n}\left(\bar{U}_n^{(u,v)}\right)^{-1}\,\mathcal{E}\,A^{n+1}\,B,
\end{cases}
\end{equation}
where  $\mathcal{E}$ and $\bar{\mathcal{E}}$ are  the matrices defined in \eqref{x9.2}, and $U_n^{(u,v)}$ and $\bar{U}_n^{(u,v)}$ are  the matrices appearing in \eqref{x9.12} and \eqref{x9.13}, respectively, for the potential pair $(u,v).$

Let us finally discuss explicit solutions to the nonlinear system \eqref{1.2a}. We can express any time-evolved reflectionless scattering data for the potential pair $(q,r)$ in
terms of the Marchenko kernels $\Omega_{n+m}^{(q,r)}$ and $\bar{\Omega}_{n+m}^{(q,r)}$ appearing in \eqref{x9.3}. Hence, as seen from \eqref{x9.2} and \eqref{x9.3}, we can explicitly determine the corresponding solution to \eqref{1.2a}, where $q_n$ and $r_n$ are explicitly expressed in terms of the matrix triplets $(A,B,C^{(q,r)})$ and $(\bar{A},\bar{B},\bar{C}^{(q,r)}).$  In fact, using these two matrix triplets as input in any of the inversion methods outlined in Section~\ref{sec:section8}, we are able to obtain explicit solution formulas for \eqref{1.2a}.

For example,
using these two matrix triplets on the right-hand sides of \eqref{x9.8aa}--\eqref{x9.11aa}, we first obtain the four scalar quantities
$[M_{nm}^{(q,r)}]_1,$
$[M_{nm}^{(q,r)}]_2,$
$[\bar{M}_{nm}^{(q,r)}]_1,$
$[\bar{M}_{nm}^{(q,r)}]_2,$
and use them in \eqref{Z.7} and \eqref{Z.7a}
to obtain the solution $(q_n,r_n)$ to
\eqref{1.2a} explicitly displayed
in terms of the matrix triplets
$(A,B,C^{(q,r)})$ and $(\bar{A},\bar{B},\bar{C}^{(q,r)}).$

We can obtain another explicit solution formula for
\eqref{1.2a} by expressing the right-hand sides of
\eqref{6.21} and \eqref{6.22} in terms of
the matrix triplets
$(A,B,C^{(q,r)})$ and $(\bar{A},\bar{B},\bar{C}^{(q,r)}).$
That formula is given by
\begin{equation*}
q_n=\tau_n-\tau_{n+1},\quad r_n=\bar{\tau}_{n-1}-\bar{\tau}_n,
\end{equation*}
where we have defined
\begin{equation}\label{S.2aa}
\tau_n:=\ds\frac{D_\infty^{(q,r)}}{E_\infty^{(q,r)}}\,\mathscr{K}_{nn}^{(u,v)},
\quad \bar{\tau}_n:=\ds\frac{E_\infty^{(q,r)}}{D_\infty^{(q,r)}}\,\bar{\mathscr{K}}_{nn}^{(p,s)},
\end{equation}
and the right-hand sides in \eqref{S.2aa}
are expressed in terms of the quantities
relevant to the
potential pair $(q,r)$ with the help
of
\eqref{T.4},
\eqref{T.4aaa}, \eqref{9.4}, \eqref{9.5},
\eqref{9.41},
\eqref{9.43}.
We get
\begin{equation*}
\begin{cases}
\tau_n=-\bar{C}^{(q,r)}\left[I-(\bar{A})^{-2}\right]^{-1}
\left[I-(\bar{A})^{-1}\right]^{-1}
(\bar{A})^{-n-1}
\left(V_n^{(q,r)}\right)^{-1}\bar{\mathcal{E}}\,(\bar{A})^{-n}\,\bar{B},\\
\noalign{\medskip}
\bar{\tau}_n=-C^{(q,r)}\,A^{n-1}(I-A)^{-1}\left(\bar{V}_n^{(q,r)}\right)^{-1}\mathcal{E}\,A^n\,B,
\end{cases}
\end{equation*}
where we have defined
\begin{equation*}
\begin{cases}
V_n^{(q,r)}:=I+\bar{\mathcal{E}}\,(\bar{A})^{-n-1}\left[I-(\bar{A})^{-1}\right]
\bar{\Upsilon}^{(q,r)}\,
(I-A^{-2})\,\mathcal{E}\,A^{2n+1}\\
\noalign{\medskip}
\phantom{xxxxxxxxxxxxxxx}
\times (I-A)^{-1}\,\Upsilon^{(q,r)}\left[I-(\bar{A}^{-2})\right]^{-1}
(\bar{A})^{-n-1},
\\
\noalign{\medskip}
\bar{V}_n^{(q,r)}:=I+\mathcal{E}\,A^{n+1}\,(I-A)
\Upsilon^{(q,r)}\,\bar{\mathcal{E}}\,(\bar{A})^{-2n-3}
\left[I-(\bar{A})^{-1}\right]^{-1}\bar{\Upsilon}^{(q,r)}\,A^{n-1},
\end{cases}
\end{equation*}
with $\Upsilon^{(q,r)}$ and $\bar{\Upsilon}^{(q,r)}$ denoting the matrices in \eqref{x9.6} for $(q,r).$

We can also obtain an explicit solution formula for \eqref{1.2a} by
using $q_n$ and $r_n$ given in \eqref{6.6} and
\eqref{6.6kk}, respectively,
after expressing their right-hand sides
in terms of the quantities relevant
to the potential pair $(q,r),$ and
this can be achieved
with the help of \eqref{6.1bbbb},
\eqref{T.4aaa}, \eqref{x9.8}--\eqref{x9.11},
\eqref{9.6}, and \eqref{9.7}.
In a similar way,
it is possible to obtain
an explicit solution formula by
using $q_n$ and $r_n$ given in \eqref{6.6hh} and
\eqref{6.8}, respectively,
after expressing their right-hand sides
in terms of the quantities relevant
to the potential pair $(q,r).$ Still another
solution formula for
\eqref{1.2a} is obtained via
\eqref{x.603} and \eqref{x.604},
and this is done as follows.
We first express the right-hand side of
the first line of
\eqref{9.52aaa}
in terms of
the matrix triplet
for the potential pair $(q,r),$
and hence recover
$u_n$ in terms of the quantities relevant to
$(q,r).$ In a similar way, we use
the analog of the second line of
\eqref{9.52aaa} for the potential pair
$(p,s)$ and obtain
$s_n$
in terms of the quantities relevant to
$(q,r).$
Finally, we use the resulting expressions for
$u_n$ and $s_n$
on the right-hand sides of
\eqref{x.603} and \eqref{x.604}
and obtain a solution formula
for
$q_n$ and $r_n$ as a solution
to \eqref{1.2a}.

\end{document}